\documentclass[conference]{IEEEtran}
\IEEEoverridecommandlockouts
\usepackage{graphicx}
\usepackage{textcomp}
\usepackage{xcolor}
\usepackage{color-edits}
\usepackage{algorithm}
\usepackage[noend]{algpseudocode}
\usepackage{float}
\usepackage{fixltx2e}
\usepackage{multicol}
\usepackage{caption}

\addauthor[Kartik]{kl}{blue}
\addauthor[Maciej]{mb}{blue}
\addauthor[Laura]{lm}{red}
\addauthor[Kelly]{ki}{orange}
\addauthor[Patrick]{pi}{blue}
\addauthor[Fabrizio]{fp}{blue}
\addauthor[Torsten]{htor}{blue}
\addauthor[Laura]{lmnew}{red}
\def\BibTeX{{\rm B\kern-.05em{\sc i\kern-.025em b}\kern-.08emT\kern-.1667em\lower.7ex\hbox{E}\kern-.125emX}}

\newif\iftr     % Full technical report
\newif\ifall    % Various stuff that might be useful but for now we don't want to use it
\newif\ifconf   % Submission to a conf or journal, with space contraints
\newif\ifsq     % Squeeze space?
\newif\ifnonb   % Non blind submission
\newif\iftodos

\newif\ifsqCAP
\newif\ifsqVS
\newif\ifsqEN
\newif\ifsqTIT

\sqtrue
\sqCAPtrue
\sqENtrue
\sqVStrue
\sqTITtrue

\usepackage{etex}
\usepackage{balance}
\usepackage{epstopdf}
\usepackage{placeins}

\usepackage{cite}

\usepackage{graphicx}
\usepackage{float}
\usepackage{dblfloatfix}
\usepackage{multirow}
\usepackage{rotating}
\usepackage[switch]{lineno} % Line numbers
\usepackage{makecell}
\usepackage{tabulary}
\usepackage{parcolumns}
\usepackage{tikz}
\usetikzlibrary{tikzmark}

\usepackage{xpatch}
\expandafter\xpatchcmd
\csname pgfk@/tikz/every picture/.@cmd\endcsname
{\thepage}{\arabic{page}}{}{}

\tikzstyle{comment} = [draw, fill=blue!70, text=white, text width=3cm, minimum height=1cm, rounded corners, align=left, font=\scriptsize]
\tikzstyle{background_alg} = [draw, fill=blue!20, opacity=0.4, inner sep=4pt, rounded corners=2pt]

\usetikzlibrary{shapes}
\usetikzlibrary{plotmarks}
\usetikzlibrary{calc, fit}

\usepackage{enumitem}

\usepackage{amsthm}
\usepackage{amsmath,amssymb,amsfonts}
\usepackage{mathtools,mathrsfs}

\newtheorem{theorem}{Theorem}[section]

\newtheorem{proposition}[theorem]{Proposition}
\newtheorem{lemma}[theorem]{Lemma}
\newtheorem{corollary}[theorem]{Corollary}

\newtheorem{definition}[theorem]{Definition}

\usepackage{soul}
\usepackage{fontawesome}
\usepackage{pifont}
\usepackage{textcomp}
\usepackage{booktabs}
\usepackage{url}
\usepackage{pbox}
\usepackage[normalem]{ulem}
\usepackage[10pt]{moresize}

\usepackage{cleveref}
\usepackage[utf8]{inputenc}
\crefname{section}{§}{§§}
\Crefname{section}{§}{§§}

\ifsqCAP
\usepackage[font={scriptsize}]{caption}
\usepackage[font={scriptsize}]{subcaption}
\else
\usepackage[font={footnotesize}]{caption}
\usepackage[font={footnotesize}]{subcaption}
\fi

\newcommand{\vspaceSQ}[1]{\ifsqVS\vspace{#1}\fi}
\newcommand{\enlargeSQ}[1]{\ifsqEN\enlargethispage{\baselineskip}\fi}

\ifsqTIT
\newcommand{\subparagraph}{}
\usepackage[compact]{titlesec}
\titlespacing*{\section}{0pt}{6pt}{2pt}
\titlespacing*{\subsection}{0pt}{5pt}{1pt}
\titlespacing*{\subsubsection}{0pt}{5pt}{1pt}
\fi

\usepackage{xcolor}
\definecolor{darkgrey}{RGB}{70,70,70}
\definecolor{lightgrey}{RGB}{200,200,200}
\definecolor{lyellow}{RGB}{255,255,100}
\definecolor{llyellow}{RGB}{250,250,180}
\definecolor{lgreen}{RGB}{144,238,144}

\usepackage[customcolors]{hf-tikz}
\hfsetbordercolor{white}
\hfsetfillcolor{vlgray}

\definecolor{vlgray}{rgb}{0.77 0.77 0.77}
\definecolor{ablack}{rgb}{0.2 0.2 0.2}
\definecolor{vllgray}{rgb}{0.9 0.9 0.9}
\definecolor{bblue}{rgb}{0.7 0.7 0.99}

\usepackage{colortbl}

\usepackage{inconsolata}
\usepackage{listings}

\ifsq
\else
\fi

\newcommand{\maciej}[1]{\textcolor{blue}{[Maciej: #1]}}

\newcounter{highlight}

\newcounter{Ahighlight}

\usepackage{scalerel,stackengine}
\stackMath
\newcommand\rwh[1]{%
\savestack{\tmpbox}{\stretchto{%
  \scaleto{%
        \scalerel*[\widthof{\ensuremath{#1}}]{\kern-.6pt\bigwedge\kern-.6pt}%
                  {\rule[-\textheight/2]{1ex}{\textheight}}%WIDTH-LIMITED BIG WEDGE
                              }{\textheight}% 
}{0.5ex}}%
\stackon[1pt]{#1}{\tmpbox}%
}

\usepackage[hang,flushmargin]{footmisc}

\DeclarePairedDelimiter\abs{\lvert}{\rvert}

\renewcommand{\epsilon}{\ensuremath\varepsilon}

\renewcommand{\phi}{\ensuremath{\varphi}}

\if 0

\usepackage[linesnumbered,ruled]{algorithm2e}
\usepackage{multicol}
\SetKwComment{Comm}{$\triangleright$\ }{}
\SetAlFnt{\scriptsize}
\SetAlCapFnt{\scriptsize}
\SetAlCapNameFnt{\scriptsize}
\SetKwInOut{Input}{Input}
\SetKwInOut{Output}{Output}

\makeatletter
\NewDocumentCommand{\LeftComment}{s m}{%
\Statex \IfBooleanF{#1}{\hspace*{\ALG@thistlm}}\(\triangleright\) #2}
\makeatother

\fi

\allfalse
\conftrue
\trfalse
\nonbfalse

\newcommand{\faY}[0]{\faBatteryFull}
\newcommand{\faH}[0]{\faBatteryHalf}
\newcommand{\faN}[0]{\faTimes}

\newcommand{\fly}{PolarFly }
\newcommand{\flyN}{PolarFly}
\newcommand{\ugalpf}{UGAL\textsubscript{PF}}
\newcommand{\topo}{ER  polarity graphs }
\newcommand{\ER}[1]{\ensuremath{ER_{#1}}}
\newcommand{\quadric}[1]{\ensuremath{W(#1)}}
\newcommand{\lone}[1]{\ensuremath{V_1(#1)}}
\newcommand{\ltwo}[1]{\ensuremath{V_2(#1)}}
\algnewcommand\algorithmicforeach{\textbf{for each}}
\algnewcommand{\IfThenElse}[3]{% \IfThenElse{<if>}{<then>}{<else>}
  \State \algorithmicif\ #1\ \algorithmicthen\ #2\ \algorithmicelse\ #3}

\newtheorem{property}{Property}

\algrenewtext{ForEach}[1]{\algorithmicforeach\ #1\ \algorithmicfordo}
\algrenewtext{EndForEach}{}

\algdef{SnE}[FOREACH]{ForEach}{EndForEach}[1]{\algorithmicforeach\ #1}%
 
\makeatletter % changes the catcode of @ to 11
\newcommand{\linebreakand}{%
  \end{@IEEEauthorhalign}
  \hfill\mbox{}\par
  \mbox{}\hfill\begin{@IEEEauthorhalign}
}
\makeatother % changes the catcode of @ back to 12

\begin{document}

%\title{\flyN: A Cost-Effective and Flexible Low-Diameter Topology\vspace{-1em}}
%\if 0
\if11
\title{\flyN: A Cost-Effective and Flexible Low-Diameter Topology\\

\thanks{This work was supported by Triad National Security, LLC, operator of the Los Alamos National Laboratory under Contract No.89233218CNA000001 with the U.S. Department of Energy, and by LANL’s Ultrascale Systems Research
Center at the New Mexico Consortium (Contract No. DE-FC02-06ER25750). The United States Government retains
and the publisher, by accepting this work for publication, acknowledges that the United States Government retains a nonexclusive, paid-up, irrevocable, world-wide license to publish or reproduce this work, or allow others to do so for United States Government purposes. This paper has been assigned the LANL identification number LA-UR-22-23079, Version 2.}
}

\author{\IEEEauthorblockN{Kartik Lakhotia\IEEEauthorrefmark{1}, Maciej Besta\IEEEauthorrefmark{2}, Laura Monroe\IEEEauthorrefmark{3}, Kelly Isham\IEEEauthorrefmark{4}, Patrick Iff\IEEEauthorrefmark{2}, Torsten Hoefler\IEEEauthorrefmark{2}, and Fabrizio Petrini\IEEEauthorrefmark{1}}
\IEEEauthorblockA{\IEEEauthorrefmark{1} 
Intel Labs,
Santa Clara, CA, 95054, USA\\
\{kartik.lakhotia,  fabrizio.petrini\}@intel.com}
\IEEEauthorblockA{\IEEEauthorrefmark{2} Scalable Parallel Computing Laboratory,
ETH Z\"{u}rich,
8092 Z\"{u}rich, Switzerland \\
\{maciej.besta, patrick.iff, torsten.hoefler\}@inf.ethz.ch}
\IEEEauthorblockA{\IEEEauthorrefmark{3} High Performance Computing Division, 
Los Alamos National Laboratory, 
Los Alamos, NM, 87545, USA \\
lmonroe@lanl.gov}
\IEEEauthorblockA{\IEEEauthorrefmark{4} Colgate University,
Hamilton, NY, 13346, USA  \\
kisham@colgate.edu}
}

\else
\title{\flyN: A Cost-Effective and Flexible Low-Diameter Topology\vspace{-5mm}\\

%\author{\IEEEauthorblockN{Authors}
%\IEEEauthorblockA{
%} \\
%}
}
\fi

\maketitle

\thispagestyle{plain}
\pagestyle{plain}

\begin{abstract}

%Co-packaging technology integrates
%optical I/O transceivers with compute nodes
%on the same module to enable
%high bandwidth and low-latency access to
%the network.
%Low-diameter direct topologies such 
%as SlimFly are key to designing
%scalable networks from co-packaged modules.
%However, they have not been adopted in 
%practice due to limited feasible
%router degrees and deployment challenges. 

In this paper we present PolarFly, a diameter-$2$ network topology based on 
the Erd\H os-R\'enyi family of polarity graphs from finite geometry.
%Polarity graphs and finite field theory. 
% \kldelete{This is a diameter-$2$ topology that }
This is a highly scalable low-diameter topology that asymptotically 
reaches the Moore bound on the number of nodes for a given network degree and diameter.

%asymptotically optimal diameter-$2$ topology with respect to the Moore's bound,
\fly achieves high Moore bound efficiency
even for the moderate radixes commonly seen in current and near-future routers, reaching more than $96\%$ of the theoretical peak. It also offers more feasible router degrees than the state-of-the-art solutions,
%\fpcomment{I don't think quoting slimfly in the abstract is a good idea} \lmcomment{agree, maybe just say "many more" and "state-of-the-art solutions"}
greatly adding to the selection of scalable
diameter-$2$ networks.
% Laura: maybe the explicit numbers should come later in the paper itself?
%: $96.2\%$ with degree 32, $98.4\%$ with degree 65 and $99.2\%$ with degree 128. PolarFly 
\fly enjoys many
other topological properties highly relevant in practice, such as a modular design and expandability that allow incremental 
growth in network size without rewiring
the whole network.
Our evaluation shows that
%As shown by our extensive evaluation, 
\fly outperforms %{\color{red}\sout{all}}
competitive networks in terms of scalability, cost and performance for various traffic patterns. %, while at the same time is
%resilient to link failures.
% , while being competitively resilient to link failures and being significantly easier to deploy in practical settings
%Our evaluation shows that PolarFly outperforms all competitive networks in 
%performance, cost, power consumption, and 
%scalability, while being competitively 
%resilient to link failures and being 
%significantly easier to deploy in practical
%settings.

% Low-diameter network topologies such as Slim Fly enable unprecedented scale, high performance, cost effectiveness, power efficiency, and resilience. However, they have not been adopted so far in practice due to challenges related to the practical deployment. In this work, we propose PolarFly: a low-diameter topology based on the theory of Polarity graphs, that resolves these issues while simultaneously not only maintaining all Slim Fly advantages, but even further increasing its scalability.
% %
% Specifically, as opposed to Slim Fly, PolarFly is both modular and extensible, making it possible to incrementally increase its size without rewiring the whole network. Moreover, PolarFly comes with more feasible router degrees, giving a wider selection of networks that can be constructed. Third, PolarFly uses the recent technological developments into copackaging and swizzles to reduce the complexity of wiring.
% %
%Our evaluation shows that PolarFly outperforms all competitive networks in performance, cost, power consumption, and scalability, while being competitively resilient to link failures and being significantly easier to deploy in practical settings.
%
\end{abstract}

% \vspace{0.5em}
%
% {\small\noindent\macb{[Anonymized] Code and report:}\\\url{https://www.dropbox.com/s/}}

\section{Introduction}
\label{sec:intro}

Traditional demand for scalable networks comes from government labs and research institutions to perform large scientific simulations. For example Fukagu \cite{Fugaku:Dongarra}, the largest supercomputer in the world at the time of this writing, connects $158,976$ processing nodes in a single system. The immediate forerunners in the Top500
%\footnote{https://top500.org/lists/top500/}
list~\cite{top500}, Sierra~\cite{sierra} and Summit \cite{zimmer2019evaluation}. use Infiniband configurations with, respectively, $4,474$ and $4,608$ processing nodes. Another notable example of large scale network is BlueGeneQ \cite{chen2012looking}, with $98,304$ network endpoints.
Meta
%\footnote{https://about.fb.com/news/2022/01/introducing-metas-next-gen-ai-supercomputer/} 
recently announced the AI Research SuperCluster (RSC)~\cite{meta_rsc}, which is expected to be the fastest AI supercomputer in the world in later 2022 with almost $10,000$ nodes. RSC will help Meta’s AI researchers build better AI models that can learn from trillions of examples, work across hundreds of different languages to analyze text, images and video together, and develop new augmented reality tools.  Microsoft, Google and Amazon also have expressed strong interest in simulating large AI models with analogous network scale \cite{rajbhandari2021zero,gutierrez2022thinking}. 

A common requirement of academic, governmental and industrial High-Performance Computing (HPC) data centers, is increased efficiency and reduced network cost, which is typically proportional to the number of network links. For this reason low-diameter networks, and in particular diameter 2 and 3 topologies, have seen growing interest in the scientific community over the last few years~\cite{valadarsky2016xpander,aksoy2021spectralfly,lei2020bundlefly,valadarsky2015xpander}. Low diameter networks are also considered an essential ingredient to tackle one of the major issues in data centers, tail latency~\cite{Dean:tail-latency}.

The emergence of high-radix optical IO modules is a technological amplifier of low diameter networks. These modules increase the network by squeezing many connections per unit of space or "shoreline", leading to higher system connectivity \cite{sipho-2.5d:Bergman, 3d-multichip:Bergman, wade2020teraphy,darpa-eri-2019,avicena,teramount,lightmatter,dustphotonics,celestialai,axalume,maniotis2020scaling}.
 In addition, tremendous efforts are being made to bring silicon photonic connections directly into the chip, using various integration methods \cite{sipho-2.5d:Bergman, 3d-multichip:Bergman, EMIB:Mahajan}, commonly known as co-packaging.

Co-packaged photonics not only reduces power consumption and 
improves performance, but also impacts the overall network design.
With co-packaged modules, each element of the chipset (compute, acceleration, memory and storage)  
becomes a first-class citizen with a direct low-latency interface to the network. 
To maximize application performance in this novel system design, the onus is
on network fabric to provide high-bandwidth, low-latency communication and extreme scalability. This further increases the
need for scalable low-diameter networks.

\subsection{State of Art Diameter 2 Topologies: Slim Fly}
Given the availability of high-radix routers, it is desirable to maximize the number of nodes that can be supported on a network of a given diameter. Slim Fly~\cite{besta2014slim} was the first topology analyzed in the networking community that explicitly optimized its structure towards the Moore bound~\cite{comb_wiki_degdiam_general}, an upper bound on the number of vertices in a graph with given diameter and for a given maximal degree. By fixing its diameter to $2$, Slim Fly reduces construction cost and power consumption, 
% since the number of routers and thus buffers and ports is the smallest for a given endpoint count. 
% Simultaneously, it also 
while ensuring low latency and high bandwidth. %\kicomment{I've seen SlimFly in some places -- we should check this} \lmcomment{Everywhere in the text it says "Slim Fly", but in a couple of figures, it says "SlimFly".}

However, Slim Fly has several issues with respect to practical layout and deployment. The number of feasible configurations/topological constructions is limited, and it is not competitive with commercially available products. And there are no results in the literature to address network re-configurability and expansion. This is a matter of great importance in real-world scenarios, where data centers need to increase the size of a compute center gradually over time without being forced to rewire and re-layout the whole interconnect.
%
%Finally, while Slim Fly came with a proposal on how to deploy it as a rack-level interconnect, it is not clear how to adapt it to upcoming co-packaged
%optical designs.\klcomment{How does PolarFly adapt to co-packaged optical design?}

\if 0
%
3) Recent low-diameter networks such as Slim Fly pushed the margin in terms of
reducing cost and power while keeping high performance and resilience.
We explore it as a potential topology to build scalable co-packaged networks.
However, its practical applicability is limited because:
%
\begin{itemize}
%    \item Few design points (degrees for which topologies exist)
    \item Lack of flexibility - need to rewire entire network to add few extra nodes
    \item Low injection bandwidth for systems with overprovisioned degrees (what to do with empty ports).
%    \item Complex Layout - what consequences does this have? \maciej{This is I think not worse than the ER graphs, TBH}
\end{itemize} 
%
\fi

\begin{figure}[htbp]
    \centering
    \includegraphics[width=\linewidth]{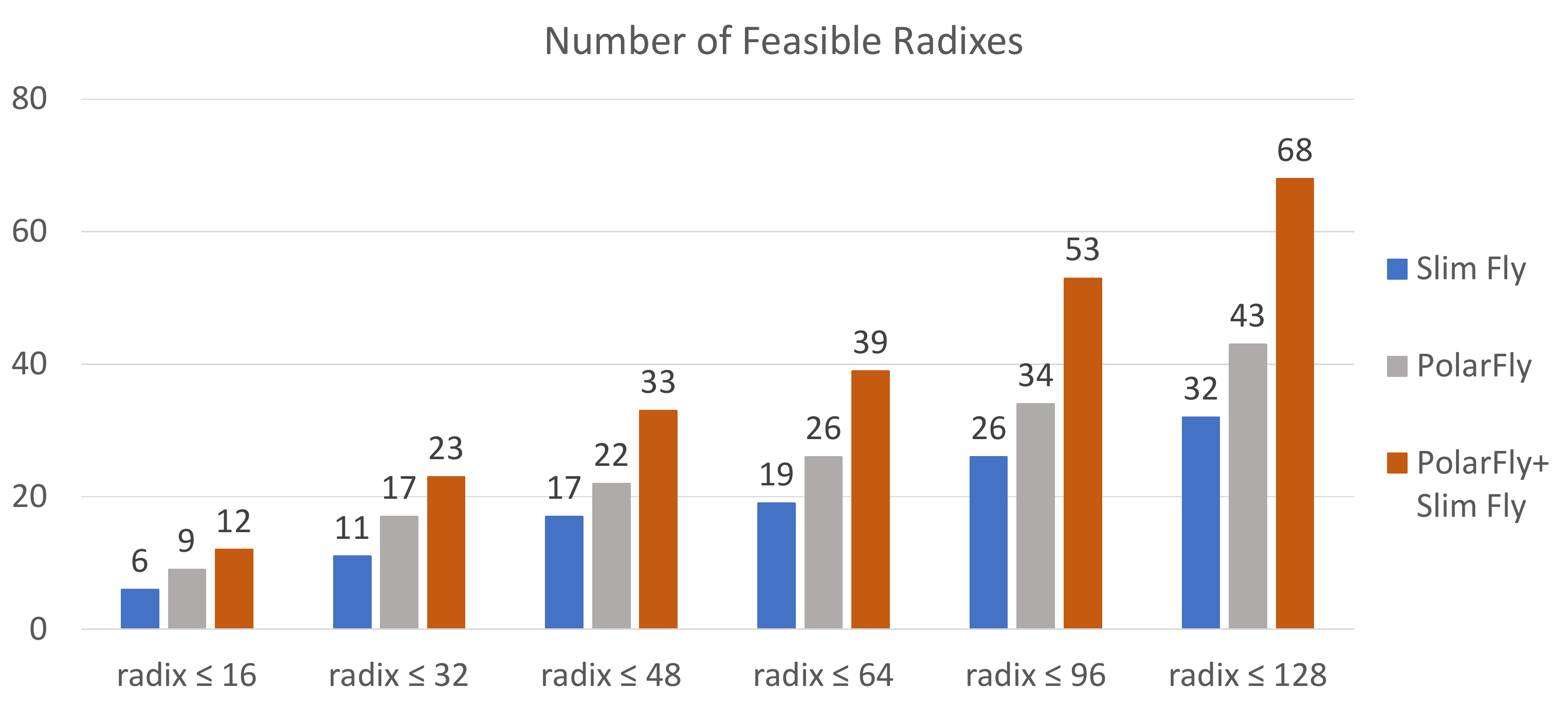}
    \caption{Design space of feasible degrees~(network radixes) for \fly and Slim Fly. Asymptotically, there are 50\% more \fly feasible degrees than what Slim Fly offers.}
    \label{fig:feasible}
\end{figure}

\subsection{Contributions}
To address the above, we present \flyN{}: a diameter-2 network 
topology that is asymptotically optimal in terms of the Moore bound. 
%\fly is based on 
%a class of
%degree-diameter 
%graphs called
%Erd\H os-R\'enyi polarity quotient graphs. %, which %overcome Slim Fly limitations. 

%\fly provides the following contributions:
\begin{itemize}
\item \fly comes with a \emph{larger
number of feasible designs}, in contrast to 
%the MMS graphs that underlie 
Slim Fly~\cite{besta2014slim}, as shown in Figure~\ref{fig:feasible}.
Asymptotically, \fly{} offers 50\% more configurations that can be constructed, without much overlap with the design space of Slim Fly.  This is especially \emph{beneficial for co-packaged systems}, where the network is integrated in the same chip with the computational engines and the network radix is fixed.
\item 
%Importantly, 
\fly can be constructed
%on nodes with a port count of \lmcomment{
from routers having radix at or near powers of $2$. In particular, \fly supports radixes $32, 48, 62$ and $128$, \emph{making the network design more practical for HW implementations}.
\item \fly has a surprisingly elegant construction that allows \emph{modularity and expansion}. Specifically, the network can be decomposed into groups of nodes, all but one of which is a fan-out of triangles, to which new node groups may be added incrementally.
\item \fly
%enjoys another advantageous topological property: it 
is asymptotically optimal in terms of the Moore bound, while Slim Fly asymptotically achieves only $8/9$ of the bound. Hence, \fly
%is scalable 
\emph{asymptotically supports as many routers as possible} for the router degree in a diameter-2 system, and offers \emph{reductions in construction costs} of up to 20\%.
\item \fly exhibits \emph{high bisection bandwidth}, empirically approaching an optimal $50$\% of edges in the cut-set as $q$ gets large. \fly also shows \emph{high resilience upon link breakage}, with the diameter experimentally staying at 4 even after $55$\% link breakage.
\end{itemize}
Our evaluation shows that \fly{} is performance- and cost-competitive with other network topologies, including Slim Fly, Dragonfly, and Fat Tree.
%Specifically, t
Thanks to its compact layout, \fly{} achieves very high saturation under random traffic with low latency, for both minimal and non-minimal adaptive routing. Moreover, under adversarial permutation patterns, \fly{} saturates between $50\%$ and $66\%$, outperforming Slim Fly and Dragonfly, and approaching the non-blocking Fat Tree.

\section{Background}
\label{sec:back}
%We start with background and notation.

\if 0

Table~\ref{tab:symbols} lists 
key symbols. 

\begin{table}[h]
\centering
\footnotesize
\begin{tabular}{rl}
\toprule
$N$&Number of endpoints in the whole network\\
% $p$&Number of endpoints attached to a router (\emph{concentration})\\
$k$&Number of channels to other routers (\emph{network radix})\\
% $k$&\emph{Router radix} ($k = k' + p$)\\
%$N_e$&Number of all channels in the network\\
% $N_r$&Number of all routers in the network\\
$D$&Network diameter\\
\bottomrule
\end{tabular}
\caption{Most important symbols used in the paper.}
\label{tab:symbols}
\end{table}

\fi

\subsection{Network Model}

We model an interconnection network as an undirected graph $G = (V,E)$; $V$ is
the set of nodes and $E$ is the set of links ($|V| = N$).
We consider only direct networks with co-packaged modules, so there is no notion of endpoints
attached to routers: each node in the network serves both as a router/switch and as a compute endpoint.
There are $N$ such nodes in total, and $k$ channels from each node to other
node (\emph{radix}).
The diameter $D$ denotes the maximum length of shortest paths between any pair of nodes.

\subsection{The Degree-Diameter Problem and Network Design}
The degree of a network is determined by current technology, and the diameter is chosen according to the system requirements. %\htorcomment{we just said we fix diameter-2?}. 
Based on this, one would like to maximize the number of nodes in such a network. 
This is the degree-diameter problem: find the maximum number $n(D,k)$ (or $N$) of vertices in a graph given maximal degree $k$ and diameter $D$. 

The degree-diameter problem is a major open problem in graph theory. For a comprehensive survey, see  \cite{miller_siran_2005}. 
Bounds exist for this problem, but few optimal graphs have been identified. Loz, P\'erez-Ros\'es and Pineda-Villavicencio give two tables \cite{comb_wiki_degdiam_general} with the largest known graphs and bounds as of 2010 for a given degree and diameter. Only a few of the graphs are known to be optimal.

%Here, the degree $d$ is arbitrary, but is determined by the technology of the day and the desired size of the system. 
\subsubsection{The Moore Bound}
The Moore bound \cite{hoffmansingleton1960} is the most general upper bound on the number of vertices $n$ for 
a graph with maximum degree $k$ and diameter $D$, and is given by
\begin{equation}\label{moore_bd}
N \le 1+k\cdot\sum_{i=0}^{D-1}(k-1)^i.
\end{equation}
\if 0
This bound is constructed via an undirected tree in $D$ levels, starting from the root at level $0$ having $k$ children. At each subsequent level, $k-1$ children are attached to each node at that level, so each internal node has degree $k$, and all nodes are reachable from the root in $D$ hops. This therefore gives the upper bound on the number of nodes possible on a graph of degree $k$ and diameter $D$.
\kicomment{Could possibly cut the above paragraph describing how the Moore bound was determined if we need space.}
\klcomment{Done}
\fi
Few graphs of any diameter and degree actually meet the Moore bound; in fact, few even come close.  Hoffman and Singleton \cite{hoffmansingleton1960}, Bannai and Ito \cite{Bannai1973OnFM}, and Damerell \cite{Damerell1973OnMG} have identified all of the graphs that meet the bound. 

The Erd\H os-R\'enyi polarity graphs were introduced by Erd\H os and R\'enyi in  \cite{erdosrenyi1962} and by Brown in \cite{brown_1966}. They have diameter $2$ and asymptotically approach the Moore bound, which is $N \le 1+k^2$
%\begin{equation}\label{moore_bd_2}
%N \le 1+k^2
%\end{equation} 
for graphs of diameter $2$. They also have properties useful for network design, which we exploit for \flyN{}.

\section{Feasibility Analysis of Candidate Topologies}\label{sec:feasibility}

There are many available topologies. However, not all are
suitable for use in a data center, especially in a co-packaged setting. In this section, we investigate representative networks and show that \fly meets the data center needs best of \looseness=-1all.

% \kicomment{Do we use these abbreviations below uniformly? I only see a couple places where we abbreviate.}
% \klcomment{I think Fat-tree needs to be fixed, although I am unsure of the purpose of abbreviations here, they are not used at all. Removing them}
We consider {Slim Fly}~\cite{besta2014slim} (a variant with $D=2$), {Dragonfly}~\cite{dally08} (the ``balanced'' variant with $D=3$), and {HyperX} (Hamming graph)~\cite{ahn2009hyperx} that generalizes {Flattened Butterflies}~\cite{kim2007flattened} with $D=2$.
We also use established three-stage {Fat Trees}~\cite{Leiserson:1985:FUN:4492.4495}.
Finally, for completeness, we also consider two Fat Tree variants, Orthogonal
Fat Trees~\cite{kathareios2015cost} and Multi-Layer Full Meshes~\cite{kathareios2015cost}.
In the following, we identify the criteria for a topology to be a suitable candidate for a data center.

\if 0 
\subsection{Analysis of Design Features}\label{sec:features}

We identify the criteria for a topology to be a feasible candidate for the copackaging setting.
Simultaneously, we analyze different networks on whether they satisfy these criteria.
The analysis is illustrated in Table~\ref{tab:feasibility}.
\fi

\textbf{Directness.} 
%
\if 0
\kldelete{First, a network must be \emph{direct}, i.e., routers and compute endpoints must form a single node, and any connections should be only between such nodes. This requirement comes from the nature of the copackaging setting, and it means that fat tree topologies are not the right candidates, because they are indirect by design, i.e., some routers have no compute endpoints assigned. \htorcomment{why can I not build switches with CPO? Needs better argument (maybe based on cost for masks/designes etc.)}}
\fi
\emph{Direct} networks can be constructed using only one type of co-packaged chip that integrates the compute, routing hardware, and communication ports in the same package.
In contrast, indirect networks such as fat trees, require design, fabrication, and deployment of
additional %"router only" 
chiplet(s) for the switches, which significantly increases their overall cost. 
% \fly can successfully be used in both network designs.

\textbf{Flexibility.} 
A \emph{flexible} network provides
many feasible configurations that could be constructed using available equipment while
delivering high performance. This means that one must be able to build networks using switches with
feasible radix.
%
\if 0
%
and (2) for a given radix, one should be able develop many networks with different sizes, while ensuring high bandwidth. The last prerequisite is again related to the budget: it is often the case that due to a limited
budget, the initial network size is smaller than the ``recommended one'', and could only
be extended to a larger size later. It would be important
that such an ``underprovisioned'' network also delivers high performance.
\klcomment{Do we address performance issues of 'underprovisioned' networks in this paper?
Let us decide on this claim after evaluating bandwidth of extensions}
% - Radix size determined by system size + technology - can be overprovisioned
%
\fi

\textbf{Low Diameter.}
Upcoming distributed shared-memory systems such as PIUMA~\cite{aananthakrishnan2020piuma}, 
and future disaggregated memory systems~\cite{guo2022clio}, 
heavily rely on low-latency remote accesses for performance scalability. 
This can only be delivered by \emph{scalable} networks with small diameter, \textbf{ideally two}, or 
networks with average path length of two.
In case of direct networks, low-diameter topologies also support higher
ingestion bandwidth. 

\textbf{Modularity.} 
In a \emph{modular} network, the nodes can be decomposed into smaller units that could be, e.g., racks, blades, or chassis.
This feature facilitates manufacturing, deployment and cabling. 
% \htorcomment{also manufacturing, see Cray Shasta}
Most of the considered networks satisfy this requirement. For example, plain Fat Trees consist of pods, while
Dragonflies have a group-based structure.

\textbf{Expandability.}
A network is \emph{expandable} if its size can be incrementally increased by adding a basic unit, such as a rack, by using empty ports in an under-provisioned network.
This need is usually related to budgetary issues -- the budget-limited purchased system is smaller than the optimal system, and may only
be extended to a larger size later. 
Incremental growth may be preferred over complete rewiring into a new topology, as the latter is much more disruptive, expensive and time-consuming.
While some of the considered networks, like Dragonfly, do enable incremental growth, it is not known
how to increase the size of the most
competitive target Slim Fly.

\if 0
%
\kldelete{A network with \emph{high reach} should have very low diameter, ideally two, or -- at least --
its average path length should be at most two. 
The reason for this condition are strict latency requirements in the copackaging setting. \htorcomment{why exactly?}
\maciej{Fabrizio, we could use some clean arguments here}
Moreover, networks with high reach support higher ingestion bandwidth \htorcomment{why?}.
In the considered networks that are direct, only Slim Fly
and HyperX satisfy this requirement.\htorcomment{this definition is too fuzzy to support this statement-a single DF group has lower diameter than SF/PF - why do we call this reach if it's defined in terms of diameter?}}
%
\fi

We now analyze the considered networks and show whether they satisfy the above criteria.
A summary of the analysis is illustrated in Table~\ref{tab:feasibility}.
All networks are at least partially modular and flexible.
Most networks have diameter two.
Only \fly{} satisfies all the criteria almost fully.

\begin{table}[ht] 
\setlength{\tabcolsep}{2.5pt}
\centering
\footnotesize
%\scriptsize
%\ssmall
%\sf
\begin{tabular}{llllll@{}}
\toprule
\textbf{Topology} & \makecell[c]{\textbf{Direct}} & \textbf{Modular} & \textbf{Expandable} & \textbf{Flexible} & \textbf{Diameter-2} \\
\midrule
Fat tree & \faN & \faY & \faY & \faY & \faN\\
Dragonfly & \faH & \faY & \faY & \faH & \faN \\
HyperX & \faH & \faY & \faY & \faH  & \faY \\
OFT & \faN & \faH & \faN & \faY & \faY \\
MLFM & \faN & \faY & \faN & \faH & \faY \\
Slim Fly & \faY & \faY & \faH & \faH & \faY \\
\midrule
\textbf{\fly} & \faY & \faY & \faH & \faY & \faY \\
\bottomrule
\end{tabular}
%\vspaceSQ{-1em}
\caption{Feasibility. ``\faY'': full support, ``\faH'': partial support, ``\faN'': no support.}
\vspaceSQ{-1em}
\label{tab:feasibility}
%\vspace{-0.5em}
\end{table}

\if 0

%\subsection{Analysis of Scalability}
%We also analyze how close different diameter-2 networks are to the Moore Bound.
%We only consider direct networks, as the only viable candidates for the co-packaging setting. The results are in Figure~\ref{fig:mb}.
%\fly{} achieves the highest scalability.

\fi

\if 0

We illustrate an example \fly{} in Figure~\ref{fig:low-diam-nets},
and compare its structure to other low-diameter direct networks.

\begin{figure}[h]
%
\centering
\includegraphics[width=0.5\textwidth]{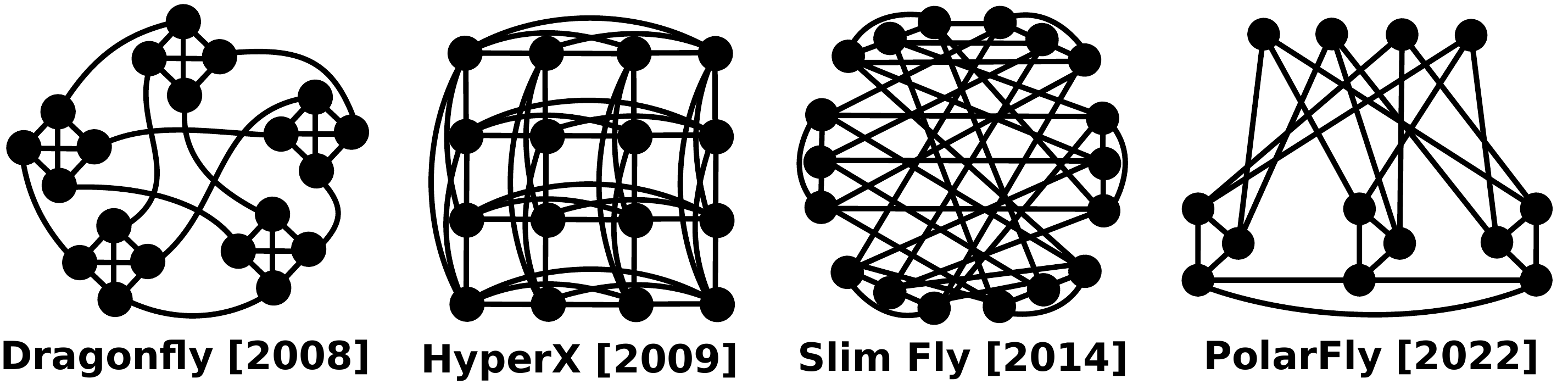}
\caption{Illustration of low-diameter direct network topologies.}
\label{fig:low-diam-nets}
%
\end{figure}

\fi
%\kldelete{\section{Erd\H os-R\'enyi Polarity Graphs}\label{sec:math_construction}}
\section{PolarFly Topology}
In this section, we discuss in detail the graph underlying the \fly layout. This mathematical description is used in the construction and for the exploitation of the graph properties.

The topology of PolarFly is an
Erd\H os-R\'enyi (ER) polarity graph, also known as a Brown graph, constructed using the relationship of points and lines in finite geometry. These were discovered independently by Erd\H os and R\'enyi \cite{erdosrenyi1962} and Brown \cite{brown_1966}. 
% , and were independently discovered by Brown \cite{brown_1966}. 
There is a great deal of mathematical structure to these graphs, and they have been studied in depth, both in the original papers \cite{erdosrenyi1962,brown_1966} and other references, e.g.  \cite{bachraty_siran_2015, miller2005, parsons_1976}.

ER graphs have several useful features for network design: 
\begin{itemize}
    \item \textbf{Low diameter.} They have diameter 2, giving a short path between any two nodes.
    \item \textbf{Scalability.}
    At the same time, they asymptotically reach the Moore bound, surpassing the scalability of all other diameter-2
    topologies. We compare to other diameter-$2$ topologies that are direct, as needed for co-packaging, and show the comparison in Figure \ref{fig:mb}.
    \item \textbf{Flexibility.} They cover a wide range of degrees,
    having degree $k=q+1$ for every prime
    power $q$. While not completely general, they meet or come very close to the radixes of many current and near-term high-radix routers. For example, for $q=31, 47, 61$ and $127$, $ER_q$ may be applied to systems with routers of radix 32, 48, 64 and 128, with all router ports used at radix 32, 48 and 128. 
     %\item \lmcomment{triangles and elegance of design}
    
    %\lmedit{\item Do we want to mention triangles at all, at least in terms of how they contribute to the simplicity of the layout?}\klcomment{As Prof. Torsten said, "simplicity" is a slightly vague term, unless we come up with a clear measure of it. I would avoid.}\lmcomment{Perhaps I misunderstood -- I thought he was saying that the layout was unclear and apparently difficult -- without a picture.}
\end{itemize}
\begin{figure}[ht]
\begin{centering}
\includegraphics[width=0.9\columnwidth]{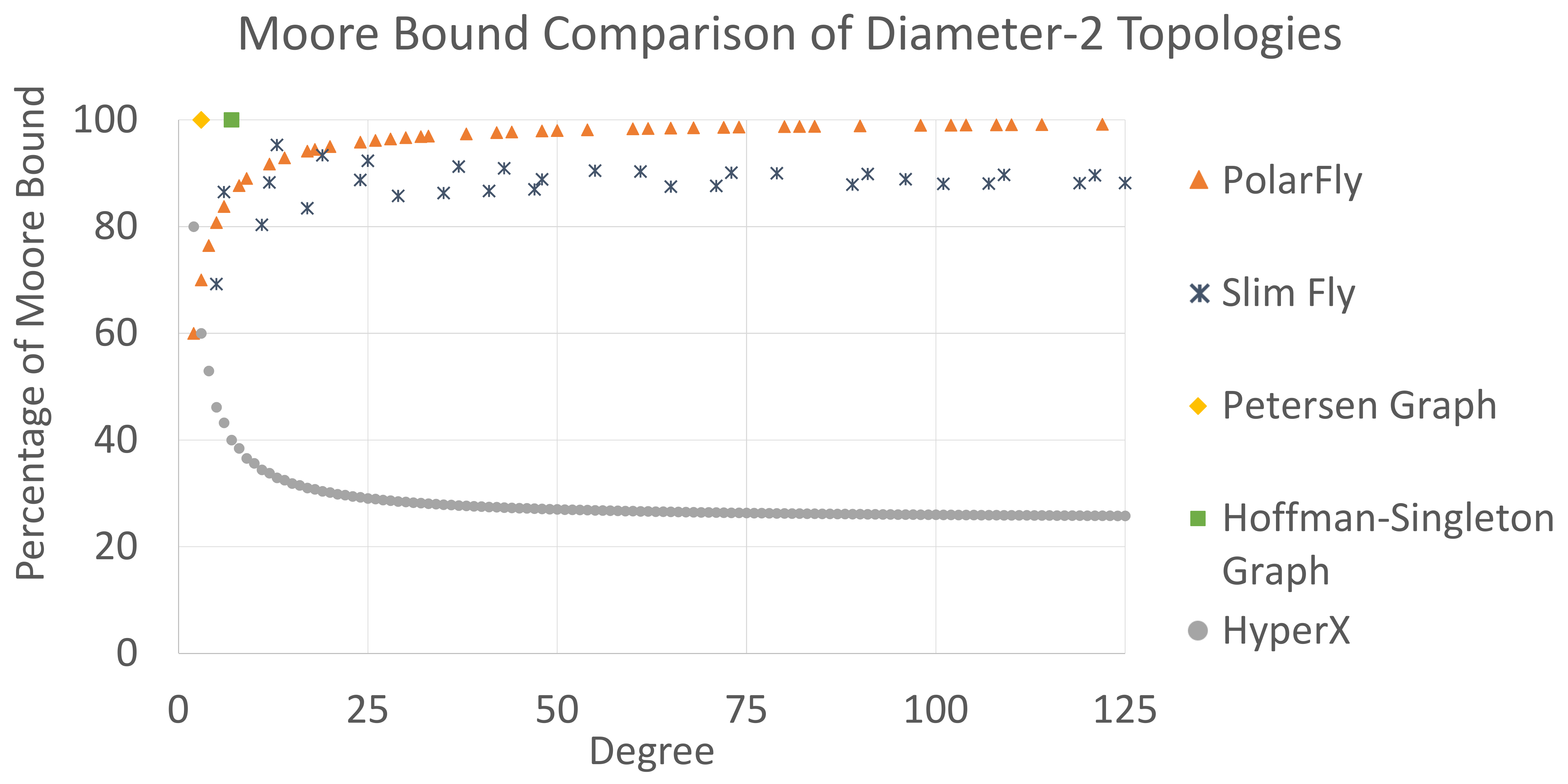}
\caption{Scalability of direct diameter-2 topologies in terms of 
the optimal Moore bound. Petersen~(10 vertices) and Hoffman-Singleton~(50 vertices) 
are the only known graphs to achieve the Moore Bound.}
\label{fig:mb}
\end{centering}
\end{figure}

%Erd\H os-R\'enyi graphs are defined in terms of prime powers $q$.  In this paper, we discuss only odd $q$, and use graph properties specific to odd $q$ in development of the algorithm presented here. We note that the graph construction is the same for even $q$, but the layout in Section X will be different. We emphasize that it is possible to 
%, but that is out of scope for this paper.
%\klcomment{This sentence makes it look like even $q$ is a different topology. Do we need to mention it here? Construction of the graphs is same for even or odd $q$, layout is different. We can mention it then.} \kicomment{Makes sense to me - I commented this paragraph out for now} \lmcomment{I do think we should say somewhere why we are only looking at odd $q$ in this paper. I think up front in the layout section would work. Maybe a sentence before property \ref{prop:er}, noting that it applies to odd $q$ and we are going to use that.}

\subsection{Some Background on Finite Fields and Their Arithmetic}\label{sec:ff}
The construction of ER polarity graphs is based upon the arithmetic operations of finite fields. 
A field is a set having addition and multiplication, where every element has an additive inverse, and every non-zero element has a multiplicative inverse. The construction of ER polarity graphs depends especially upon the existence of multiplicative inverses. 

The set of integers modulo a prime $p$ is an example of a finite field. 
Finite fields $\mathbb{F}_q$ of order $q$ exist for all prime powers $q$, and for no other integers. 
Finite fields are 
also called Galois extension fields. They are 
fundamental to many areas in mathematics and computer science, and are discussed in detail in \cite{mceliece_1987,lidl_niederreiter_1994,gallian_2020} and in many other references. 

It is important to note that addition and multiplication operations in $\mathbb{F}_q$ are quite different from those in $\mathbb{R}$:
\begin{itemize}
\item If $q=p$ is a prime, then addition and multiplication are just modular arithmetic over $p$.
\item If $q=p^m$ is a prime power, with $m>1$, then
%the arithmetic 
addition and multiplication in $\mathbb{F}_q$ are 
%not simply modular over $q$, but instead are 
derived from modular arithmetic over an irreducible degree-$m$ polynomial over $\mathbb{F}_p$. For further details on arithmetic in the finite fields $\mathbb{F}_q$, with $q$ not prime, see 
any of the references listed above.
%\cite{mceliece_1987}. 
\end{itemize}

We will use primes $q=p$ for the examples in this paper for simplicity, so the addition and multiplication in the dot products is just modular arithmetic over $p$. The same graph construction will hold for any prime power $q=p^m$, using the associated arithmetic of $\mathbb{F}_q$ for the dot products. 

\subsection{Geometric Intuition}\label{sec:geom_int}
%\subsubsection{Geometric Construction}
%ER graphs are based on the geometric concept of perpendicularity, or orthogonality. The \emph{dot product} is a convenient way of expressing this; two length-$n$ vectors $v$ and $w$ are orthogonal when $v\cdot w = \sum_{i=0}^n v_i w_i = 0$. 

% In general, graphs may be understood to express relationships between objects. The relationship that 
Erd\H os-R\'enyi polarity graphs express the orthogonality (or perpendicularity) between vectors, or equivalently, lines passing through the origin.  The \emph{dot product} is a convenient way of expressing this; two length-$n$ vectors $v$ and $w$ are orthogonal when $v\cdot w = \sum_{i=1}^n v_i w_i = 0$. 

Note that multiples of a vector retain the same orthogonality relationships as the original vector. So for our purposes, we may consider all multiples of a vector to be the same, and simply choose one as representative of all its multiples.

ER polarity graphs are defined by vertices that represent length-$3$ vectors over a field, and edges that exist between two vertices if the vectors they represent are orthogonal. This graph has diameter $2$.

As an example, consider ordinary Euclidean $3$-dimensional space. The existence of $2$-hop paths between any two vectors in the corresponding graph depends upon this fact: any pair of (non-multiple) vectors has a vector to which both are orthogonal: their cross-product. The $2$-hop path linking them passes through the orthogonal vector. This is shown in Figure~\ref{fig:cp}.
\begin{figure}[ht]
\begin{centering}
\includegraphics[width=0.45\columnwidth]{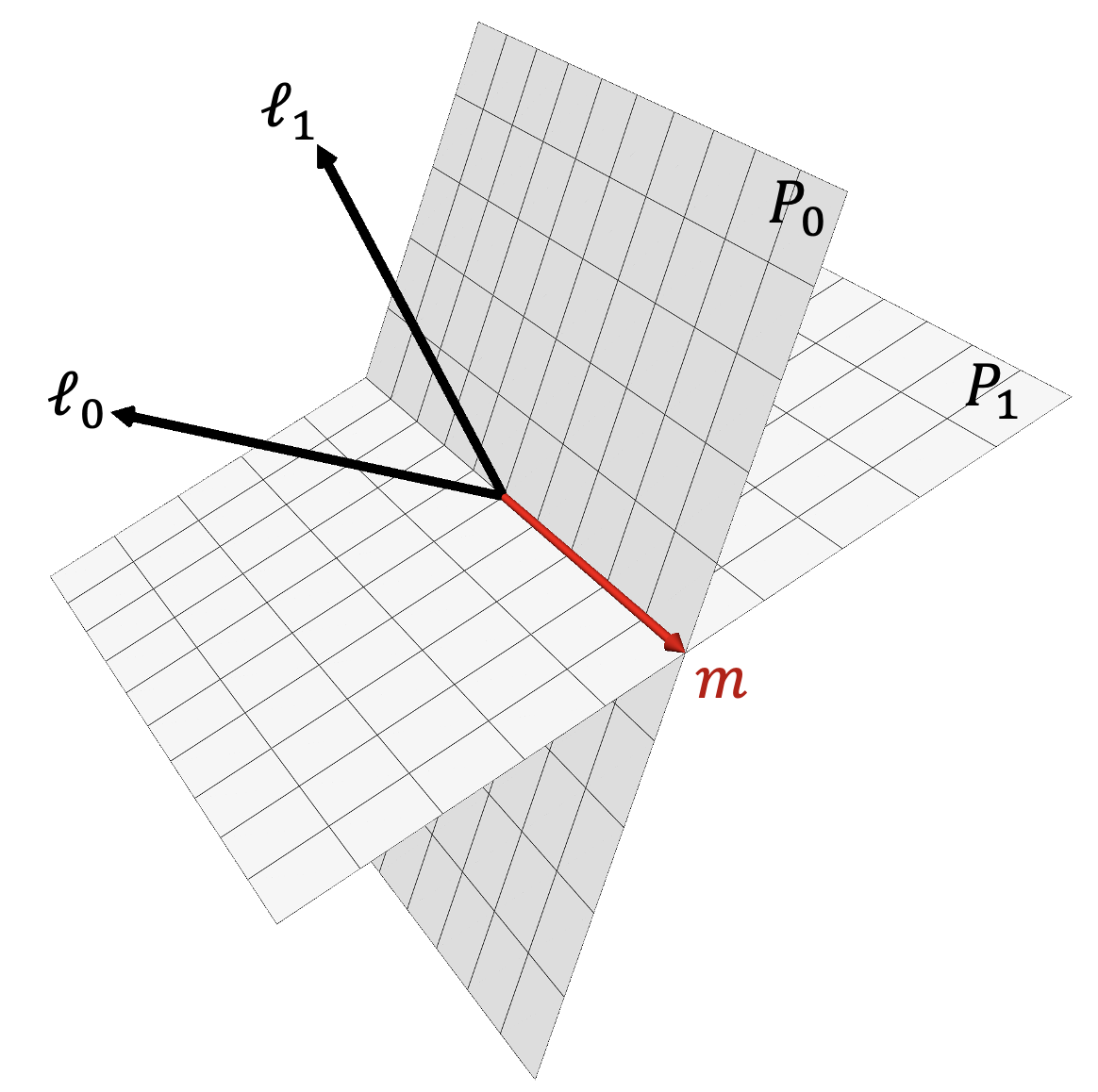}
\caption{Let $\ell_0$ and $\ell_1$ be arbitrary lines in Euclidean $3$-space passing through the origin. The line $\ell_0$ is perpendicular to the plane $P_0$, and the line $\ell_1$ is perpendicular to the plane $P_1$. The two planes intersect in a line $m$ which is perpendicular to both $\ell_0$ and $\ell_1$. A graph including $\ell_0$, $\ell_1$ and $m$ as vertices has edges ($\ell_0$, $m$) and ($m$, $\ell_1$), so there is a $2$-hop path from  $\ell_0$ to $\ell_1$ passing through $m$. This construction may be generalized to $\mathbb{F}_3^3$, using the dot product to represent perpendicularity.}
\label{fig:cp}
\end{centering}
\end{figure}

 Euclidean $3$-space obviously has infinitely many lines and planes passing through the origin. However, a similar construction may be used to obtain a finite space of dimension $3$ over the finite field $\mathbb{F}_q$.
 %, the Galois field of order a prime power $q$. 
%A field is a set having addition and multiplication, where every element has an additive inverse, and every non-zero element has a multiplicative inverse. The multiplicative inverse is needed so that each line may be represented by exactly one vector, having $1$ as its first non-$0$ entry. Finite fields $\mathbb{F}_q$ of order $q$ exist for prime powers $q$, and only for those $q$. 
%These fields are called Galois extension fields. Galois fields are fundamental to many areas in mathematics and computer science, and are discussed in great detail in many references, such as \cite{g1} and \cite{g2}. \lmcomment{Tried to rewrite this so that just enough is said to be intuitively clear, and no more.}
%This construction gives a geometric space with $q^2+q+1$ lines, since any non-zero multiple of a point on the line is on the same line. This is why the multiplicative inverses of a field are needed. 
The geometric relationships are similar to those discussed above in the Euclidean case. Orthogonality is again expressed by the dot product, using the addition and multiplication from $\mathbb{F}_q$. So the construction over $\mathbb{F}_q$ also gives rise to a graph of diameter $2$, but this time the graph is finite.
 
%It is important to note that the addition and multiplication operations in $\mathbb{F}_q$ are quite different from addition and multiplication in $\mathbb{R}$. If $q=p$ is a prime, then addition and multiplication are just modular arithmetic over $p$. If $q=p^k$ is a prime power, with $p$ a prime and $k>1$, then the Galois field $\mathbb{F}_q$ has operations derived from modular arithmetic over an irreducible degree-$k$ polynomial over the field $\mathbb{F}_p$, as discussed in \cite{g1} and \cite{g2}. 
 
Because $\mathbb{F}_q$ is finite with modular arithmetic, some non-zero vectors in $\mathbb{F}_q^3$ have the interesting property that they are orthogonal to themselves, which never happens in Euclidean space. For example, consider $\mathbb{F}_3^3$, where the arithmetic operations are modular addition and multiplication$\mod 3$. The vector $[1,1,1]$ is self-orthogonal, since $[1,1,1]\cdot [1,1,1]=1+1+1=0 \mod3$. 
\subsection{Construction With Dot Products over $\mathbb{F}_q$}
ER polarity graphs are easily constructed using the set of non-zero left-normalized vectors $[x,y,z] \in \mathbb{F}_q^3$ as vertices. These are vectors in which the first non-zero entry is $1$.% We use left-normalized vectors as representatives of their multiplicative classes.   

For example, in $\mathbb{F}_3^3$, $[1,0,2]$ and $[0,1,0]$ are vectors in the set under consideration, but $[0,2,1]$ would not be, since its first non-zero entry is $2$. Instead, $[0,2,1]$ is multiplied by the multiplicative inverse of $2$ (mod $3$), giving its left-normalized representation: $2\cdot [0,2,1] = [0,1,2]$. The existence of multiplicative inverses in the field $\mathbb{F}_q$ assures us that each non-zero vector can be represented as a left-normalized vector.

\ER{q} is then constructed as follows:

\begin{itemize}
\item The vertices are the left-normalized vectors in $\mathbb{F}_q^3$. 
\item The edges are pairs $(v,w)$ of vertices that are orthogonal to each other, as per the dot product;
in other words, using the addition and multiplication of $\mathbb{F}_q$, the dot product of $v$ and $w$ is 0.
\end{itemize}
Self-orthogonal vertices are distinguished from the others and are called 
\textit{quadric}. All other vertices are called \textit{non-quadrics}. Quadrics may be considered to have a self-loop, and play a special role in the construction of \flyN.

We show an example of this construction in Figure \ref{fig:dpc}. There is a great deal of structure in the $ER_q$ graph, some of which can be seen there and in the graph layout in Figure~\ref{fig:clustercake_design}. We exploit this in the construction of an efficient network.
\begin{figure}[ht]
\begin{centering}
\includegraphics[width=0.45\columnwidth]{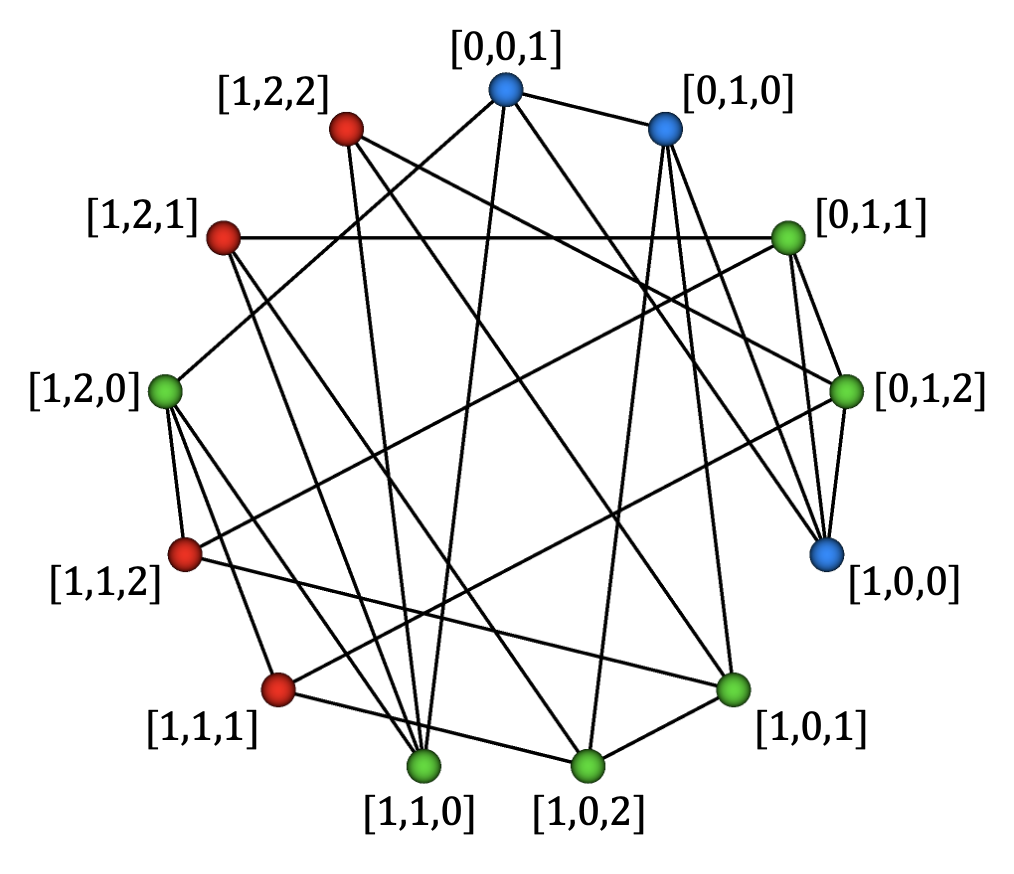}
\caption{The dot-product construction of $ER_3$. The left-normalized vectors of $\mathbb{F}_3^3$ are the vertices, arranged lexicographically, and clockwise starting at the top. The base field is $\mathbb{F}_3$. $3$ is a prime rather than a prime power, so the operations are addition and multiplication mod $3$. Edges exist between vertices $v$ and $w$ when the dot product $v \cdot w$ is $0$, using arithmetic from $\mathbb{F}_3$. 
For example, the vertex $[1,1,1]$ is adjacent to $[0,1,2]$, since the dot product $([1,1,1]\cdot [0,1,2]) = 0+1+2 \equiv 0 \mod 3$. The self-adjacent quadrics ($W$) for which the dot products $w\cdot w$ are $0$ are red. Vertices adjacent to quadrics ($V_1$) are green, and vertices not adjacent to quadrics ($V_2$) are blue. }
\label{fig:dpc}
\end{centering}
\end{figure}

\subsection{Minimal paths, intermediate points and routing}\label{subsec:intermediate}
	Since an ER graph has diameter 2, the minimal path between two vertices is either 1 or 2 hops. 
	A path is of length 1 if and only if the vertices are orthogonal (so their dot product is $0)$.

	There is only one minimal path of length $1$ or $2$ between vertices. For reasons of efficiency, table-based routing is the best method for finding paths. However, the unique intermediate vertex on a $2$-hop path may also be found by solving a pair of linear equations representing the dot-product construction.

	Intuitively, the two vertices are represented by two distinct lines, which are orthogonal to a single line (or single vertex in projective space), as seen in Figure~\ref{fig:cp}. %We are in projective space, so there is a unique representative of the line $m$ orthogonal to the lines $\ell_0$ and $\ell_1$ in the figure. 
	So the problem of finding the intermediate vertex in a minimal 2-hop path reduces to finding that unique orthogonal line.
	
	 The intermediate vertex $v$ between vertices $s$ and $d$ is orthogonal to both $s$ and $d$, so $s \cdot v = d \cdot v = 0$. Thus $v$ is found by solving this system of equations, via an augmented \looseness=-1matrix:
	$$
	\left[\begin{array}{ccc|c}  
		s_0 & s_1 & s_2 & 0\\  
		d_0 & d_1 & d_2 & 0  
	\end{array}\right]
	$$
	Since $s$ and $d$ are not multiples of each other, and since all vectors in $ER_q$ must be left-normalized, there will be a unique solution $v$ to this system of equations,
	which may be found by Gaussian elimination or otherwise.

For example, in $ER_3$, the vectors $(0,0,1)$ and $(1,2,2)$ are not orthogonal,
(since their dot product is not $0$)
so the minimal path has length $2$. The left-normalized solution in $\mathbb{F}_3^3$ to the resulting augmented matrix 
	$$
	\left[\begin{array}{ccc|c}  
		0 & 0 & 1 & 0\\  
		1 & 2 & 2 & 0  
	\end{array}\right]
	$$
is $(1,1,0)$. Figure \ref{fig:dpc} confirms that  $(1,1,0)$ is indeed intermediate on the unique length-$2$ path between $(0,0,1)$ and $(1,2,2)$.

Another, perhaps simpler way to obtain the
	vertex $v$ orthogonal to both $s$ and $d$ is to compute the cross-product of $s$ and $d$, as discussed in~\cite{parsons_1976}:
	\begin{equation}\label{eq:cross_product}
	s\times d = \left(s_2d_3-d_3s_2,\ s_3d_1 - d_3s_1,\ s_1d_2 - s_2d_1\right)
	\end{equation}
	Since multiples of a vector represent the same vertex, the coordinates obtained from cross-product can be left-normalized to obtain $v$. For example, in $ER_3$, the intermediate vertex between the vectors $(0,0,1)$ and $(1,2,2)$ is given by:
	$$
	(-2, 1, 0) = (1, 1, 0)
	$$
	as vectors modulo 3.

	Because vectors in $ER_q$ are left-normalized, this method is quite efficient, in the worst case needing only two multiplies and three adds in $\mathbb{F}_q$ to compute the cross-product, then at most another two multiplies for the left-normalization.
	
\subsection{Formal Construction}\label{sec:formal}
%\htorcomment{everything until here is nice and really cute! We need to realize that we have probably lost 60\% of the readers of the architecture track at SC by now} \lmcomment{yeah, it's rough. I've come across this before with the math/CS crossover papers I've done. Nothing to be done about it here, this is such a math-based construction. Except explain the math as clearly as possible ... Maybe prep them for it with an up-front remark in the intro that this is math-inspired.}
\subsubsection{A Bipartite Graph  From Finite Geometry}\label{sec:bipartite}
The projective plane $\textrm{PG}(2,q)$ is a geometric structure that arises from projecting lines and planes in three-dimensional space over $\mathbb{F}_q$ to points and lines respectively. As discussed above, projective points can be thought of as the left-normalized vectors in $\mathbb{F}_q^3$. In more formal language, each point $[a]$ in $\textrm{PG}(2,q)$ is an equivalence class of 3-tuples where $(a_1 , a_2 , a_3) \sim (b_1, b_2, b_3)$ if and only if these tuples are multiples of each other. Projective lines $(b_1: b_2:b_3)$ contain all points $[x]$ so that $b_1x_1 + b_2x_2+ b_3x_3=0$ in $\mathbb{F}_q$.  By counting such vectors, we see there are $q^2+q+1$ points in $\textrm{PG}(2,q)$. Dually, there are $q^2+q+1$ lines in $\textrm{PG}(2,q)$ as we will see. 

We will construct $ER_q$ graphs for any prime power $q$ using properties of points and lines in $\textrm{PG}(2,q)$.  To begin, we first build a bipartite graph $B(q)$. The vertex set of $B(q)$ is $U \cup V$  where $U$ is the set of points in $\textrm{PG}(2,q)$ and $V$ is the set of lines in $\textrm{PG}(2,q).$ There is an edge between $v \in U$ and $w\in V$ if and only if the point $v$ lies on the line $w$. 
In $\textrm{PG}(2,q)$, each point lies on $q+1$ lines, and each line contains $q+1$ points. The graph $B(q)$ has $2(q^2+q+1)$ vertices and degree $q+1.$ The graph $B(q)$ has diameter 3, which will be reduced to diameter 2 by a polarity construction in the next section.

\subsubsection{Decreasing the Diameter Using a Polarity Map}\label{sec:polarity}
For each point $[a]$ in $\textrm{PG}(2,q)$, the \emph{dual} $[a]^\perp$ is the line in $\textrm{PG}(2,q)$ which contains all points $(x_1, x_2, x_3)$ so that $a_1 x_1 + a_2 x_2 + a_3 x_3 =0$. Since the dual map is a bijection, this shows that $\textrm{PG}(2,q)$ contains $q^2+q+1$ lines. The dual of a line can be defined symmetrically. Notice that $([a]^\perp)^\perp =[a]$. Clearly, 
$[x]$ lies on line $[a]^\perp$ if and only if $[a]$ lies on
$[x]^\perp$. Such a bijection is also known as \textit{polarity}.

In order to decrease the diameter of $B(q)$, we use this polarity: take $B(q)$ and glue the vertices $v \in U$ and $w \in V$ together if and only if $[w] = [v]^\perp$. That is, combine the point $[v] \in U$ and the line $[w]\in V$ together if they are duals of each other. Define \ER{q} to be the graph formed by applying this gluing process to $B(q)$. \ER{q} has $q^2+q+1$ vertices since we glued pairs of vertices in $B(q)$ together. The degree is still $q+1$, since every line passing through point $[v]$ is
glued to a point on $[v]^\perp$, but now the diameter is reduced to 2.

This construction is quite general: if a polarity map exists on a bipartite graph with $N=2n$ vertices, maximum degree $k$, and diameter $D$, it can be used to construct another graph with $n$ vertices, maximum degree $k$, and diameter $D-1$. 

\subsubsection{Quadric Vertices}
In a finite geometry, the point $[a]$ may lie on its own dual, the line $[a]^\perp$. When this occurs, i.e., when $a_1^2 + a_2^2 + a_3^2=0$, the vertex $[a]$ is called a \emph{quadric vertex}. In \ER{q}, there is a loop at vertex $[a]$ since $[a]$ and $[a]^\perp$ are glued together. These quadric vertices are the same as the self-orthogonal vectors discussed at the end of Section \ref{sec:geom_int}, and will be discussed further in Section \ref{sec:layout_property} in terms of the layout of the network.
%\subsubsection{Generalized Polygons/Triangles}
%\lmcomment{We might want to limit this section to generalized triangles, with attention to how they can impact the form of the network. This also speaks to odd $q$ since there are no triangles for even $q$.}
%\kicomment{The projective plane $\textrm{PG}(2,q)$ is a generalized triangle, but the ER graphs are not - what do we need to say about generalized triangles? Is it sufficient just to explain that a triangle is a 3-cycle?}
%\klcomment{I think we can just skip this subsection.}
%\subsection{Abas Graph Construction}
%\klcomment{I think we can skip Abas in this paper} \lmcomment{so just cut this section now?}
%\begin{itemize}
%    \item Exist for all degrees
%    \item vertex transitive
%\end{itemize}
%\subsubsection{Semidirect Product}
%\subsubsection{Intuition on structure, Two sets of vertices}
%\change{\subsection{Implications of ER graphs for networking}
%	ER graphs have several advantages for networking: low diameter, high degree, and others, which we discuss here. Layout and expansion are discussed at length in further sections.
\subsection{Structural Properties of \topo}\label{sec:layout_property}
We make heavy use of the structure of ER graphs in the design of the network discussed in this paper. 

ER graphs \ER{q} have $N=q^2+q+1$ vertices, degree $k=q+1$, and 
diameter $D=2$. The vertex set of \ER{q} can be divided into three
disjoint subsets~\cite{parsons_1976}:
\begin{itemize}
\setlength\itemsep{.35em}
    \item \quadric{q}$\rightarrow$ set of $q+1$ quadric vertices. 
    % \quadric{q} is an independent set of 
    % $q+1$ vertices.
    
    \item \lone{q}$\rightarrow$ set of $\frac{q(q+1)}{2}$ vertices 
    adjacent to \quadric{q}. 
    % \lone{q} contains $\frac{q(q+1)}{2}$ vertices
    
    \item \ltwo{q}$\rightarrow$ set of $\frac{q(q-1)}{2}$ vertices not adjacent to 
    \quadric{q}.
\end{itemize}
The following properties used in the construction and analysis of the network were presented by Bachrat\'y and \v{S}ir{\'a}\v{n} in \cite{bachraty_siran_2015}. 

\begin{property}\cite{bachraty_siran_2015}\label{prop:er}
For every odd prime power $q$, \ER{q} has the following properties:
    \begin{enumerate}
\setlength\itemsep{.35em}
        \item 
        %\quadric{q} is an independent vertex set i.e.
        No two vertices in 
        \quadric{q} are directly 
        connected.
        Every vertex in \quadric{q} is adjacent to exactly $q$ vertices in \lone{q}.\label{prop:quad_adj}

        \item Every vertex in \lone{q} is adjacent 
        to exactly $2$ vertices in \quadric{q}, and 
        $\frac{q-1}{2}$ vertices each in \lone{q} 
        and \ltwo{q}.\label{prop:lone_adj}
        
        \item Every vertex in \ltwo{q} is adjacent 
        to exactly $\frac{q+1}{2}$ vertices each in 
        \lone{q} and \ltwo{q}.
        
        \item There is exactly one path of length two
        between every vertex pair (considering the  self-loop of the self-adjacent quadrics as an edge).\label{prop:2path}
        %\footnote{Every quadric vertex $v$ is incident with a self loop which is a part of the 2-hop paths between $v$ and its neighbors.}.
        \item As a corollary, the edges incident with quadric vertices
        do not participate in any triangle. 
        Any edge incident with two non-quadric 
        vertices participates in exactly
        one triangle.
        \label{prop:path}
    \end{enumerate}
\end{property}

%\kldelete{\section{Network Design with Erd\H os-R\'enyi Graphs}
%Erd\H os-R\'enyi graphs have useful features for network design\klcomment{I feel that description of ER graph properties is getting repeated}: 
%\begin{itemize}
%    \item They have diameter 2, giving a short path between any two nodes.
%    \item They have an abundance of triangles with good connectivity between them. \kledit{These triangles form the basis of PolarFly layout,} as shown in the example in Figure~\ref{fig:clustercake_design}.
%    \kldelete{, and is useful in the expansion of Polarfly, which permits incremental expansion of an existing network \klcomment{we did not use triangles for expanding}.}
%    \item Although they do not have an intuitively obvious hierarchical structure, the connectivities of the quadrics can be used to generate a hierarchy which we use for Polarfly.
%\end{itemize}
%\subsection{Layout and Modularity}\label{sec:layout}}
\section{PolarFly Layout}\label{sec:layout}
Network layouts need a modular topology 
decomposable into smaller units
for easy and cost-effective
deployment.
Since ER graphs are derived from polarity 
quotient graphs of finite projective planes, 
such structures are not trivially
available, in contrast to topologies derived from multiple generating sets, such as Slim Fly~\cite{besta2014slim} 
or Bundlefly~\cite{bundlefly_2020}, in which 
modular units can be readily obtained
from individual generators.

%Erd\H os-R\'enyi graphs are derived from polarity quotient graphs of finite projective planes. This means that the uniform hierarchical structures required for a layout are not trivially available from the construction. This differs from topologies derived from multiple generators, such as Slimfly~\cite{besta2014slim} or Bundlefly~\cite{bundlefly_2020}, which are easily partitioned into multiple racks.
Instead, we use the connectivity of quadrics to other vertices
to obtain a \textbf{modular} and \textbf{generalized}
layout for \flyN. Property~\ref{prop:er} from Section~\ref{sec:layout_property} tells us that every quadric $v \in \quadric{q}$ is connected to exactly $q$ vertices in \lone{q}. We use these $q$ vertices to construct $q$ clusters, plus the quadrics themselves as the $(q+1)^{st}$ cluster. 
 
 More formally, Algorithm~\ref{alg:racks} assigns the vertices in \ER{q}
 into $q+1$ clusters, which corresponds to assigning nodes to racks in \fly.
%  routers in \fly into $q+1$ racks that correspond to vertex \textit{clusters} in \ER{q}. 
We use the terms \textit{racks} and \textit{clusters} interchangeably. 
 
 %, which are naturally partitioned into hierarchical subsets that make up the blades or racks in the system.
% naturally exhibit a modular structure
% that can be easily partitioned into multiple racks. 

For brevity, we only discuss \ER{q} for odd $q$ (even 
radix) as even prime powers $q=2^i$ are sparse in the set
of all prime powers. The layout for even $q$ is similarly
modular, and is derived using an analogue to Property~\ref{prop:er}
for even $q$ \cite{bachraty_siran_2015}. 
%Based on the properties of the $ER_q$ graphs discussed in Section \ref{sec:layout_property}, we propose 
\begin{algorithm}[ht]
%    \caption{Arranging \fly routers into racks}
    \caption{\fly layout}
    \label{alg:racks}
    \begin{algorithmic}[1]
        \Statex{\ER{q}$\leftarrow$ ER Graph of max degree $q+1$}
        \State{Initialize empty clusters (racks) $C_0, C_1 \ldots, C_q$}
        \State{Add all quadrics \quadric{q} to $C_0$}
        \State{Select an arbitrary quadric $v\in W(q)$}
        \ForEach{vertex $u$ adjacent to $v$}
            \State{Add $u$ to an empty cluster $C_i$}
            \State{Add all non-quadric neighbors of $u$ to $C_i$}\label{line:add_neigh}
        \EndForEach
    \end{algorithmic}
\end{algorithm}
%\begin{figure*}[t]
%
%\includegraphics[width=1\textwidth]{ER-1.pdf}
%\caption{Layout and expansion methods. }
%\label{fig:expand}
%\end{figure*}

%\klcomment{We can choose to keep or remove the proofs for the following claims. Adding mostly proof sketches for now. The Lemma environment is creating strange indexing with section numbers.} \lmcomment{Can we put any proofs in the appendix?} \klcomment{We can't. The Artifact appendix is auto generated from the form that we fill on SC portal. We are not allowed to add appendices to the submission before acceptance.}
%\lmcomment{as per Kartik, no proofs in the appendix. So proofs should either be short or cut entirely. Or sketches! Good keywords to use: "Clearly..." "Proof left for the reader..." "It is obvious that ..."}\htorcomment{hahahaha}
Figure~\ref{fig:layout} is a diagram of the layout, and Figure~\ref{fig:clustercake_design} shows an example of the $ER_q$ graph structure supporting this layout. 

\begin{proposition}\label{claim:assign1}
Algorithm~\ref{alg:racks} adds every vertex in \ER{q} to exactly one cluster.
\end{proposition}
\begin{proof}
Each vertex is
at most $2$ hops away from the quadric $v$
selected in Algorithm~\ref{alg:racks} line $3$. The algorithm assigns clusters 
to $W(q)$ and all non-quadrics at shortest distance $\leq2$ from $v$.
Thus, every vertex is added to at least one cluster.
By Property~\ref{prop:er}.\ref{prop:path},
the vertices adjacent to $v$ 
are independent and have no common neighbor
aside from $v$. Hence, non-quadric vertices are added to at most one cluster.
The quadrics are added to exactly one cluster, $C_0$.
\end{proof}
%The proof will be complete if we show that every vertex is added to at most one cluster. This is clearly true for quadrics that are added to $C_0$.From Property~\ref{prop:er}.\ref{prop:path}, we can say that the vertices adjacent to $v$ are independent and do not have any common neighbor other than $v$. Hence, any non-quadric vertex is added to at most one cluster.
%\begin{claim}\label{claim:assign1}Algorithm~\ref{alg:racks} adds every vertex in \ER{q} to exactly one cluster.
%\end{claim}
%\lmcomment{Do you mean "Algorithm~\ref{alg:racks} adds every vertex in \ER{q} to at least one cluster."?}\klcomment{Thanks, fixed}
%\begin{proof}
%Each vertex is at most $2$-hops away from the selected quadric $v$, and algorithm~\ref{alg:racks} assigns clusters to $W(q)$ and all non-quadrics at shortest distance $\leq2$ from $v$.
%\end{proof}
%\begin{claim}
%Algorithm~\ref{alg:racks} adds every vertex in \ER{q} to exactly one cluster.
%\end{claim}
%\begin{proof}
%This is clearly true for quadrics that are added to $C_0$.
%From property~\ref{prop:er}.\ref{op:path},
%we can say that the vertices adjacent to $v$ are independent and do not have any common neighbor other than $v$. Hence, any non-quadric vertex is added to at most one cluster. This combined with the claim~\ref{claim:assign1} shows that every vertex is added to exactly one cluster.
%\end{proof}
%\kicomment{could we combine these claims into one?} \lmcomment{I think we should.}
\subsection{Intra-rack Layout}\label{sec:layout_intra}
\subsubsection{The quadrics cluster}
The layout of the quadrics cluster is quite simple: there are no edges within $C_0$ by  Property~\ref{prop:er}.\ref{prop:quad_adj}. This is seen in $ER_7$ in 
Figure~\ref{fig:clustercake_design}(\subref{fig:clustercake_a}), and also in $ER_3$ in Figure~\ref{fig:dpc}.
\begin{figure}[ht]
\begin{centering}
\includegraphics[width=.6\columnwidth]{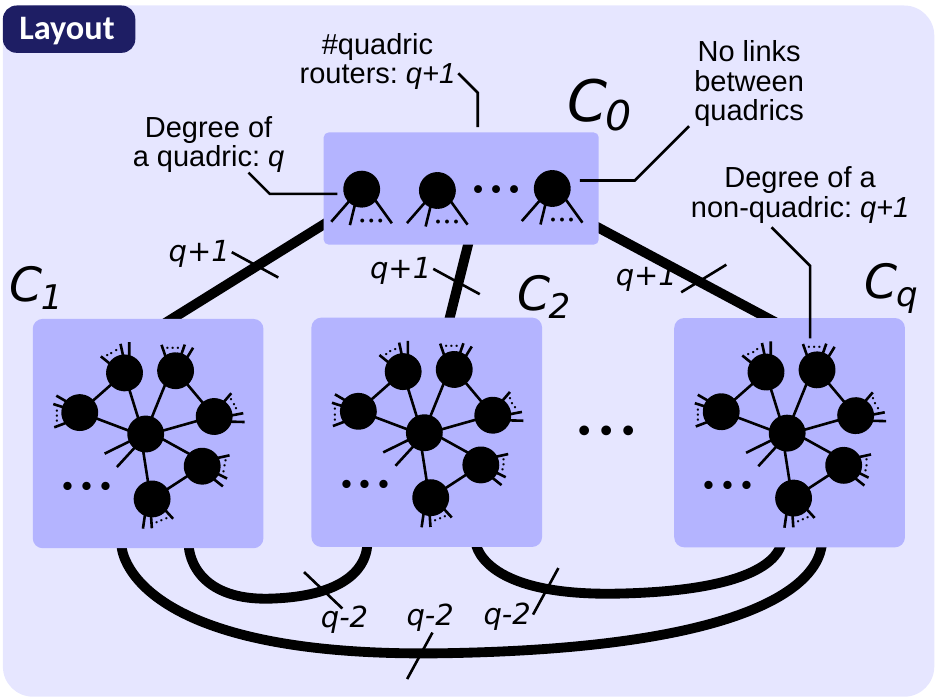}
\caption{The layout.}
\label{fig:layout}
\end{centering}
\end{figure}
\subsubsection{Non-quadrics clusters}
%Visually, $C_i$ has a fan-like structure. $C_i$ has a center element, and all other elements are in triangles containing that center, forming fanblades. This can be seen in Figure~\ref{fig:clustercake_design}.
The layout of a non-quadrics cluster is also rather simple: the edges of a non-quadrics cluster form a fan made up of $\frac{q-1}{2}$ triangles. The fan has a center vertex, %\kldelete{which is a neighbor of one of the quadrics} 
and centers of all fans have a common quadric neighbor. The triangles share only the center, and are otherwise disjoint, giving a fan-blade appearance. This is easily seen in $ER_7$ in Figure~\ref{fig:clustercake_design}, and may be traced out in $ER_3$ in Figure~\ref{fig:dpc}.

\begin{proposition}\label{prop:tris_in_cluster}
The vertices of a cluster form $\frac{q-1}{2}$ triangles, all having the center of the cluster as a common vertex.
\end{proposition}
\begin{proof}
%\kldelete{a center element $c_i \in V_1(q)$ is chosen that is adjacent to a quadric} 
Let $v$ be the quadric chosen as a starter. Then every vertex $c_i$ adjacent 
to $v$ is selected as the center of cluster $C_i$. By Property~\ref{prop:er}.\ref{prop:lone_adj}, $c_i$ has $q-1$ non-quadric neighbors, and with $c_i$, these
%\kldelete{and by Algorithm~\ref{alg:racks}, these, with $c_i$,} 
make up the $q$ elements of the non-quadric cluster $C_i$. So there is an edge between $c_i$ and any other vertex in $C_i$. 

If $u$ is one of the non-center vertices, Property~\ref{prop:er}.\ref{prop:path} tells us that the edge $(c_i,u)$ is contained in exactly one triangle in $C_i$. 
This shows that $C_i$ consists of exactly 
$\frac{q-1}{2}$ edge-disjoint
triangles, each of which contains $c_i$.
\end{proof}

%Let $G_i$ denote the subgraph induced by vertices in cluster $C_i$. We analyze the structure of $G_i$ to derive intra-rack layout and wiring for \flyN.  

%Consider an arbitrary cluster $C_{i\mid i> 0}$ that contains a neighbor $u$ of the selected quadric $v$~(algorithm~\ref{alg:racks}).
% Let $C_i$ contain vertex $u$ that is a neighbor of
% the selected quadric $v$ (algorithm~\ref{alg:racks}). 
%All $q-1$ non-quadric neighbors of $u$ are in $C_i$. Hence, $\abs{C_i}=q$, and there is an edge between $u$ and all other vertices in $C_i$. We refer to $u$ as the \textit{center} of $C_i$.

%Let $w\in C_i$ be a non-quadric neighbor of the center $u$. By Property~\ref{prop:er}.\ref{prop:path},  the edge $(u,w)$ is contained in exactly one triangle in $G_i$.
% In other words, every vertex in $C_i\setminus \{u\}$ is adjacent to $u$ and exactly one more
% vertex in $C_i$. 
%Hence, $G_i$ consists of exactly $\frac{q-1}{2}$ edge disjoint triangles, each of which contains $u$. Visually, $G_i$ has a fan-like structure with each triangle being a fanblade and $u$ being the center of the fan\footnote{For even prime power $q$, each cluster $C_{i\mid 0<i\leq q+1}$ induces a star graph with $\abs{C_i}=q$. $C_0$ consists of a single node which is adjacent to the center of every star.}.

This fan structure gives \fly a modular layout with isomorphic 
structures induced by all $q$ non-quadric clusters. In physical terms, this means that
the $q^2$ non-quadric nodes can be deployed as $q$ identical copies of the same rack. 
%\kldelete{much simpler intra-rack layout than Slimfly has.} %\htorcomment{how do we define simple? Maybe illustrate more (so far it sounds more complex to me tbh} \lmcomment{Does the figure help visualize it? It is just a bunch of triangles arranged as a fan, with the center hooked up to a quadric.} \klcomment{Best to avoid loose terminology like "simpler".} \lmcomment{But I think we \textit{can} say "simple".}
Moreover, the \textit{fan
layout for \fly is generalizable to any odd prime power $q$},
unlike Slim Fly where generator sets and thus
the intra-rack layout can change significantly with router radix.
% This facilitates construction of different \fly networks using.

% This generalizability allows \fly to expand incrementally without rewiring the entire system, as we discuss in Section~\ref{sec:extensible}.
\begin{figure}[ht]
    \centering
    \subfloat[A non-quadric cluster. An arbitrary quadric (in red) generates the center vertices (in bright green). For each center vertex $c$, all non-quadric neighbors of $c$ are added to the cluster. Edges of one such cluster are rendered in black. This figure illustrates Section~\ref{sec:layout_intra}.\label{fig:clustercake_a}]
    {
        \includegraphics[width=0.20\textwidth]{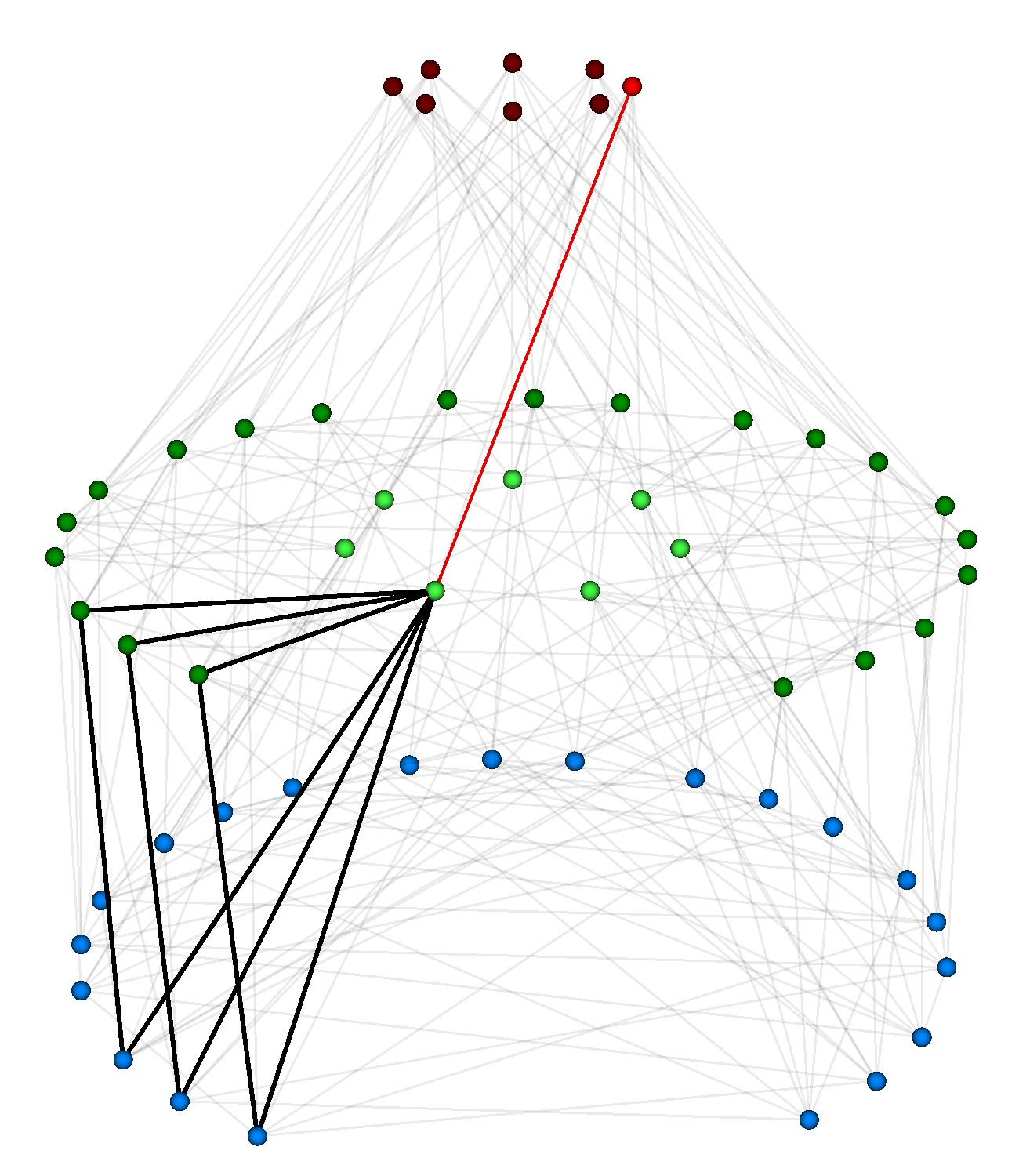}
    }
     \quad
    \subfloat[Inter-cluster connectivity from the cluster shown in (a). Edges to the quadrics are shown in red, edges to the \lone{q} vertices in another cluster in green, and edges to the \ltwo{q} vertices in the other cluster in blue. This figure illustrates Propositions~\ref{claim:C0_Ci} and \ref{claim:Ci_Cj}.\label{fig:clustercake_b}]
    {
        \includegraphics[width=0.20\textwidth]{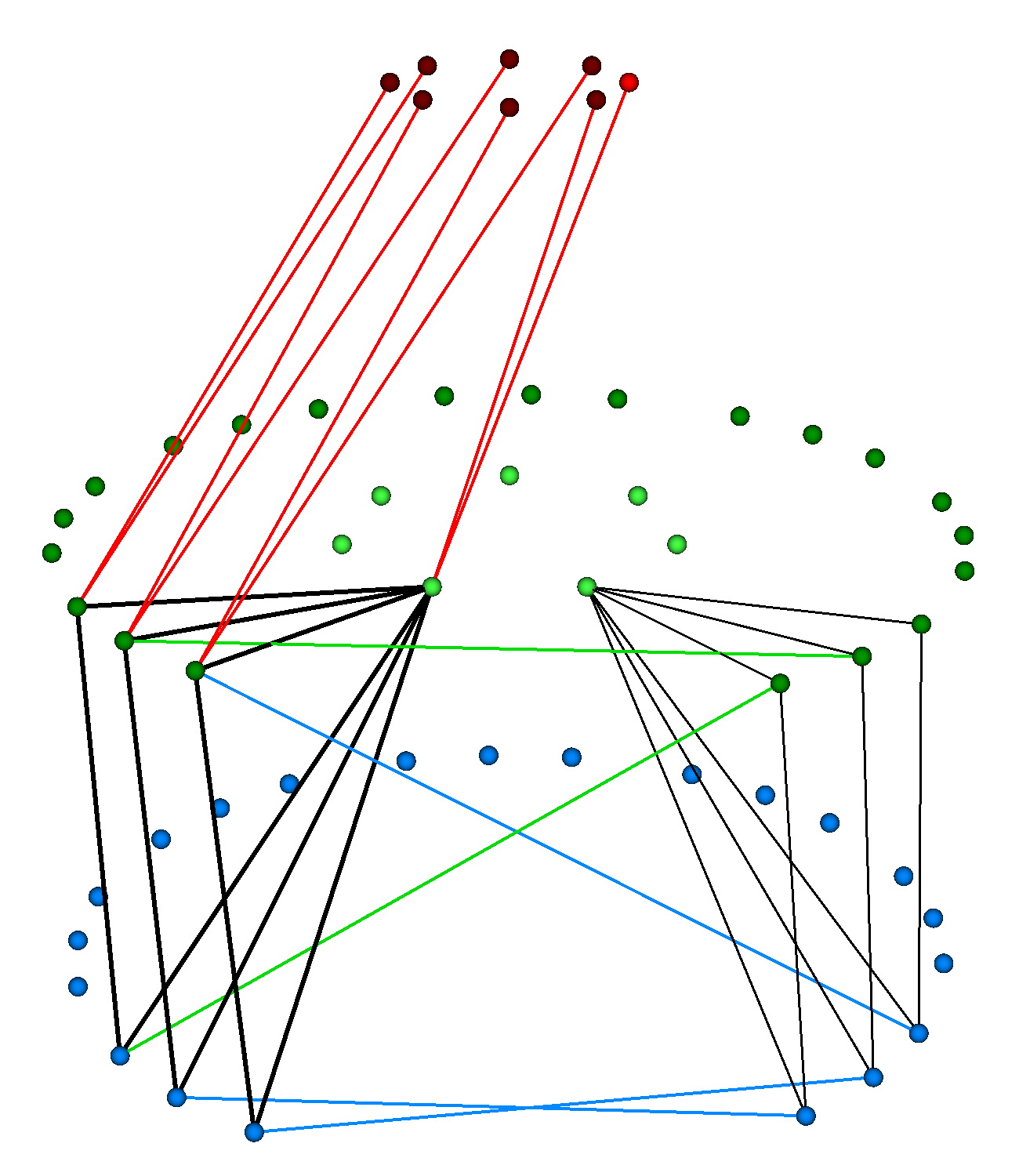}
    }
    \caption{Layout for \fly with $q=7$. 
    The quadrics \quadric{7}~(quadrics) are in red (top layer), \lone{7} are in green (middle layer) and \ltwo{7} are in blue (bottom layer).
    The left figure shows the $\frac{q-1}{2}=3$ fan-blades~(triangles) 
    emanating from the center $c$~(light green) of a non-quadric cluster.
    The right figure shows the $q-2=5$ edges between 
    two non-quadric clusters, and $q+1=8$ edges between a non-quadric cluster
    and quadrics. }
    \label{fig:clustercake_design}
\end{figure}

\subsection{Inter-rack Layout}\label{sec:layout_inter}
There are two types of racks in \fly: one 
quadric rack $C_0$ and $q$ non-quadric
racks $\{C_1, ... C_q\}$. In this section,
we describe the connectivity between two racks.
For ease of notation, we use \lone{q,C_i} and 
\ltwo{q,C_i} to denote vertices of $C_i$ in 
\lone{q} and \ltwo{q} subsets, respectively~(Section~\ref{sec:layout_property}). 
\subsubsection{Connections between $C_i$ and the quadrics}\label{sec:layout_quads}
%There are a total of $q+1$ edges between $C_0$ (the quadrics cluster) and $C_i$. Each vertex in $\lone{q,C_i}$ connects to two quadrics. There are $\frac{q+1}{2}$ vertices in $\lone{q,C_i}$, and each quadric only connects to one vertex in cluster $C_i$. There are no connections between $C_0$ and $\ltwo{q, C_i}$. 
%This is stated more formally in Proposition \ref{claim:C0_Ci}, and the proof gives a detailed explanation of the connections.
\begin{proposition}\label{claim:C0_Ci}
For every cluster $C_i$, with $i>0$:
\begin{enumerate}[leftmargin=*]
    \item Each vertex in \lone{q,C_i} is 
    adjacent to exactly two vertices in 
    $C_0$.\label{claim:C0_Ci_1}
    \item There are exactly $q+1$ links 
between $C_i$ and $C_0$.\label{claim:C0_Ci_2}
    \item Each quadric in $C_0$ 
is adjacent to exactly one vertex in $C_i$.\label{claim:C0_Ci_3}
\end{enumerate} 
\end{proposition}
\begin{proof}
Proposition~\ref{claim:C0_Ci}.\ref{claim:C0_Ci_1} follows directly from Property~\ref{prop:er}.\ref{prop:lone_adj}.
$C_i$ has $\frac{q+1}{2}$ vertices from \lone{q}: the center 
vertex and its $\frac{q-1}{2}$ neighbors in \lone{q}. Hence, 
there are $q+1$ edges between $C_i$ and $C_0$.
Since there is an average of $1$ link to $C_i$ per quadric,
to prove Proposition~\ref{claim:C0_Ci}.\ref{claim:C0_Ci_3}, it is 
sufficient to show that no quadric is adjacent to more than 
one vertex in $C_i$.
Let $v\in W(q)$ be adjacent to multiple vertices in $C_i$.
Either $v$ lies in a triangle or there are multiple
$2$-hop paths between $v$ and the center of $C_i$.
This contradicts Property~\ref{prop:er}.\ref{prop:path}.
\end{proof}
%using $ER_7$ 
\subsubsection{Connections between the non-quadric clusters}\label{sec:layout_nonquads}
%Let $C_i$ and $C_j$ be two clusters with $0 < i,j \le q$. We now describe the connections between these two clusters. Every vertex in $C_i$ except the center and one non-center vertex in \lone{q}, connects to exactly one vertex in $C_j$. The center does not connect to any vertex in $C_j$. This gives a total of $q-2$ edges between the two clusters. This is seen for $ER_7$ in the green and blue edges in Figure~\ref{fig:clustercake_design}(\subref{fig:clustercake_b}). 
%We state this more formally in the Proposition \ref{claim:Ci_Cj}. The proof explains the connections in more detail.
\begin{proposition}\label{claim:Ci_Cj}
Let $C_i$ and $C_j$ be two clusters with $0 < i,j \le q$ and $i\ne j$. Let $c_i$ denote the center of cluster $C_i$. Then: 
%for any $0<i\leq q$. 
%For every $0<i,j\leq q$ where $i\neq j$:
\begin{enumerate}
    \item Every vertex in \ltwo{q,C_i} is adjacent
    to exactly one vertex in $C_j$.\label{claim:Ci_Cj_1}
    \item There are $q-2$ independent edges between $C_i$ and $C_j$.
    \label{claim:Ci_Cj_2}
    \item There exists a vertex $u'\in$\lone{q,C_i}$\setminus \{c_i\}$ 
    such that every vertex in \lone{q,C_i}$\setminus\{u', c_i\}$ is adjacent
    to exactly one vertex in $C_j$, and
    $u'$ is not adjacent to $C_j$. \label{claim:Ci_Cj_3}
    \label{claim:Ci_Cj_4}
\end{enumerate}
\end{proposition}
\begin{proof}
For a vertex $w\in C_i$ such that $w\neq c_i$, let
$S(w)$ denote the set of edges between $w$ and all
clusters other than $C_i$ and $C_0$. 
If $w\in$\lone{q,C_i}, then $\abs{S(w)}=q-3$, and if 
$w\in$\ltwo{q,C_i}, then $\abs{S(w)}=q-1$. Every edge in $S(w)$ must connect 
$w$ to a unique cluster, otherwise there will be
multiple 2-hop paths from $w$ to the center of a cluster. 
This proves Proposition~\ref{claim:Ci_Cj}.\ref{claim:Ci_Cj_1}.

We give a brief sketch of the proof for 
Proposition~\ref{claim:Ci_Cj}.\ref{claim:Ci_Cj_2}
for brevity.
Adding $\abs{S(w)}$ for all vertices $w\in C_i$,
we see that there are $(q-1)(q-2)$ edges between 
$C_i$ and other non-quadric clusters. If 
Proposition~\ref{claim:Ci_Cj}.\ref{claim:Ci_Cj_2}
is false, then there exists a non-quadric
vertex at least $3$-hops away from the center $c_i$. 
This is a contradiction since \ER{q} has diameter 2.

From Propositions~\ref{claim:Ci_Cj}.\ref{claim:Ci_Cj_1}
and \ref{claim:Ci_Cj}.\ref{claim:Ci_Cj_2}, we see that there
are $(q-3)/2$ independent edges between \lone{q,C_i}
and $C_j$. None of these edges are incident to $c_i$. Further, $\abs{\lone{q,C_i}\setminus\{c_i\}}=\frac{q-1}{2}$. 
Hence, there must be exactly one $u'\in$\lone{q,C_i}$\setminus\{c_i\}$
which is not adjacent to $C_j$. Vertex
$u'$ can be identified easily as it shares a common quadric
neighbor with center $c_j$.
\end{proof}
%This is seen for $ER_7$ in the green and blue edges in
%Figure~\ref{fig:clustercake_design}(\subref{fig:clustercake_b}). 
Propositions~\ref{claim:C0_Ci} and \ref{claim:Ci_Cj}
show that there are $q-2$ links between all $\frac{q(q-1)}{2}$ pairs of non-quadric racks, and
$q+1$ links between quadric rack and each of the $q$ non-quadric racks. Thus, PolarFly exhibits a nearly
balanced all-to-all connectivity between racks. Moreover, links between any pair of racks can be
\textit{bundled} into
cost-effective solutions such as a single multi-core fiber~\cite{bundlefly_2020}.

Using $ER_7$ as an example, Proposition~\ref{claim:C0_Ci} may be seen in the red edges of
Figure~\ref{fig:clustercake_design}(\subref{fig:clustercake_b}), and Proposition~\ref{claim:Ci_Cj} may be seen in the green and blue edges of 
Figure~\ref{fig:clustercake_design}(\subref{fig:clustercake_b}).
% $q+1$ \htorcomment{I'd not use asymptotic notation here, can be more precise - i.e., what share of the cables can be bundled up? It seems to me that it's close to a balanced all-to-all connectivity, which is the same for other such topologies (SF, DF)?} independent links
% between any two racks, that can be \textit{bundled} into
% cost-effective solutions such as a single multi-core fiber~\ref{x}. Thus, PolarFly racks exhibit a nearly
% balanced all-to-all 
% For any pair of non-quadric
% clusters, these links connect unique routers 
% in the corresponding racks.
% \klcomment{Write something about how the 
% connectivity can be established with shufflers/swizzlers easily.}
% \subsection{ER Graphs}
% \klcomment{Only for even degrees (odd $q$)}
%     \begin{itemize}
%         \item Clusters $\leftrightarrow$ racks
%         \item $C_0 \leftrightarrow C_i$ connections
%         \item $C_i \leftrightarrow C_j$ connections
%         \item can be bundled
%     \end{itemize} 

% \subsection{Abas Graphs}
% \klcomment{TODO}
\subsection{Triangles and Other Polygons in \flyN}
\fly has many triangles, and no quadrangles. The absence of quadrangles is because each pair of non-quadric vertices has exactly one path of length $2$, by Property~\ref{prop:er}.\ref{prop:2path}. If \fly had a quadrangle, there would then be vertices with two paths of length $2$ between them.

We discuss the abundance of triangles and the implications of this throughout the rest of this section. This becomes important when we analyze the path diversity of \fly in the next section.

\subsubsection{The Number and Types of Triangles in \fly}
Triangles are of two kinds: those entirely internal to a non-quadric cluster, and those linking three distinct non-quadric clusters together. 

\begin{proposition}\label{prop:num_tris}
There are ${q+1 \choose 3}$ triangles in $ER_q$.
\end{proposition}
\begin{proof}
 By Property~\ref{prop:er}.\ref{prop:path}, edges not incident to a quadric participate in one triangle, and edges incident to a quadric participate in no triangles. There are $\frac{q(q+1)^2}{2}$ total edges, and $q(q+1)$ edges incident to a quadric. Thus $$\frac{(q+1) q (q-1)}{2}$$ edges are not incident to a quadric. This implies that there are $$\frac{(q+1) q (q-1)}{6}= {q+1 \choose 3}$$ triangles. 
\end{proof}
\begin{proposition}\label{prop:tris}
The triangles of $ER_q$ either join three distinct non-quadric clusters or are internal to a single non-quadric cluster. In particular,
\begin{enumerate}[label=(\alph*),itemsep=1.25ex]
    \item ${q \choose 3}$ triangles join non-quadric clusters.\label{prop:tris_intercluster}
    \item ${q \choose 2}$ triangles are internal to non-quadric clusters. \label{prop:tris_internal}
\end{enumerate}
\end{proposition}
\begin{proof}
There are a total of $\binom{q}{2}(q-2)$ inter-cluster edges across non-quadric clusters. From Property~\ref{prop:er}.\ref{prop:2path}, each such edge participates in exactly one triangle.  So there are 
$$\frac{q(q-1)(q-2)}{2 \cdot 3} = {q\choose 3}$$ such inter-cluster triangles. Each of these must be made up of three distinct clusters: if not, then the edge $(a,b)$ internal to a cluster will also be on a triangle entirely internal to that cluster, thus there will be two length-$2$ paths between $a$ and $b$, which cannot occur.

There are $\frac{q-1}{2}$ triangles internal to a non-quadric cluster, by Proposition~\ref{prop:tris_in_cluster}, and there are $q$ such clusters, so there are $$q\cdot \frac{q-1}{2} = {q \choose 2}$$ triangles internal to non-quadric clusters.

Finally, $${q\choose 3}+{q \choose 2} = {q+1\choose 3},$$ so this accounts for all triangles in $ER_q$, by Proposition~\ref{prop:num_tris}.
\end{proof}

\subsubsection{Intra-cluster triangles}
The $\frac{q-1}{2}$ internal triangles in a non-quadric cluster share the cluster center as a triangle vertex, giving the edges of the cluster the form of a triangle fan-out. They pairwise share no other vertices. 

If $q\equiv 1 \mod 4$, the vertices of each internal triangle consist of the center and either two vertices from $V_1$, or two vertices from $V_2$. 

If $q\equiv 3 \mod 4$, the vertices of each internal triangle consist of the center, an element of $V_1$ and an element of $V_2$. This may be seen in Figures~\ref{fig:clustercake_b} and \ref{fig:comparecontrast_17_19}.

\subsubsection{Inter-cluster Triangles and a Block Design on Clusters}
{
The remaining triangles are all inter-cluster triangles.
In this section, we prove the following Theorem~\ref{th:clustertriplets_triangles}, which says that every non-quadric cluster triplet is joined by one triangle.

This gives rise to a $3-(q, 3, 1)$ block design, where the $q$ non-quadric clusters are the points, triangles joining cluster triplets are the blocks, and each set of $3$ points appears in $1$ block. 
Block designs are well known combinatorial structures that express symmetries in terms of the points and blocks  \cite{van2001course}, and are therefore of interest in constructing networks.
%\newline\newline
\begin{theorem}\label{th:clustertriplets_triangles} Every triplet of non-quadric clusters under any cluster layout is connected by exactly one triangle. 
\end{theorem}
%\newline\newline
The theorem is a consequence of 
the symmetry between all length-$2$ paths in $ER_q$ that have a quadric as the intermediate node. %Corollary 5 
This symmetry was first shown in \cite{parsons_1976}, and is restated below as Theorem~\ref{th:aut_w_path}. 
%(There are several other interesting symmetries on paths and on triangles in $ER_q$, also shown in \cite{parsons_1976}.) 

To prove Theorem~\ref{th:clustertriplets_triangles}, it suffices to show that all non-quadric triplets are joined by at most one triangle. The theorem then follows from an application of the pigeonhole principle.

To do this, we show a corollary of Theorem~\ref{th:aut_w_path} that expresses a symmetry on non-quadric cluster triplets. A technical lemma exhibits an example of such triplets that are joined by at most one triangle. The symmetry in the corollary then implies that this is true of all non-quadric triplets in $ER_q$.

%Theorem~\ref{th:clustertriplets_triangles} follows from a corollary of Theorem~\ref{th:aut_w_path} that expresses a symmetry on non-quadric cluster triplets. A technical lemma exhibits a set of such triplets having the property that each is joined by at most one triangle, and the symmetry in the corollary implies that this is true of all non-quadric triplets. An application of the pigeonhole principle then completes the proof.

%Theorem~\ref{th:clustertriplets_triangles} follows from a corollary to Theorem~\ref{th:aut_w_path} that expresses symmetry on non-quadric cluster triplets, then a technical lemma and an application of the pigeonhole principle to the number of non-quadric cluster triplets.

%We then derive a corollary which says that if one non-quadric cluster is such that all non-quadric cluster triplets including that cluster has a certain property, then all non-quadric cluster triplets do, and show one type of cluster that has this property. The theorem then follows by showing that such a cluster actually exists, and applying the pigeonhole principle to the number of non-quadric cluster triplets.
}
\begin{theorem}\label{th:aut_w_path} \cite[Corollary 5]{parsons_1976}
Let $(s_{0},w_0,d_{0})$ and $(s_{1},w_1,d_{1})$ be paths of length $2$ in $ER_q$, where $s_{i}$ and $d_i \in V_1$ and $w_i \in W$, the quadrics cluster. Then there exists some automorphism $\theta$ of $ER_q$ such that $\theta(w_0) = w_1$, $\theta(s_{0}) = s_{1}$ and $\theta(d_{0}) = d_{1}$. \end{theorem}

{

\begin{corollary}\label{cor:clustertriplets_dualtriangles}
If there exists some non-quadric cluster $X$ such that every non-quadric cluster triplet that includes $X$ is joined by at most one triangle, then every non-quadric cluster triplet is joined by at most one triangle.
\end{corollary}
\begin{proof}
    Let $w$ be the starter quadric, and let $X$ be a cluster meeting the above condition. Let $x_c$ be the center of $X$.
    
    We assume that there exists some triplet $(D, E,F)$ of distinct non-quadric clusters joined by more than one triangle, and show a contradiction.
    
    Let $d_c$, $e_c$ and $f_c$ be the centers of $D, E$, and $F$ respectively. By assumption, the triplet $(D, E, F)$ is joined by two distinct triangles with vertices $(d_0, e_0, f_0)$ and $(d_1, e_1, f_1)$.
    
    Let $Y\neq X$ be an arbitrary cluster with center $y_c$. 
    By Theorem~\ref{th:aut_w_path}, there is an automorphism $\theta$ of $ER_q$ so that $\theta(w)=w$, $\theta(d_c)=x_c$, and $\theta(e_c)=y_c$. 
    Any automorphism of $ER_q$ preserves edges, so $\theta(f_c)$ is connected to $\theta(w)=w$, making $\theta(f_c)$ the center of some cluster $Z$. Again, $\theta$ preserves edges, so two distinct triangles $(\theta(d_0), \theta(e_0), \theta(f_0))$ and $(\theta(d_1), \theta(e_1), \theta(f_1))$ link $X$, $Y$ and $Z$.
    
    By the condition on $X$, $(X,Y,Z)$ cannot be a triplet, so $Z$ must be one of $X$ or $Y$. But then the triangle $(\theta(d_0), \theta(e_0), \theta(f_0))$ joins exactly two non-quadric clusters. This contradicts Proposition~\ref{prop:tris}.
\end{proof}
}
{
\begin{lemma}\label{lemma:one_tri_special_case}
Let $X$ be a non-quadric cluster whose center $x$ has form $(1, x_1, x_2)$ as a point in $\mathbb{P}^2(\mathbb{F}_q)$.
Then any triple of distinct non-quadric clusters that includes $X$ is joined by at most one triangle.
%Then there is no non-quadrics cluster $Z$ such that the triple $(X,Y,Z)$ is joined by multiple triangles.
\end{lemma}
\begin{proof}
Let $X$ be as stated, and let $Y, Z\neq X$ be distinct clusters such that the triple $(X,Y,Z)$ is joined by the triangle $(a,b,c)$ with $a\in X$, $b\in Y$ and $c\in Z$. Write each point $r$ as $r = (r_0, r_1, r_2)$, a point in $\mathbb{P}^2(\mathbb{F}_q)$. Given centers $x, y,$ and $z$ of $X$, $Y$, and $Z$, we will write down a system of equations to count how many possible triples $(a,b,c)$ can exist.

By construction, the points $a$, $b$ and $c$ are connected to the respective centers $x$, $y$ and $z$, and to each other in ER$_q$. This implies that
$$a\cdot x = b\cdot y = c\cdot z =
a\cdot b = b\cdot c = c\cdot a  = 0 
$$
in $\mathbb{F}_q$.
Because $a\cdot x = 0$ and $x=(1, x_1, x_2)$,
\begin{equation}
a_0 = -(a_1x_1 + a_2x_2). \label{eq:triangle_syseq1}
\end{equation}
Because $a\cdot b = b \cdot y = 0 $ and 
$a\cdot c = c \cdot z = 0$, $(a,b,y)$ and $(a,c,z)$ are two-hop paths. 
The cross-product derivation of the intermediate vertex of a two-hop path given in (\ref{eq:cross_product}) from Section~\ref{subsec:intermediate} then shows that
\begin{equation} b = a\times y \, \,\, \text{ and } \,\,\, c = a\times z. \label{eq:cross_prod_bc}
\end{equation}
Using $b\cdot c = 0$, the above cross-products may be substituted for $b$ and $c$, giving:
$$
\left(a\times y\right)\cdot (a\times z) = 0  
$$
Substituting (\ref{eq:triangle_syseq1}) into this equation gives
\begin{equation}
r_{11}a_1^2 + r_{12}a_1a_2 + r_{22}a_2^2 = 0 \label{eq:triangle_crossproduct}
\end{equation}
where $r_{11}, r_{12}$ and $r_{22}$ are constants in $y_i, z_i$ for $i \in \{0,1,2\}$. Notice that (\ref{eq:cross_prod_bc}) implies we can determine the entries of the points $b$ and $c$ completely once we determine $a$. Thus we have reduced our problem to understanding the number of solutions to (\ref{eq:triangle_crossproduct}) in terms of $a_1$ and $a_2$.

{
Either $a_2$ is $0$ or $a_2$ is invertible. If $a_2 = 0$, then (\ref{eq:triangle_crossproduct}) implies $a_1=0$, and (\ref{eq:triangle_syseq1}) implies $a_0=0$. This is impossible since $(0,0,0) \not \in \mathbb{P}^2(\mathbb{F}_q)$. So $a_2$ must be invertible.

Without loss of generality, we may then solve for a vertex of the form $a'=(a_0', a_1',1)$, since $a'$ may then be multiplied by $a_2$ and then left-normalized to give $a \in \mathbb{P}^2(\mathbb{F}_q)$. 

In that case, (\ref{eq:triangle_crossproduct}) reduces to a quadratic polynomial in $a_1'$, which has at most two solutions. However, notice that $a = b = c = w$, the starter quadric, satisfies all of the equations. In this case, $a$, $b$, and $c$ do not form a triangle. This implies there is at most one valid solution to (\ref{eq:triangle_crossproduct}) and at most one vertex triplet $(a,b,c)$ that form a triangle between 
the clusters $X, Y$ and $Z$. 
}

\if 0
{\color{green}
Since all multiples of a vector represent the same vertex, we solve
the system of equations~\ref{eq:triangle_crossproduct} and \ref{eq:triangle_syseq1} for $a_2=0$ and $a_2=1$. The solution 
vectors can be normalized with the corresponding value of the leftmost non-zero coordinate to obtain
left-normalized coordinates of the vertex $a$ that satisfies these equations.
% Since $a$ is in $\mathbb{P}^2(\mathbb{F}_q)$, we can normalize $a_2$ to be 0 or 1. 

If $a_2 = 0$, then (\ref{eq:triangle_crossproduct}) implies $a_1=0$, and (\ref{eq:triangle_syseq1}) implies $a_0=0$. This is impossible since $(0,0,0) \not \in \mathbb{P}^2(\mathbb{F}_q)$. If $a_2 = 1$, (\ref{eq:triangle_crossproduct}) reduces to a quadratic polynomial in $a_1$, which has at most two solutions. However, notice that $a = b = c = w$, the starter quadric, satisfies all of the equations. In this case, $a$, $b$, and $c$ do not form a triangle. This implies there is at most one valid solution to (\ref{eq:triangle_crossproduct}).
Therefore, there
is at most one vertex triplet $(a,b,c)$ that form a triangle between 
the clusters $X, Y$ and $Z$.
}
\fi 
\end{proof}

%\begin{theorem}\label{th:clustertriplets_triangles}
%Every triplet of non-quadric clusters under any cluster layout is connected by exactly one triangle. 
%Every inter-cluster triangle connects exactly one triplet of non-quadric clusters.
%\end{theorem}

We are now ready to prove the main theorem of this section.

\begin{proof}[Proof of Theorem \ref{th:clustertriplets_triangles}]
Let $w=(w_0, w_1, w_2)$ be the quadric vertex that generates the cluster layout. At least one of $w_1$ and $w_2$ is non-zero, since $w$ is non-zero and self-orthogonal.

There is thus at least one vector $x$ having form $(1, x_1, x_2)$ that is orthogonal to $w$ (as can be calculated using the dot product on $w$), and this vector $x$ is the center of some non-quadric cluster $X$. By Lemma~\ref{lemma:one_tri_special_case}, any non-quadric cluster triplet that includes $X$  is joined by at most one triangle.

Corollary~\ref{cor:clustertriplets_dualtriangles} 
then implies that every non-quadric cluster triplet is joined by at most one triangle.
By Proposition~\ref{prop:tris}\ref{prop:tris_intercluster}, the number of inter-cluster triangles is $q\choose 3$, which is the same as the number of non-quadric cluster triplets. The theorem immediately follows by the pigeonhole principle. 
\end{proof}
}
{
\subsubsection{Distribution of inter-cluster triangles}

\begin{table}[ht]
\setlength{\tabcolsep}{2.5pt}
\centering
\renewcommand{\arraystretch}{2}%
\begin{tabular}{|c|c|c|c|c|p{6cm}|}
    \hline
     & $(v_1,v_1,v_1)$ & $(v_1,v_1,v_2)$ & $(v_1,v_2,v_2)$ & $(v_2,v_2,v_2)$ \\
\hline
$q\equiv 1 \mod 4$ & $\frac{q(q-1)(q-5)}{24}$ & $0$ & $\frac{q(q-1)^2}{8}$ & $0$ \\
\hline
$q\equiv 3 \mod 4$ & $0$ & $\frac{q(q-1)(q-3)}{8}$ & $0$ & $\frac{(q+1)q(q-1)}{24}$ \\
\hline
\end{tabular}
\caption{Distribution of inter-cluster triangles, of different forms. The variable $v_1$ indicates a vertex in $V_1$, and the variable $v_2$ indicates a vertex in $V_2$.}
%\vspaceSQ{-1em}
\label{tab:dist_tris}
%\vspace{-0.5em}
\end{table}
We can further classify the triangles in $ER_q$. The distribution of the inter-cluster triangles is as shown in Table~\ref{tab:dist_tris}.
The table follows from a simple combinatorial argument.

We note that every element of $V_1$ serves as a center in the two layouts induced by its adjacent quadrics. Since the choice of a particular layout does not affect adjacency at all, triangles remain the same in all layouts, in terms of participating vertices. Any triangle with a participating center is entirely internal to the center’s cluster. So if $q\equiv 1 \mod 4$, any triangle holding an element of $V_1$ must be of the form $(v_1, v_1, v_1)$ or $(v_1, v_2, v_2)$. Likewise, if $q\equiv 3 \mod 4$, any triangle holding an element of $V_1$ must be of the form $(v_1, v_1, v_2)$.
We also know that the total number of inter-cluster triangles for any $q$ is $q \choose 3$.
 
Clusters are either internal to a cluster, or entirely inter-cluster, as in the proof of Proposition~\ref{prop:tris}. This also implies that if a triangle is internal to a cluster with a center $c$ in a given layout, it will be entirely inter-cluster for any layout in which $c$ is not a center, in other words, layouts induced by any of the $q-1$ quadrics not adjacent to $c$.
 
We assume some particular layout starting with an arbitrary starting quadric, and calculate the inter-cluster triangles for that layout. 
The triangles of course remain the same for any particular layout.

Case 1: $q\equiv 1 \mod 4$.
First, we calculate the number of inter-cluster triangles of the form $(v_1, v_1, v_1)$. Choose one of the $q$ non-quadric clusters at random. There are $\frac{(q-1)}{2}$ non-center $V_1$ elements in the cluster. We choose one of these and call it $w$. Looking at a layout in which $w$ is a center, $w$ participates in $\frac{q-1}{2}$ triangles of the form $(v_1, v_1, v_1)$, of which exactly two are internal to the cluster, so the other $\frac{q-5}{2}$ triangles are entirely inter-cluster. Considering all the non-center $V_1$ elements in all the clusters, we get $$\frac{q(q-1)(q-5)}{4}$$ triangles, but each is counted six times. So there are $$\frac{q(q-1)(q-5)}{24}$$ triangles in total of the form $(v_1, v_1, v_1)$. We also have that $w$ participates in $\frac{q-1}{2}$ triangles of the form $(v_1, v_2, v_2)$. None of these is internal to the cluster, since all internal triangles of that form are $(c, v_2, v_2)$, where $c$ is the center. Considering all the non-center $V_1$ elements in all the clusters, we get $$\frac{q(q-1)^2}{4},$$ but each is counted twice. So there are $$\frac{q(q-1)^2}{8}$$ triangles in total of the form $(v_1, v_2, v_2)$. Since
$$\frac{q(q-1)^2}{8} + \frac{q(q-1)(q-5)}{24} = {q \choose 3},$$ we have accounted for all of the inter-cluster triangles for $q\equiv 1 \mod 4$.

Case 2: $q\equiv 3 \mod 4$.
First, we calculate the number of inter-cluster triangles of the form $(v_1, v_1, v_2)$. Choose one of the $q$ non-quadric clusters at random. There are $\frac{q-1}{2}$ non-center $V_1$ elements in the cluster. We choose one of these and call it $w$. We see by considering a layout in which $w$ is a center that $w_0$ participates in ($\frac{q-1}{2}$ triangles of the form $(v_1, v_1, v_2)$, of which exactly $1$ is internal to the cluster, so the other $\frac{q-3}{2}$ triangles are entirely inter-cluster. Considering all the non-center $V_1$ elements in all the clusters, we get $\frac{q(q-1)(q-3)}{4}$ triangles, but each is counted twice. So there are $$\frac{q(q-1)(q-3)}{8}$$ triangles in total of the form $(v_1, v_1, v_2)$. We know that when $q\equiv 3 \mod 4$, there are no triangles of the form $(v_1, v_2, v_2)$ nor of the form $(v_1, v_1, v_1)$, so the remaining triangles must be of the form $(v_2, v_2, v_2)$, and there are $${q \choose 3} - \frac{q(q-1)(q-3)}{8} = \frac{(q+1)q(q-1)(q-1)}{24}$$ of these.

As a corollary of Table~\ref{tab:dist_tris}, we have Table~\ref{tab:int_vertices}, giving the possible types of the intermediate vertex of an alternative $2$-hop path between two adjacent vertices. Such a 2-hop path always exists if neither of the vertices are quadric.
\begin{table}[ht]
\setlength{\tabcolsep}{10pt}
\centering
\renewcommand{\arraystretch}{2}%
\begin{tabular}{|l|l|l|l|}
\Xhline{1.2pt}
     &      & $\mathbf{v_1}$ & $\mathbf{v_2}$ \\ 
\Xhline{1.2pt}
\multirow{2}{*}{$q\equiv 1\mod 4$} & $\mathbf{v_1}$ & $v_1$ & $v_2$ \\ \cline{2-4} 
     & $\mathbf{v_2}$ & $v_2$ & $v_1$ \\ 
\Xhline{1.2pt}
\multirow{2}{*}{$q\equiv 3\mod 4$} & $\mathbf{v_1}$ & $v_2$ & $v_1$ \\ \cline{2-4} 
    & $\mathbf{v_2}$ & $v_1$ & $v_2$ \\ 
\Xhline{1.2pt}
\end{tabular}
\caption{Types of the intermediate vertices in a $2$-hop path between two adjacent non-quadric vertices.}
%\vspaceSQ{-1em}
\label{tab:int_vertices}
%\vspace{-0.5em}
\end{table}
}

\section{Expandability}\label{sec:extensible}
% \klcomment{Writing some description of expandability.
% It may be covered in background; can remove from 
% here if that's the case.}
% While \fly can scale close to Moore bound, the desired 
% system size may be limited by the budget or buyer requirements.
% The router radix is also a function of 
% technology and needs to be over-provisioned to the 
% nearest available choice, and 
% leaving some ports disconnected 
% on each router. 
% % The desired network can then 
% % be built by using a \fly of appropriate degree, and leaving
% % some ports on each router empty. 

% However, with the availability
% of new devices or additional budget, there may be a need to 
% expand the system in future.
% Therefore, a desirable topology must allow
% \textit{incremental} integration of new nodes 
% (using empty router ports)
% without dismantling and rewiring the entire system.
Expandability is crucial in budget-driven
scenarios, as discussed in Section~\ref{sec:feasibility}.
%, as discussed in Section \ref{sec:feasibility}.
In an underprovisioned expandable network, unused ports on the nodes 
can be used to incrementally connect additional
nodes to the network.
\if 0
Scalable topologies derived from solutions to
degree-diameter problem, such as Slimfly or 
Bundlefly, do not permit such incremental expansion\klcomment{Should we claim this}\htorcomment{can we show this is true? I'm not 100\% sure}.
\fi
% without sacrificing scalability and
% performance. 
In this section, we show that \fly affords \textbf{incremental expansion}
and present two methods to accomplish this. 
\begin{figure}[th]
\begin{centering}
\includegraphics[width=.8\columnwidth]{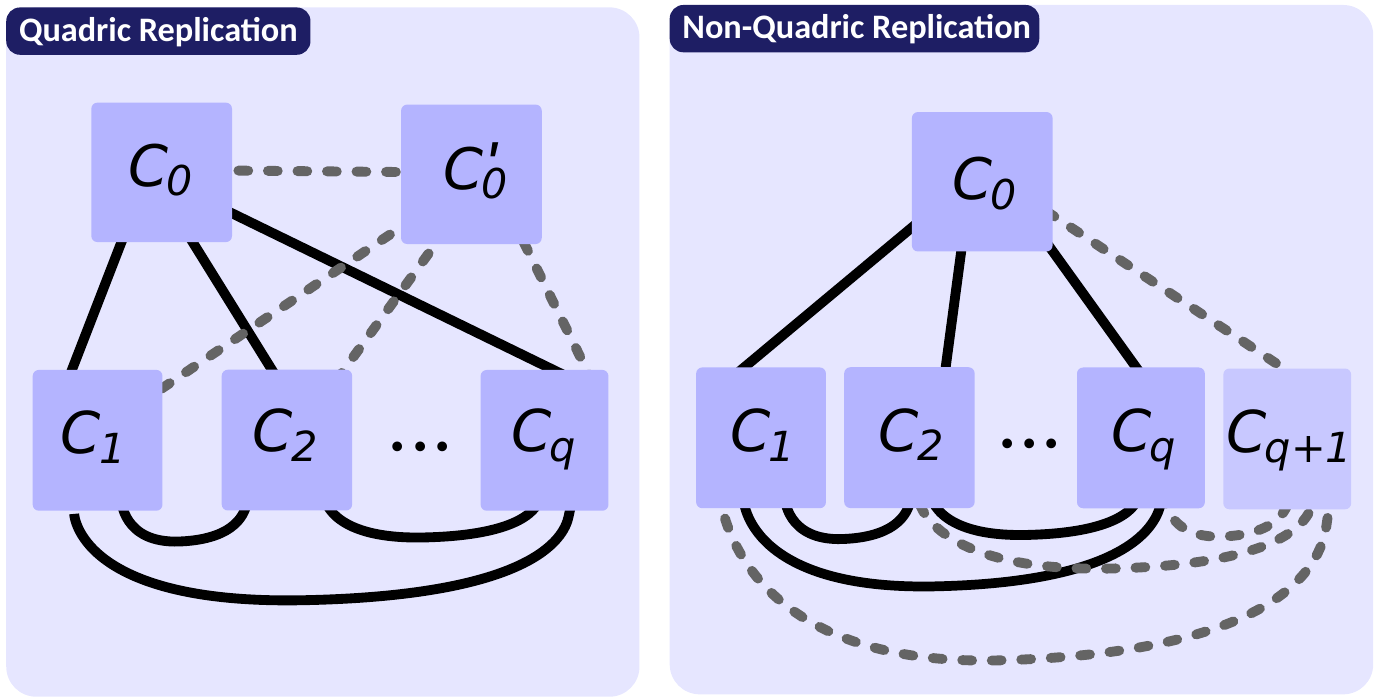}
\caption{Expansion methods.}
\label{fig:expansion}
\end{centering}
\end{figure}

Importantly, these methods
do not require rewiring of the
existing links. They offer a trade-off across
different parameters, as summarized in Table~\ref{tab:expansion}.
% The first two methods are based on \textit{cluster
% replication} in \flyN, which is defined as follows:
These methods are based on \textit{cluster
replication} in \flyN, which is defined as follows:

% \begin{definition}\label{defn:replication}
% Given a graph $G(V, E)$ with $p$ clusters $\{C_0, C_1 ... C_{p-1}\}$, 
% replication of a cluster $C_i$ creates a 
% new graph $G'(V\cup C'_i, E')$. For every vertex $v\in C_i$, there exists a replica $v'\in C'_i$  such that 
% \begin{itemize}%[leftmargin=*]
%     % \item For every vertex $v\in C_i$, there exists a
%     % replica $v'\in C'_i$, and $v'$ is adjacent to $v$ i.e. $(v, v')\in E'$.
%     \item For every edge $(v,w)\in E$ where 
%     $v,w\in C_i$, the corresponding replicas 
%     $v',w'\in C_i'$ are also adjacent i.e. 
%     $(v',w')\in E'$.
%     \item For every edge $(v,w)\in E$ where $v\in
%     C_i$ and $w\notin C_i$, the replica $v'$ is 
%     adjacent to $w$ i.e. $(v',w)\in E'$.
% \end{itemize}
% \end{definition}
\begin{definition}\label{defn:replication}
Given a graph $G(V, E)$, 
replication of a vertex cluster $C\subseteq V$ creates a 
new graph $G'(V\cup C', E')$. For every vertex $v\in C$, 
there exists a replica $v'\in C'$  such that in the graph $G'$:
\begin{itemize}%[leftmargin=*]
    % \item $v'$ is adjacent to $v$ i.e. $(v, v')\in E'$.
    \item For every intra-cluster edge $(v,w)\in E$ between two
    vertices 
    $\{v,w\} \in C$, 
    the corresponding replicas 
    $\{v',w'\}\in C'$ are also adjacent, i.e. 
    $(v',w')\in E'$.
    \item For every inter-cluster edge $(v,w)\in E$ where $v\in
    C$ and $w\notin C$, the replica $v'$ is 
    adjacent to $w$, i.e. $(v',w)\in E'$.
\end{itemize}
\end{definition}

Physically, replication is achieved by simply
adding an additional rack of nodes, which
has similar intra-rack layout and connectivity to rest of the clusters as
its original counterpart.
% Structurally, a replicated rack is \textit{identical} 
% to the original and has same connectivity to other clusters as its original. 
Hence, cluster replication methods~(sec.\ref{sec:expand_M1} and \ref{sec:expand_M2}) 
allow \textit{modular expansion 
without rewiring any of the existing links}.

\subsection{Replicating the Quadric Cluster}\label{sec:expand_M1}
One way to expand \fly is to replicate 
quadric cluster $C_0$ as per Def.\ref{defn:replication},
until the desired scale is reached. 
After replication, to increase the network radix of quadrics, 
we directly connect every quadric $v\in C_0$
and all of its replicas with each other.
% also add a link between 
% every node $v\in C_0$ and its replica $v'\in C'_0$.
It can be shown that every replication of $C_0$:
\begin{enumerate}
    \item Increases the number of vertices by $q+1$, while
    preserving the diameter $D=2$.
    \item Increases the degree of quadrics $W(q)$ (and their replicas) 
    and vertices in \lone{q} by $1$ and $2$, respectively.
    \item Creates $q+1$ edges between the replicated
    cluster $C'_0$ and all other clusters.
\end{enumerate}

With this method, 
using $n$ additional ports per node,
the size of \fly can be increased by $\frac{n(q+1)}{2}$,
while keeping the diameter $D=2$.
% Physically, each replication can be accomplished by
% simply adding a rack similar to $C_0$, and connecting
% it to existing racks, as shown in 
% fig.\ref{fig:expand}. This can be done 
% \textit{without rewiring any of the existing connections}. 
However, new links are only added between quadric 
nodes and \lone{q}. Hence, a large number of
quadric replications can result in a non-uniform degree distribution.

\begin{table}[ht]
\setlength{\tabcolsep}{2.5pt}
\centering
\footnotesize
%\scriptsize
%\ssmall
%\sf
\resizebox{\linewidth}{!}{
\begin{tabular}{cccccc@{}}
\toprule
\textbf{Method} & \makecell[l]{\textbf{Scalability}} & 
\makecell[l]{\textbf{Degree}\\\textbf{Distribution}} & 
\makecell[l]{\textbf{Diameter}} & 
\makecell[l]{\textbf{Average Shortest}\\ \textbf{Path Length}} & 
\makecell[l]{\textbf{Rewiring}} \\
\midrule
Replicate Quadrics & $\frac{q+1}{2}$ & Non-uniform & $2$ & $<2$ & None\vspace{1mm}\\
Replicate Non-Quadrics & $\approx q$ & Uniform & $3$ & $<2$ & None \\
%Upgrade \fly & High & Uniform & $2$ & $<2$ & Partial \\
%
\bottomrule
\end{tabular}}
%\vspaceSQ{-1em}
\caption{Characteristics of Expansion Methods. \textbf{Scalability} 
refers to the increase in number of nodes per
unit increase in the maximum network radix.
}
\vspaceSQ{-1em}
\label{tab:expansion}
%\vspace{-0.5em}
\end{table} 

\subsection{Replicating Non-Quadric Clusters}\label{sec:expand_M2}
In this method, we expand \fly by
replicating non-quadric clusters $C_{i\mid i>0}$ in a 
round-robin order, as per Def.\ref{defn:replication}. 
The replica of $C_i$ is labeled $C_{q+i}$,
as shown in Figure~\ref{fig:expansion}.

Replicating a non-quadric
cluster does not add edges to existing
center vertices, which can lead
to a non-uniform degree distribution. 
To mitigate this,
we note that for every non-quadric cluster $C_{j}$ where $i\neq j$~(and replica $C_{q+j}$ if it exists), there is exactly one 
vertex in $C_i$ with no edges to $C_j$~($C_{q+j}$, respectively).
We connect the replica of this vertex
with the centers of $C_j$ ($C_{q+j}$, respectively).
It can be shown that for any $n\leq q$, $n$ such replications of non-quadric clusters:
\begin{enumerate}
    \item Increase the number of vertices by $qn$, which is approximately $2\times$ compared to
    quadric replication.
    \item Increase the maximum degree by $n+1$.
    \item Increase diameter to $3$ -- for every vertex $u\in C_{i\mid i>0}$ 
    (replica $u'\in C_{i+q}$), there are at most $q-1$
    vertices, all in replica $C_{i+q}$ ($C_i$, respectively), that are at a shortest
    distance of $3$ from $u$ ($u'$, respectively).
\end{enumerate}

With non-quadric repliction, new links are  distributed 
across all vertices, providing a 
near-uniform degree distribution. 
% Each replication can be accomplished by
% simply adding a rack~(fig.\ref{fig:expand}) 
% \textit{without rewiring any of the existing connections}. 
While the diameter of topology increases to $3$, 
the average shortest path length is
still clearly less than $2$. 

\section{Routing}
\label{sec:routing}

% We now describe routing protocols used in PolarFly. 
To facilitate the adoption of PF, we rely on established schemes and show in the evaluation (Section~\ref{sec:eval}) that they deliver high performance. However, to show the highest PF potential, we also develop a new adaptive protocol suited specifically for PF. Note that under co-packaged setting,
nodes and routers are the same entity in direct networks.
% \kicomment{This is the only place I have seen so far that uses abbreviations.}
% \klcomment{Agree, we used abbreviations here because they are used in figure labels. Glad that none of the reviewers complained :P}
\subsection{Minimal Static Routing}

With minimal static routing,
%~(MIN) 
a packet is routed 
from its source router $R_s$ over the minimal path to its
destination router \looseness=-1$R_t$. 
% Since \fly is a diameter-2 topology,
% the length of a minimal path is either 1 or 2 hops.

\subsection{Valiant Routing}\label{sec:val}
Let $R_s$ and $R_t$ denote the source and
destination routers, respectively.
For each packet, the Valiant routing scheme~\cite{valiant1982scheme} selects a random router $R_r$ such that 
$R_r\neq R_s$ and $R_r\neq R_t$. Then, it routes the packet from
$R_s$ to $R_r$ and $R_r$ to $R_t$ along the corresponding shortest paths. This avoids potential
hot spots in the network, but reduces the available bandwidth.

The general Valiant design selects \emph{some} intermediate router.
For \flyN, we use a variant which we call \textit{Compact Valiant}, where $R_r$ 
is chosen from the neighborhood of $R_s$. The path length for 
any packet in Compact Valiant is at most $3$-hops, 
as opposed to $4$-hops in general Valiant.
This reduces the
amount of bandwidth wasted on links 
due to intermediate \looseness=-1traffic.

% Such routing is not beneficial if the path from $R_s$ to $R_t$ is a subset of the $R_r$ to $R_t$
% contains $R_s$.

However, the $3$-hop route selection will be disadvantageous if 
the shortest path between $R_r$ and $R_t$ goes through the source router $R_s$.
In this scenario, the random neighbor selection would result in the packets 
bouncing back to the source router.
Fortunately, in \flyN, this situation is easily avoided as it occurs only when $R_s$ and $R_t$ are adjacent. 
% Specifically, $q$ neighbors of $R_s$ bounce the packet back in \flyn{}.
Hence, we use Compact Valiant only when $R_s$ and $R_t$
are not adjacent.

% in \fly if $R_s$ and $R_t$ are adjacent 
% to each other.
% % under certain conditions.
% Specifically, in this case, the shortest path between $R_r$ and $R_t$ would go via the source router $R_s$,
% and the random neighbor selection would result in the packets ``bouncing back'' to the source router.
% To avoid such bouncing back, we use Compact Valiant only when $R_s$ and $R_t$
% are not adjacent.

\if 0
\kldelete{To avoid this scenario, we enforce the 
following conditions to be 
considered while choosing $R_r$ for the Compact Valiant:

\begin{enumerate}
    \item If $R_s$ and $R_t$ are adjacent and neither is a quadric, 
    then there is exactly one router $R_r\neq R_t$ in the neighbors of $R_s$
    such that the shortest path between $R_r$ and $R_t$ 
    does not contain $R_s$. %$\left(R_s, R_t\right)$. 
    In this case, Compact Valiant has only one choice of $R_r$
    and it incurs $2$ hops as this $R_r$ is also adjacent to $R_t$. 
%    
    \item If $R_s$ and $R_t$ are adjacent and either one is 
    a quadric, the shortest path between $R_t$ and any other 
    neighbor of $R_s$ goes via $R_s$. 
    % contains the link $\left(R_s, R_t\right)$.
    In this case, Compact Valiant cannot be used, and we fall back to a random router being two hops away from both $R_s$ and $R_t$.
\end{enumerate}

\klcomment{Given two adjacent vertices, if none of them is a quadric, then only one choice. If quadric, no choice. The valiant 
misrouting set is all the non-adjacent and one choice if no quadric.}}
\fi

\subsection{Adaptive Routing}

In adaptive routing, the router into which a packet is first injected decides whether this packet should be routed
over a minimal path or over a Valiant path. 
This decision is made on the basis of occupancy of 
\emph{local} output buffers used in the respective paths, as well as the lengths of the considered paths.
% To make this routing decision, the router takes into account both
% the occupancy of output buffers used in a given path, as well as the lengths of the considered paths.
This routing algorithm is called Universal Globally-Adaptive Load-balancing (UGAL)~\cite{ugal2005scheme}.
%
%   We consider a practical deployable version of UGAL, referred to as UGAL-L (``Local'') where each router has access to the occupancy of its own local
%   output buffers, but not to those of other routers. The routing decision is based solely on the length of a given path and on the occupancy of the local output buffer towards said path.
%   This version can be deployed in practice.

\if 0
We consider two versions of UGAL:

\begin{enumerate}
    \item   Global UGAL (UGAL-G): In this version, we assume that every router knows the 
            queue-occupancy of every other router. Hence, the packet's first router can make its
            routing decision based on the occupancy of all queues on a given path. This scenario 
            is not very realistic since in practice, a router will most likely not have access to
            the other routers' queues, but it serves as an approximation of an ideal implementation
            of the UGAL algorithm.
    \item Local UGAL (UGAL-L): Each router has access to the occupancy of its own
    output buffers, but not to those of other routers. The routing decision is based solely on the length of a given path and on the occupancy of the local output buffer towards said path. This version can be deployed in practice.
\end{enumerate}

\fi

For \flyN, we also explore a UGAL variant which
we call \ugalpf. To achieve high bandwidth, 
\ugalpf\  reduces the average hops per packet
by using:
\begin{itemize}[leftmargin=*]
    \item Compact Valiant described in Section~\ref{sec:val}, and
    \item Adaptation threshold -- Valiant 
    path is chosen over min-path only when fractional occupancy 
    of the output buffer towards
    min-path is greater than a threshold ($\frac{2}{3}$ in our case).
\end{itemize}
Thus, \ugalpf\ offers a trade-off between adaptability
of UGAL and low hop count of minimal static routing.

% \klcomment{We can replace these rules
% with a simple statement saying we do not
% employ compact valiant when $R_s$ and
% $R_t$ are adjacent.}

\begin{figure*}[t]
\centering
%\vspaceSQ{-2em}
\begin{subfigure}[t]{0.26 \textwidth}
\centering
\includegraphics[width=1.0\columnwidth]{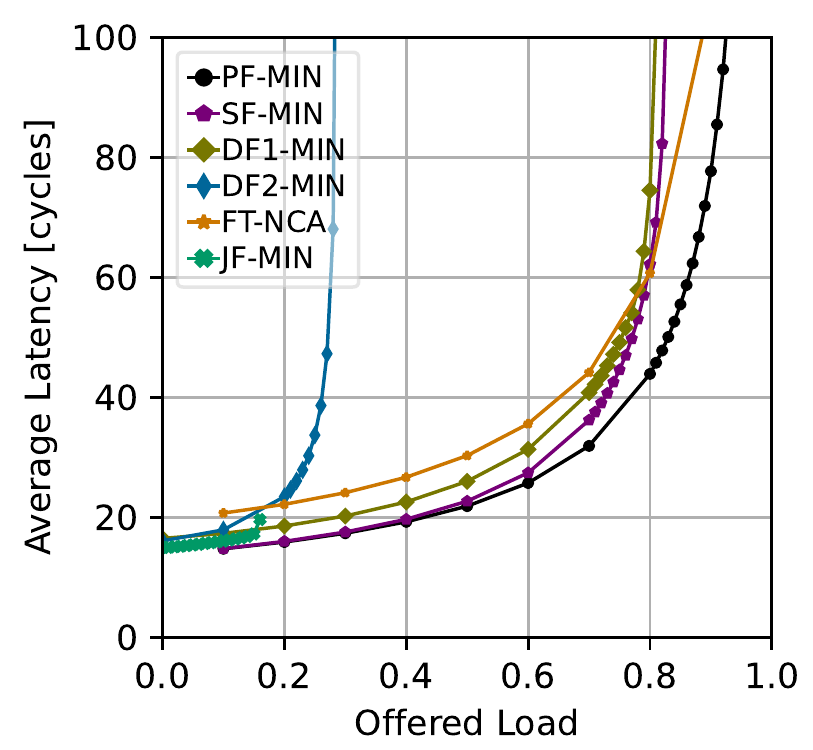}
%
%\vspaceSQ{-2.0em}
\caption{Uniform traffic with min-path routing.}
\label{fig:uniform1}
\end{subfigure}
%\quad
\begin{subfigure}[t]{0.24 \textwidth}
\centering
\includegraphics[width=1.0\columnwidth]{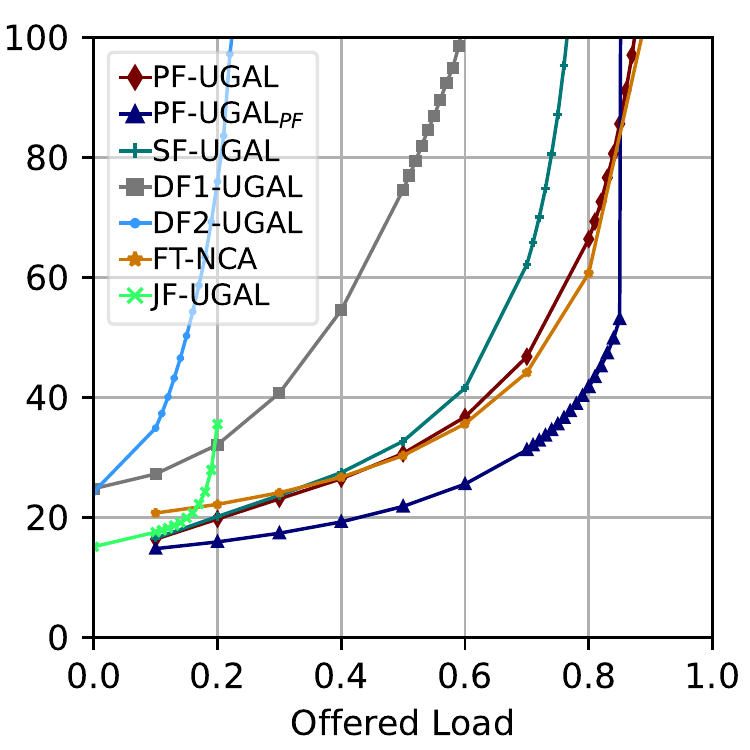}
%
%\vspaceSQ{-2.0em}
\caption{Uniform traffic with adaptive routing.}
\label{fig:perm}
\end{subfigure}
\begin{subfigure}[t]{0.24 \textwidth}
\centering
\includegraphics[width=1.0\columnwidth]{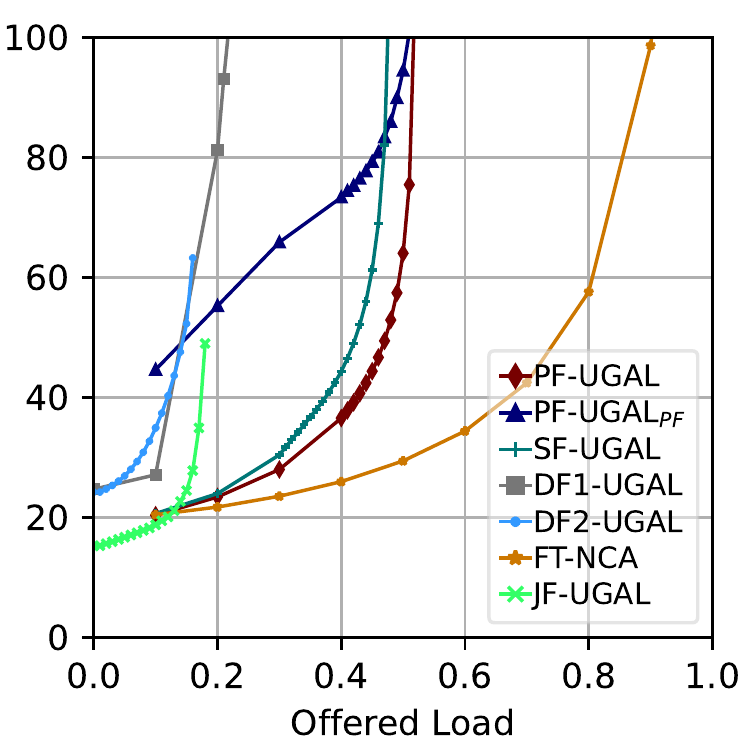}
%
%\vspaceSQ{-2.0em}
\caption{Random Permutation Traffic.}
\label{fig:tornado1}
\end{subfigure}
\begin{subfigure}[t]{0.24 \textwidth}
\centering
\includegraphics[width=1.0\columnwidth]{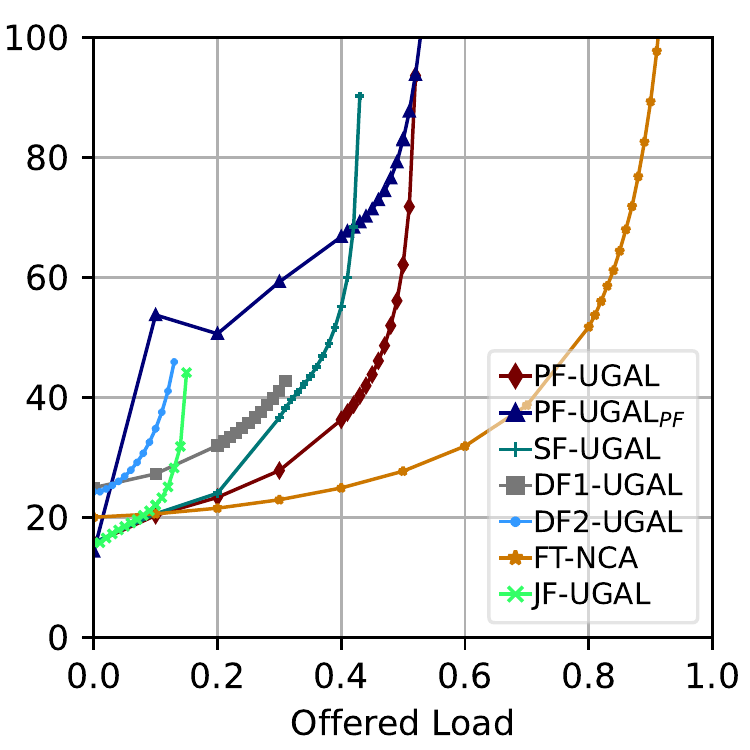}
%
%\vspaceSQ{-2.0em}
\caption{Tornado Permutation Traffic.}
\label{fig:tornado2}
\end{subfigure}
%
%\vspaceSQ{-1em}
\caption{Performance analysis (comparison with other topologies).}
\label{fig:performance-analysis-all}
%\vspaceSQ{-1em}
\end{figure*}

\section{Performance Analysis}\label{sec:eval}

We now evaluate the latency and throughput of PolarFly.

\subsection{Methodology and Comparison Targets}

We compare \flyN{} to Slim Fly~\cite{besta2014slim} (as the most competitive diameter-2 network),
Dragonfly~\cite{dally08} (as a popular recent choice when developing interconnects), Jellyfish~\cite{Singla:2012:JND:2228298.2228322}~(as a random expander network) and 3-level fat tree~\cite{Leiserson:1985:FUN:4492.4495} (as the most widespread existing interconnect baseline).
Except fat tree, all topologies are direct.
As numerous past works
illustrate, networks such as torus, hypercube or Flattened Butterfly 
% \kldelete{, or Random Networks,} 
are less competitive in latency and bandwidth~\cite{besta2014slim, dally08, dally07}.

We use two variants of Dragonfly -- (a)~balanced Dragonfly~(DF1), and (b)~Dragonfly with radix and scale almost
equivalent to \flyN{}~(DF2).
Configurations of the baseline topologies are given in
Table~\ref{tab:performance-config}.

\begin{table}[ht]
\setlength{\tabcolsep}{2.5pt}
\centering
\footnotesize
%\scriptsize
%\ssmall
%\sf
\resizebox{\linewidth}{!}{
\begin{tabular}{llcc@{}}
\toprule
\textbf{Network} & 
\makecell[l]{\textbf{Parameters}} & 
\makecell[l]{\textbf{Number of Routers}} & 
\makecell[l]{\textbf{Network Radix}} \\
\midrule
PolarFly (PF) & q=31, p=16         & 993       & 32 \\
Slim Fly (SF) & q=23, p=18       & 1058      & 35 \\
Balanced Dragonfly (DF1) & a=12, h=6, p=6     & 876       & 17 \\
Equivalent Dragonfly (DF2) & a=6, h=27, p=10  & 978 & 32 \\
Jellyfish (JF) & -- & 993 & 32\\
Fat Tree (FT) & n=3, k=18          & 972       & 36 \\
\bottomrule
\end{tabular}}
%\vspaceSQ{-1em}

\caption{Configuration of topologies used for simulations.}
%\vspaceSQ{-1em}
\label{tab:performance-config}
%\vspace{-0.5em}
\end{table}

Following traffic patterns are simulated to effectively analyze the network \looseness=-1behavior:
% that represent important HPC workloads. 
\begin{enumerate}[leftmargin=*]
    \item \emph{Uniform} random traffic -- for each packet,
    the source selects a destination uniformly at random (representing graph
processing and distributed-memory graph algorithms, sparse linear algebra solvers, and adaptive mesh refinement 
methods~\cite{yuan2013new, besta2014slim, besta2017push, besta2020communication, sakr2021future, gianinazzi2018communication, besta2015accelerating}). 
\item \emph{Tornado}  traffic -- endpoints on every router $i$ send all traffic halfway across to endpoints on router $i + \frac{N}{2}$ modulo $N$.

    \item \emph{Random permutation} traffic -- a fixed permutation 
mapping of source to destination is chosen uniformly at random from the set of all permutations.
In \flyN{}, Tornado and Random permutation traffic are adversarial for min-path routing because there is
only one shortest path between any pair of routers.

    \item Finally, two special permutation traffic patterns \emph{Perm1Hop} and \emph{Perm2Hop} 
    are chosen to analyze \ugalpf. 
    In Perm1Hop, every router communicates with a $1$-hop neighbor -- the
    min-path length is $1$-hop and valiant path length in \ugalpf\ is $4$-hops.
    In Perm2Hop, every router communicates with a $2$-hop neighbor -- the min-path length is $2$-hop and valiant path length in \ugalpf\ is $3$-hops.
    % adversarial traffic patterns are used 
\end{enumerate}
% In uniform 
% traffic, each source sends equal amount of data to each destination (representing graph 
% processing, sparse linear algebra solvers, adaptive mesh refinement 
% methods~\cite{yuan2013new, besta2014slim}). 
%
\if 0
%
In bit shift traffic, 
the bits $d_i$ of the destination address are set based on the bits $s_i$ 
of the source address (representing some stencil workloads and 
collectives~\cite{yuan2013new, besta2014slim}). Since these traffic patterns 
require the number of active endpoints to be a power of two, only a subset of 
endpoints are active in those scenarios. We show results for two traffic patters of 
this category: bit reverse ($d_i = s_{b-i-1}$) and bit shuffle ($d_i = s_{i-1~mod~b}$) 
where $b$ is the number of bits in the address.
%
\fi
%

We use the established BookSim simulator~\cite{jiang2013detailed} to conduct cycle-accurate simulations. 
Each router along with all of its endpoints in BookSim represents a single
co-packaged node.
To mimic co-packaged setting under permutation traffic,
we enforce that all endpoints of a router send data to endpoints of 
only one other router. In other words, permutations 
are computed between routers, and not endpoints.
%
\if 0
%
To avoid experiencing flow
control issues that are outside the scope of this work, we use single flow control unit (single flit) packets
in simulations, following the strategy by Kim et al.~\cite{dally08}. 
%
\fi
%

Packets of size $4$ flits each are injected with a Bernoulli process. 
We use input-queued routers with $128$ flit buffers per port and $4$ virtual channels.
In all simulations, we use a warm-up phase where no measurements are taken,
to ensure that the simulator first reaches a steady-state.

\subsection{Discussion of Results -- Comparison against Baselines}
Figure~\ref{fig:performance-analysis-all} compares the performance
of \flyN{}~(PF) and the topologies shown in Table
\ref{tab:performance-config}. 
The labels follow the scheme
\textit{$<$network$>$-$<$routing$>$}.
The offered load in Figure~\ref{fig:performance-analysis-all} is normalized to the maximum 
capacity of each network.

For Permutation traffic, min-path routing in 
direct networks can achieve at most $\frac{1}{p}$ of peak throughput, because all $p$ endpoints of a source router
access the same path to the destination router. 
Hence, we only compare adaptive routing performance under permutation patterns in these topologies.

In general, we observe that \flyN{} offers superior performance -- for all traffic patterns, it outperforms all competitive direct topologies.
Its advantages over Jellyfish and Dragonfly in terms of 
lower latency, are a direct consequence of its low diameter.
% At small injection bandwidths, the latency
% of \fly and Slim Fly is much smaller than higher diameter topologies.
Its benefits over Slim Fly in terms of higher saturation
bandwidth, are due to careful design of routing protocols that exploit \flyN{}
structure to ensure that the routing decisions are as good as 
possible. Amongst the Dragonflies, the balanced DF1 
outperforms DF2 whose throughput is bottlenecked by the 
traffic volume on intra-group links.

\begin{figure}[ht]
\centering
%\vspaceSQ{-2em}
\begin{subfigure}[t]{0.25 \textwidth}
\centering
\includegraphics[width=1.0\columnwidth]{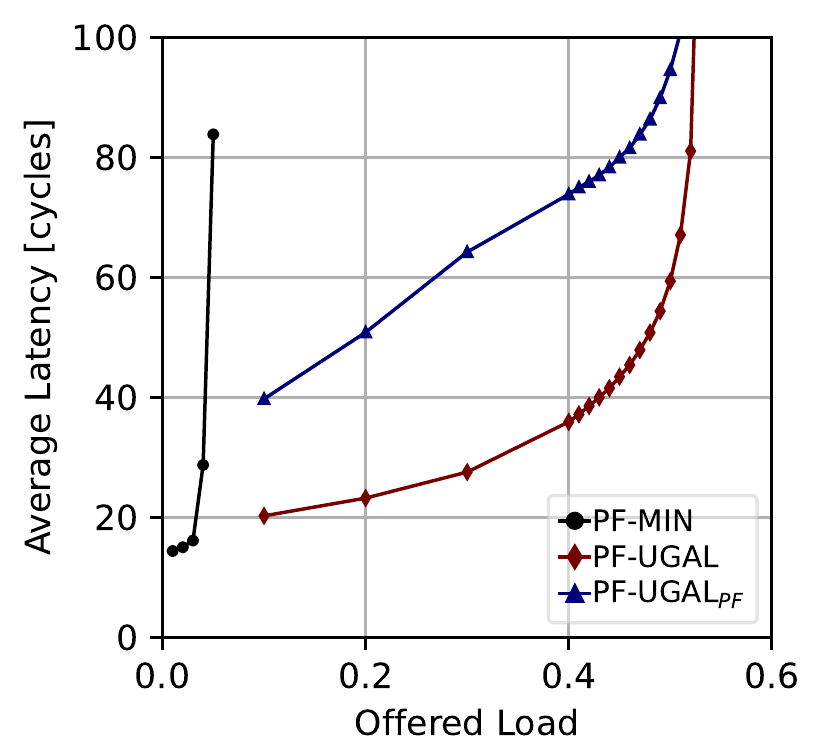}
%
%\vspaceSQ{-2.0em}
\caption{Perm2Hop Permutation Traffic}
\label{fig:adv2}
\end{subfigure}
%\quad
\begin{subfigure}[t]{0.23 \textwidth}
\centering
\includegraphics[width=1.0\columnwidth]{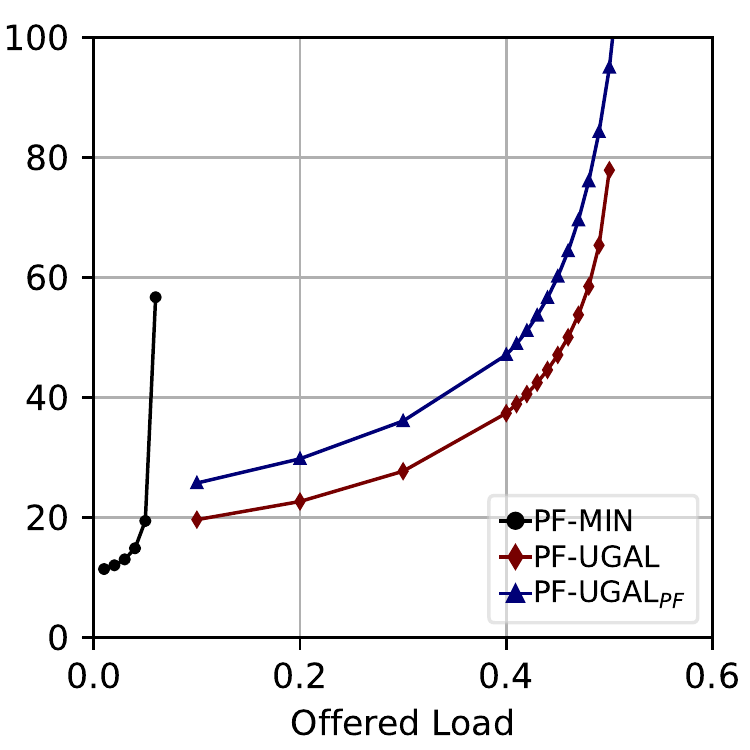}
%
% \vspaceSQ{-2.0em}
%
\caption{Perm1Hop Permutation Traffic}
%\caption{Round robin replication.}
\label{fig:adv1}
\end{subfigure}
%
%\vspaceSQ{-1em}
\caption{Performance of adaptive routing in PolarFly under Permutation Traffic}
\label{fig:performance-analysis-permutation}
%\vspaceSQ{-1em}
\end{figure}

For the Uniform traffic, 
%we observe that 
the 
adaptive routing based on Compact Valiant~(\ugalpf)
exhibits latency and saturation throughput comparable to that of min-path routing, while
significantly outperforming other adaptive algorithms and baseline topologies.
Remarkably, the maximum throughput sustained by \fly for uniform traffic is comparable to the fat tree, with considerable reduction in \looseness=-1latency. 

% For Permutation traffic, minimal routing in \fly can only achieve $\frac{1}{p}$ peak throughput, because all $p$ endpoints of the source 
% access the same path to the destination. However, with adaptive algorithms UGAL and \ugalpf, \fly is able to sustain up to $50\%$ of the full injection bandwidth.
For Random and Tornado Permutation traffic patterns, \fly is able to sustain up to $50\%$ of the full injection bandwidth, using adaptive algorithms UGAL and \ugalpf.
The performance of these patterns is similar to Perm2Hop 
traffic shown in 
figure~\ref{fig:adv2}, as min-path for most packets is $2$-hops long.
The total buffer space
in the min-path is higher compared to Perm1Hop~(Figure~\ref{fig:adv1}, all $1$-hop min-paths), rendering \ugalpf{} slower to adapt to congestion.
Hence, \ugalpf{} has considerably
higher latency than UGAL for Tornado, Random and Perm2Hop permutation patterns. 
% This is because min-path for most packets
% in these patterns is 2 hops long, increasing the total buffer space
% in the min-path, and rendering \ugalpf{} slower to adapt to congestion.
UGAL has relatively higher entropy in terms of path selection, resulting
in smaller queues inside routers and lower latency. 
Figure~\ref{fig:performance-analysis-permutation} also
provides detailed insight into adversarial nature of permutation patterns for min-path routing in \flyN{}. It 
can only withstand 5\% of the full injection bandwidth, compared
to almost 50\% with adaptive routing.
% For both patterns, adaptive routing withstands almost up to 50\% of the full injection bandwidth.

\if 0
This is further confirmed with more detailed analyses into the adversarial traffic patterns on PF (Figures~\ref{fig:adv1}
and \ref{fig:adv2}). Here, the used
adaptive routing is able to withstand almost up to 50\% of the full injection bandwidth.
\fi

\subsection{Discussion of Results - PolarFly Size}

Next, we investigate the impact of PolarFly size
on the performance by (a)~varying $q$ (radix), and (b)~expanding
network incrementally using the methods described in Section~\ref{sec:extensible}.
We analyze balanced variants of PolarFly topology under uniform traffic i.e. the ratio of number of endpoints to network radix is maintained to $1:2$ in all experiments.
% We also investigate the impact of proposed expansion schemes on network performance, as shown in Figure~\ref{fig:performance-analysis-expansion}. 
\begin{figure}[ht]
\centering
%\vspaceSQ{-2em}
\begin{subfigure}[t]{0.25 \textwidth}
\centering
\includegraphics[width=1.0\columnwidth]{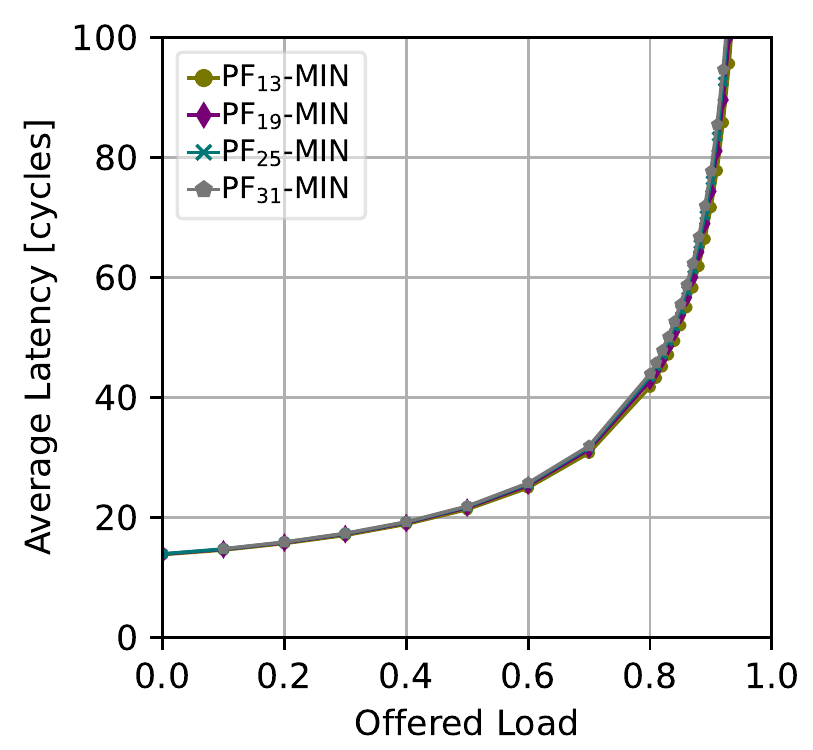}
%
%\vspaceSQ{-2.0em}
\caption{Min-path Routing}
\label{fig:size-min}
\end{subfigure}
%\quad
\begin{subfigure}[t]{0.23 \textwidth}
\centering
\includegraphics[width=1.0\columnwidth]{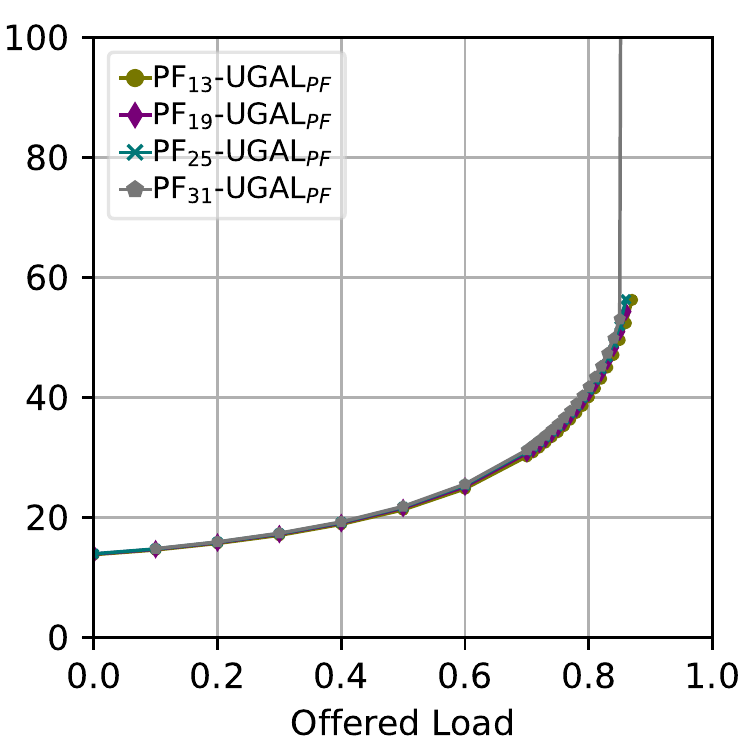}
%
% \vspaceSQ{-2.0em}
%
\caption{Adaptive Routing}
%\caption{Round robin replication.}
\label{fig:size-ugal}
\end{subfigure}
%
%\vspaceSQ{-1em}
\caption{Performance Comparison of Polarfly of different sizes under uniform traffic}
\label{fig:performance-analysis-size}
%\vspaceSQ{-1em}
\end{figure}

Figure~\ref{fig:performance-analysis-size} shows the latency
and throughput for PolarFly for $q=13, 19, 25$ and $31$, which corresponds to $183, 381, 651$ and $993$ routers, respectively.
The labels follow the scheme
\textit{$<$network$>$\textsubscript{$q$}-$<$routing$>$}.
All PolarFlies provide similar saturation bandwidth and latency for both min-path and \ugalpf\  routing. This shows that
PolarFly performance is stable with respect to the size
of the network.
\begin{figure}[ht]
\centering
%\vspaceSQ{-2em}
\begin{subfigure}[t]{0.25 \textwidth}
\centering
\includegraphics[width=1.0\columnwidth]{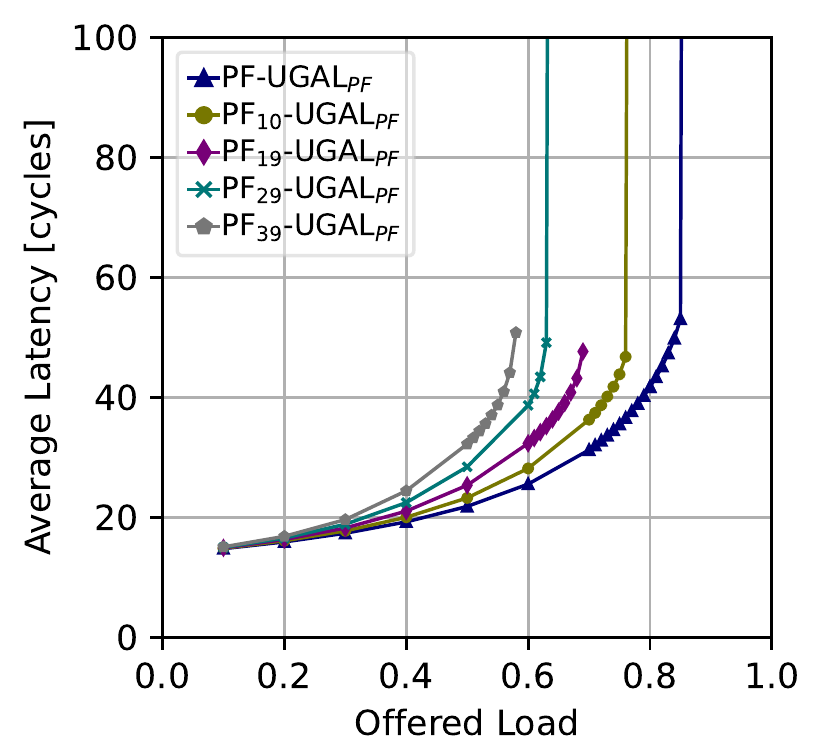}
%
%\vspaceSQ{-2.0em}
\caption{Replicating the quadrics cluster.}
\label{fig:uniform2}
\end{subfigure}
%\quad
\begin{subfigure}[t]{0.23 \textwidth}
\centering
\includegraphics[width=1.0\columnwidth]{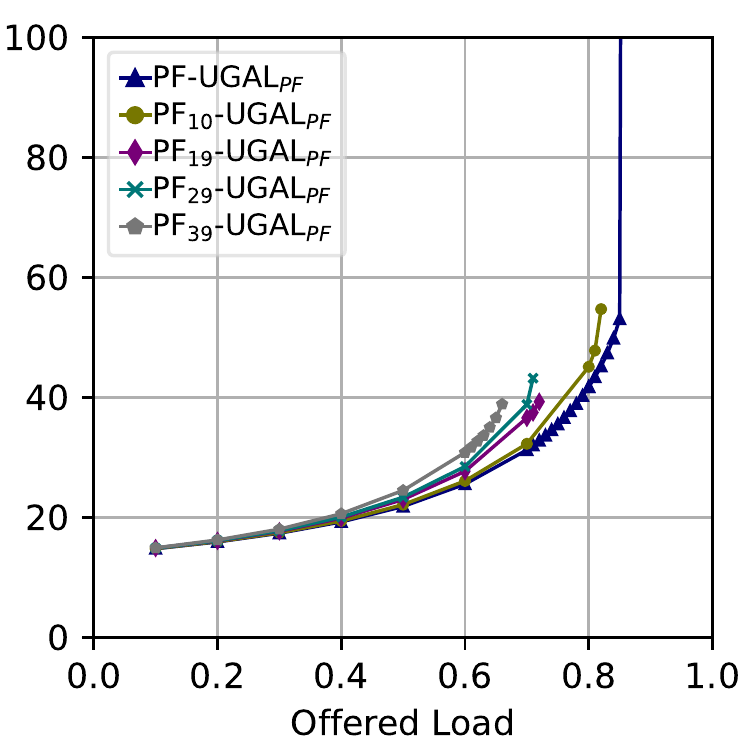}
%
% \vspaceSQ{-2.0em}
%
\caption{Replicating non-quadrics clusters.}
%\caption{Round robin replication.}
\label{fig:bitrev4}
\end{subfigure}
%
%\vspaceSQ{-1em}
\caption{Performance analysis of incrementally expanded \flyN{}.}
\label{fig:performance-analysis-expansion}
%\vspaceSQ{-1em}
\end{figure}

Figure~\ref{fig:performance-analysis-expansion} shows the latency and throughput of Polarfly incrementally expanded by
adding $3$, $6$, $9$ and 
$12$ clusters by quadric or non-quadric cluster replication, which corresponds to approximately $10\%$, $19\%$, 
$29\%$ and $39\%$ increase in network size, respectively.
The labels of incrementally expanded networks follow the scheme
\textit{$<$network$>$\textsubscript{$<$size increase in
percent$>$}-$<$routing$>$}.
% Specifically, we analyze uniform traffic on \fly after adding $3$, $6$, $9$ and 
% $12$ clusters through each method, which corresponds to approximately $10\%$, $19\%$, 
% $29\%$ and $39\%$ increase in system size, respectively. The number of endpoints per router are correspondingly increased to maintain a $1:2$ ratio with maximum network radix.
We observe that $39\%$ incremental growth in size using quadric replication results in a $31\%$ drop
in throughput.
Comparatively, non-quadric replication 
creates only $19\%$ drop in throughput for
an equivalent increase in network size,
thanks to its near-uniform degree distribution.
Moreover, after the first replication, subsequent non-quadric replications have
little impact on maximum throughput -- $73\%$ of peak bandwidth 
with $10\%$ incremental growth vs $67\%$ of peak
bandwidth with $39\%$ incremental growth.

\section{Structural Analysis}
We compare bisection bandwidth and link failure resilience of PolarFly, against the topologies given in Table~\mbox{\ref{tab:performance-config}}.

% \marginpar{\vspace{3em}\colorbox{yellow}{\textbf{R-4}}}

\subsection{Bisection bandwidth}
\begin{figure}[!ht]
\begin{centering}
\includegraphics[width=\columnwidth]{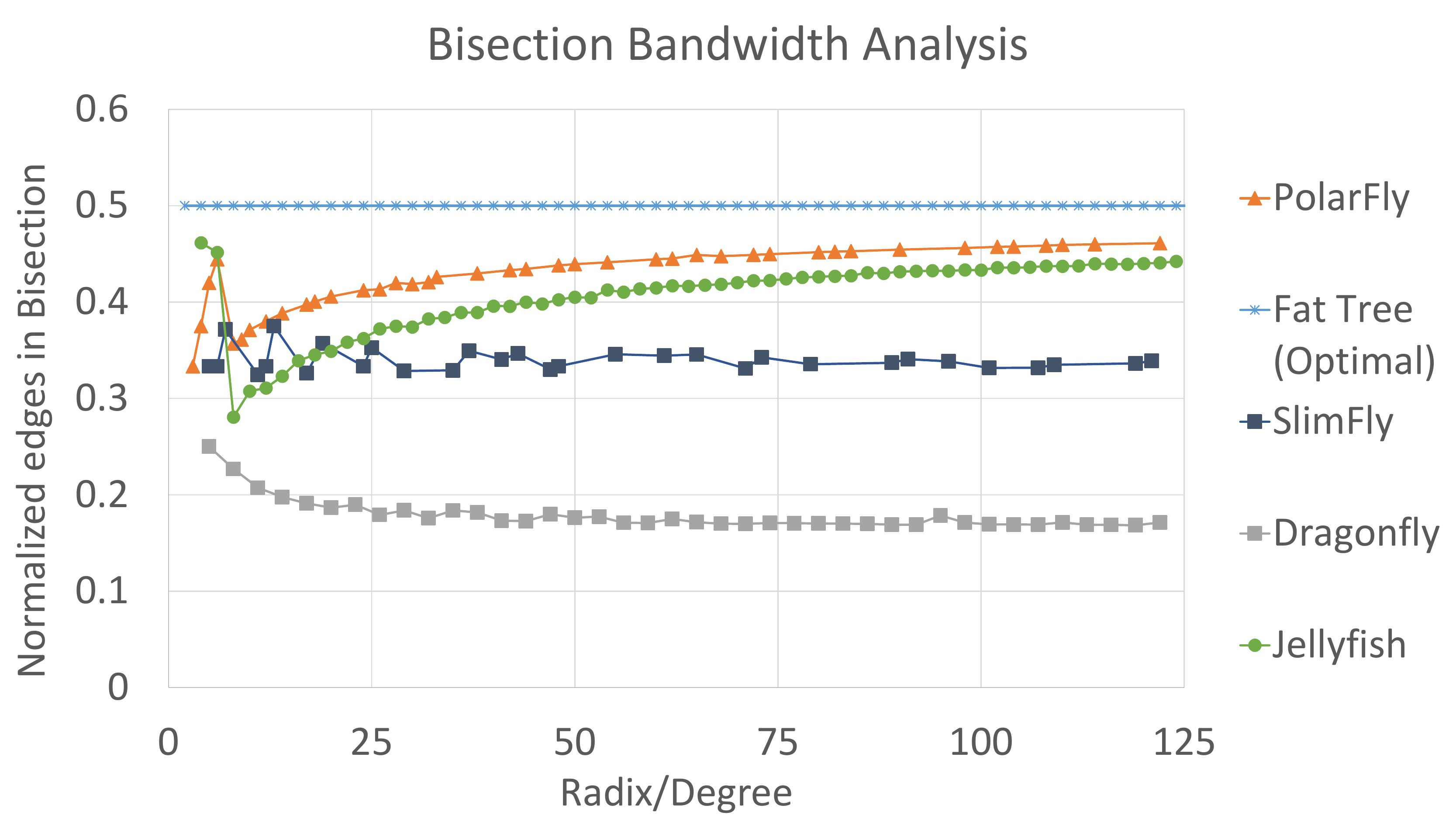}
\caption{Bisection bandwidth of different topologies shown by the number of links in the cut normalized to total links in the topology.}
\label{fig:bisection}
\end{centering}
\end{figure}

Figure~\ref{fig:bisection} shows the bisection bandwidth
of compared topologies in terms of the fraction of 
edges in the bisection cut-set computed
by METIS~\cite{metis}. Fat Trees provide optimal bisection bandwidth with
$50$\% edges lying in the cut-set.
PolarFly closely approximates the optimal 
ratio, reaching it asymptotically. For 
network radix $\geq 18$, PolarFly has more than 
$40\%$ links crossing the bisection, even 
surpassing random expander networks such as Jellyfish\cite{Singla:2012:JND:2228298.2228322}. 
This is not suprising since PolarFly topology expands 
extremely well, enforcing an almost Moore Bound spanning tree view from each vertex, whereas
Jellyfish relies on random distribution of 
links and only achieves $50\%$ of links in expectation for a random bisection.
PolarFly has significantly higher bisection
bandwidth compared to deterministic topologies SlimFly and Dragonfly,
that have only $33\%$ and $17\%$ 
links in bisection.
\begin{figure*}[t]
\centering
%\vspaceSQ{-2em}
\begin{subfigure}[t]{0.24 \textwidth}
\centering
\includegraphics[width=1.0\columnwidth]{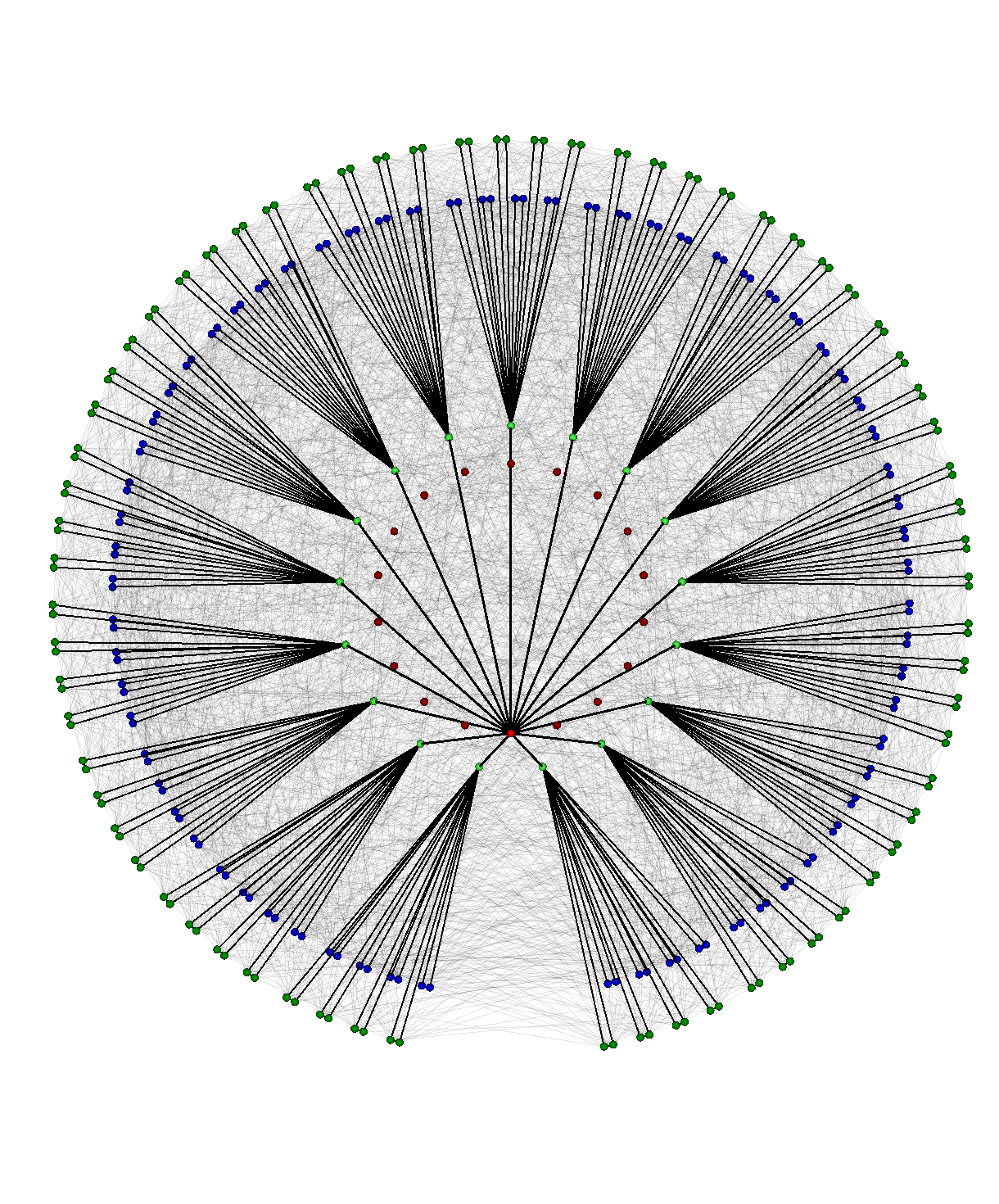}
%
%\vspaceSQ{-2.0em}
\caption{Overhead view of $ER_{17}$.}
\label{fig:overhead_17}
\end{subfigure}
%\quad
\begin{subfigure}[t]{0.24 \textwidth}
\centering
\includegraphics[width=1.0\columnwidth]{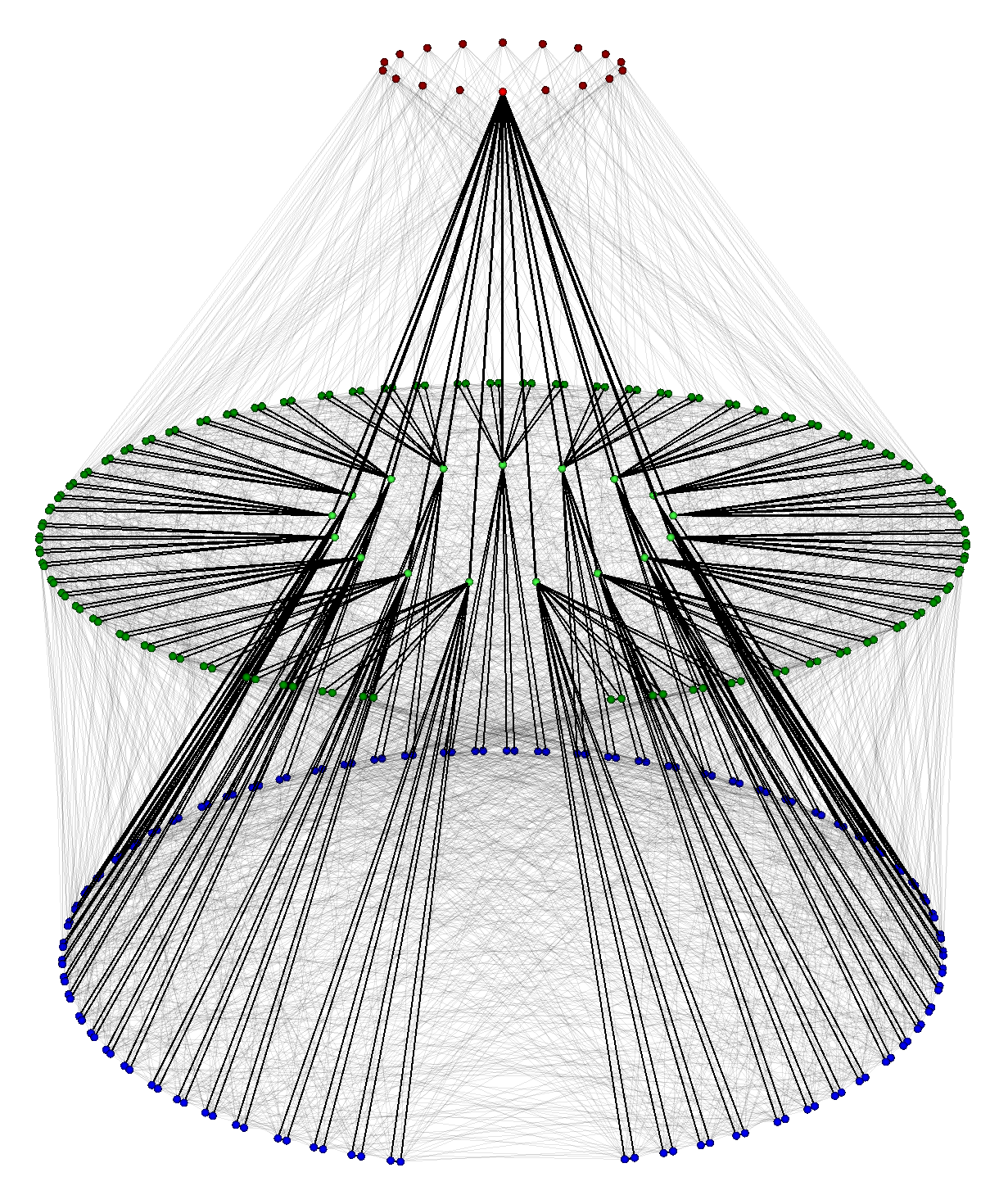}
%
%\vspaceSQ{-2.0em}
\caption{Side view of $ER_{17}$.}
\label{fig:side_17}
\end{subfigure}
\begin{subfigure}[t]{0.24 \textwidth}
\centering
\includegraphics[width=1.0\columnwidth]{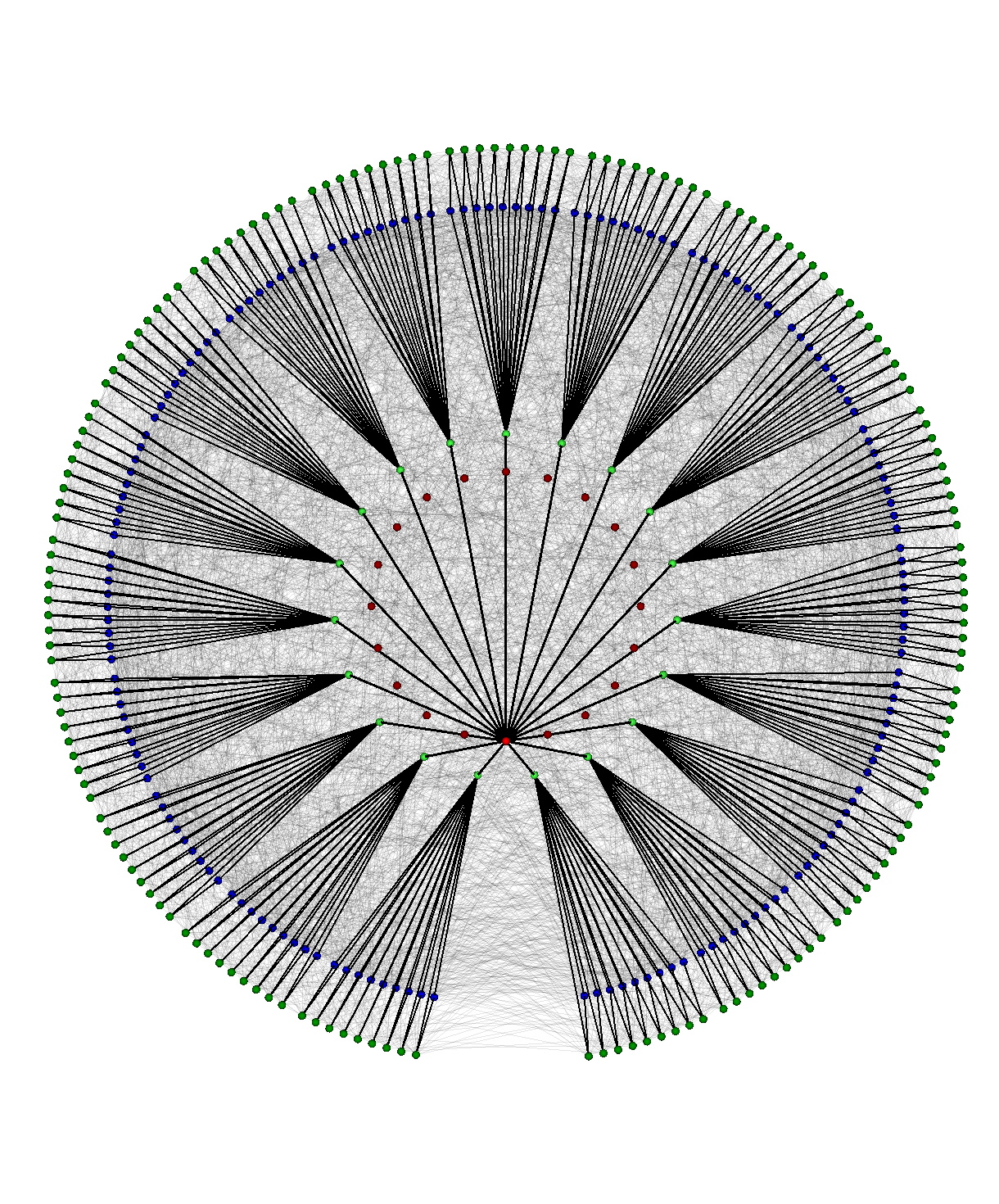}
%
%\vspaceSQ{-2.0em}
\caption{Overhead view of $ER_{19}$.}
\label{fig:overhead_19}
\end{subfigure}
\begin{subfigure}[t]{0.24 \textwidth}
\centering
\includegraphics[width=1.0\columnwidth]{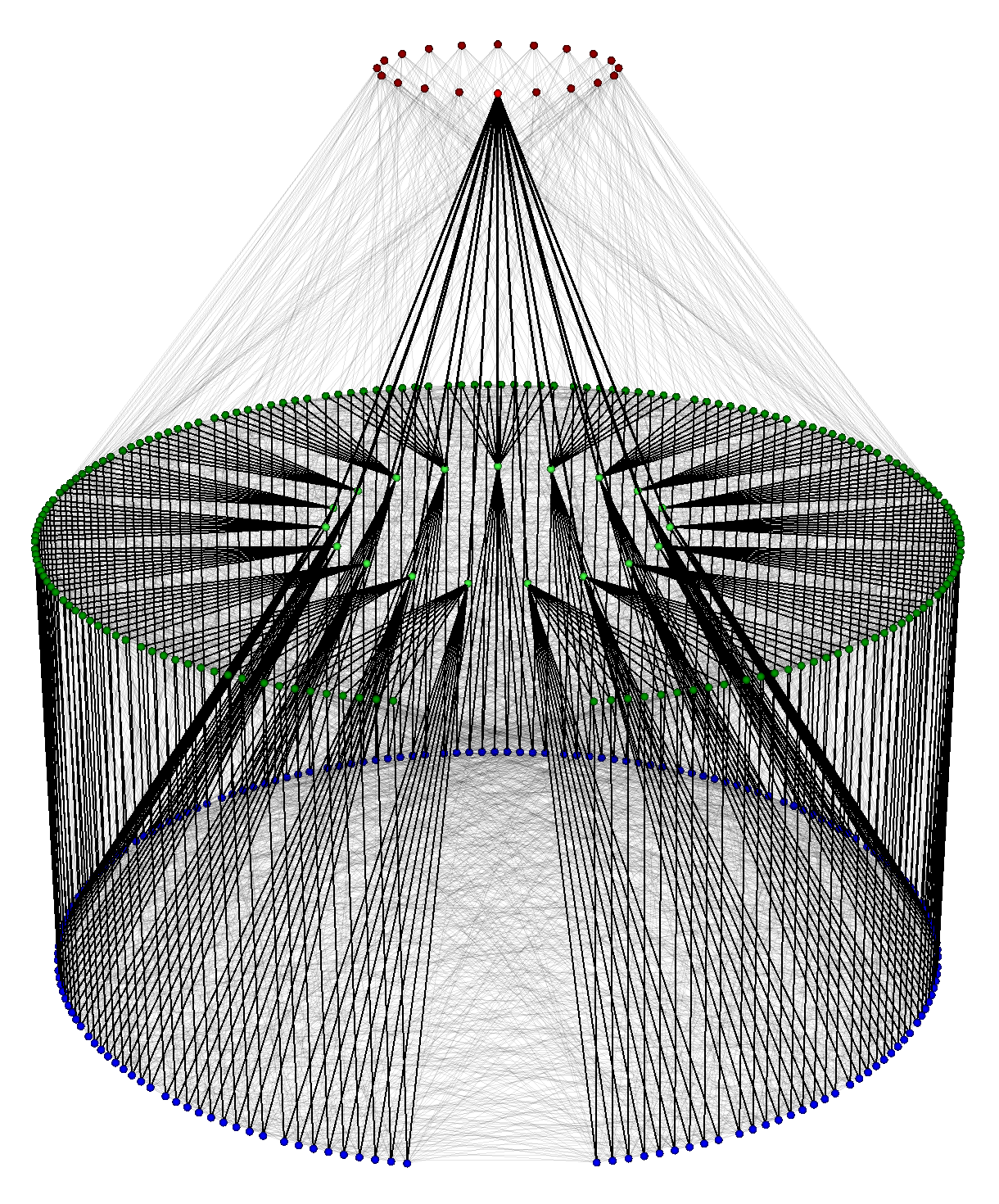}
%
%\vspaceSQ{-2.0em}
\caption{Side view of $ER_{19}$.}
\label{fig:side_19}
\end{subfigure}
%
%\vspaceSQ{-1em}
\caption{Graphs for $ER_{17}$ and $ER_{19}$ are compared. The edges from the starting quadric to the clusters  and the cluster edges are rendered in black. All other edges are rendered in light grey. Quadrics are in red, the centers are in light green, $V_1$ vertices are in green in the top layer, and $V_2$ vertices are in blue in the bottom layer. The triangle fan-outs internal to the clusters may be compared for $q =17$ and $q=19$: $17 \equiv  1 \mod 4$, and the triangle pairing of $V_1$ vertices with each other, and $V_2$ vertices with each other may be seen in subfigures~\ref{fig:overhead_17} and \ref{fig:side_17}, with no vertical edges within a cluster joining the $V_1$ upper layer with the $V_2$ lower layer. Likewise, $19 \equiv 3 \mod 4$, so the triangle pairing of $V_1$ vertices with $V_2$ vertices are seen in subfigures~\ref{fig:overhead_19} and \ref{fig:side_19}, in this case with vertical edges within a cluster joining the $V_1$ upper layer with the $V_2$ lower layer. }
\label{fig:comparecontrast_17_19}
%\vspaceSQ{-1em}
\end{figure*}

	\subsection{Fault Tolerance and Path Diversity}
	On the topology configurations given in Table~\ref{tab:performance-config}, we simulate $100$ random link failures until network disconnection, 
	and compute the median 
	disconnection ratio.\footnote{Mean and 
	Standard Deviation statistics cannot be used because if 
	any run disconnects at a particular failure ratio, its 
	diameter becomes infinite.} 
	We then randomly select a run with median
	disconnection ratio, and report its variation in
	network diameter and average shortest path length in 
	Figure~\ref{fig:resilience}.
	We also analyze path diversity in
	PolarFly in Table~\ref{tab:pathdiv}
	to better understand its behavior
	under link failures.
\begin{table}[htbp]
\setlength{\tabcolsep}{6.0pt}
\centering
\footnotesize
\resizebox{\linewidth}{!}{%
\begin{tabular}{lccc@{}}
\toprule
\textbf{Path length} & \textbf{Conditions} & \textbf{Number of paths}  %\textbf{Number of $w$ per $v$} 
\\
\midrule
$1$ & $v$, $w$ adjacent  & $1$ \\[-0.75em]
\\\hline 
\\[-0.75em]
$2$ & $v$, $w$ adjacent and one of $v$, $w$ quadric & $0$ \\
 & all other cases & $1$ 
\\[-0.75em]
\\\hline 
\\[-0.75em]
%$3$ &  $v$ quadric, $x$ not quadric & $q-1$ \\
$3$ & $v$, $w$ adjacent & $0$ \\
 &  $v,w$ not adjacent, $x$ not quadric & $q - 1$ \\
 &  $v,w$ not adjacent, $x$ quadric & $q$ 
 \\[-0.75em]
\\\hline 
%\\[-0.75em]
%$4$ & many cases, depending on: & \\
% & $1)$ adjacency of $v$ and $w$, and & $\mathcal{O}(q^2)$ \\
% & $2)$ vertex-type of $v$ and $w$\\
\\[-0.75em]
 $4$ & $v,w$ adjacent and neither of $v$, $w$ quadric &$(q-1)^2$\\ & $v,w$ adjacent and one of $v$, $w$ quadric &$q^2-q$\\ 
 & $v,w$ not adjacent and both of $v$, $w$ quadric &$q^2-q$\\ 
 & $v,w$ not adjacent, $v,w \in V_1$, $x$ not quadric &  $q^2-4$\\ 
 & $v,w$ not adjacent, $v$ quadric, $w \in V_1$ &  $q^2-3$\\ 
 & $v,w$ not adjacent, $v,w \in V_1$, $x$ quadric & $q^2-2$ \\ 
 & $v,w$ not adjacent, $v \in V_1$, $w \in V_2$ & $q^2-2$ \\ 
 & $v,w$ not adjacent, $v$ quadric, $w \in V_2$ &  $q^2-1$\\ 
 & $v,w$ not adjacent, $v \in V_2$, $w \in V_2$ &  $q^2$\\ 
\bottomrule
\end{tabular}}
\caption{Path diversity in $ER_q$ for small lengths: the number of paths of given lengths between arbitrary vertices $v$ and $w$, and the 
conditions under which such paths
exist. The vertex $x$ is the 
unique intermediate vertex 
between $v$ and $w$ (if it exists). There are many different cases for path length $4$; however, all are $\mathcal{O}(q^2)$. They are listed from smallest number of paths to largest.
}
%\vspaceSQ{-1em}
\label{tab:pathdiv}
\end{table}
\begin{figure}[ht]
\begin{centering}
\includegraphics[width=\columnwidth]{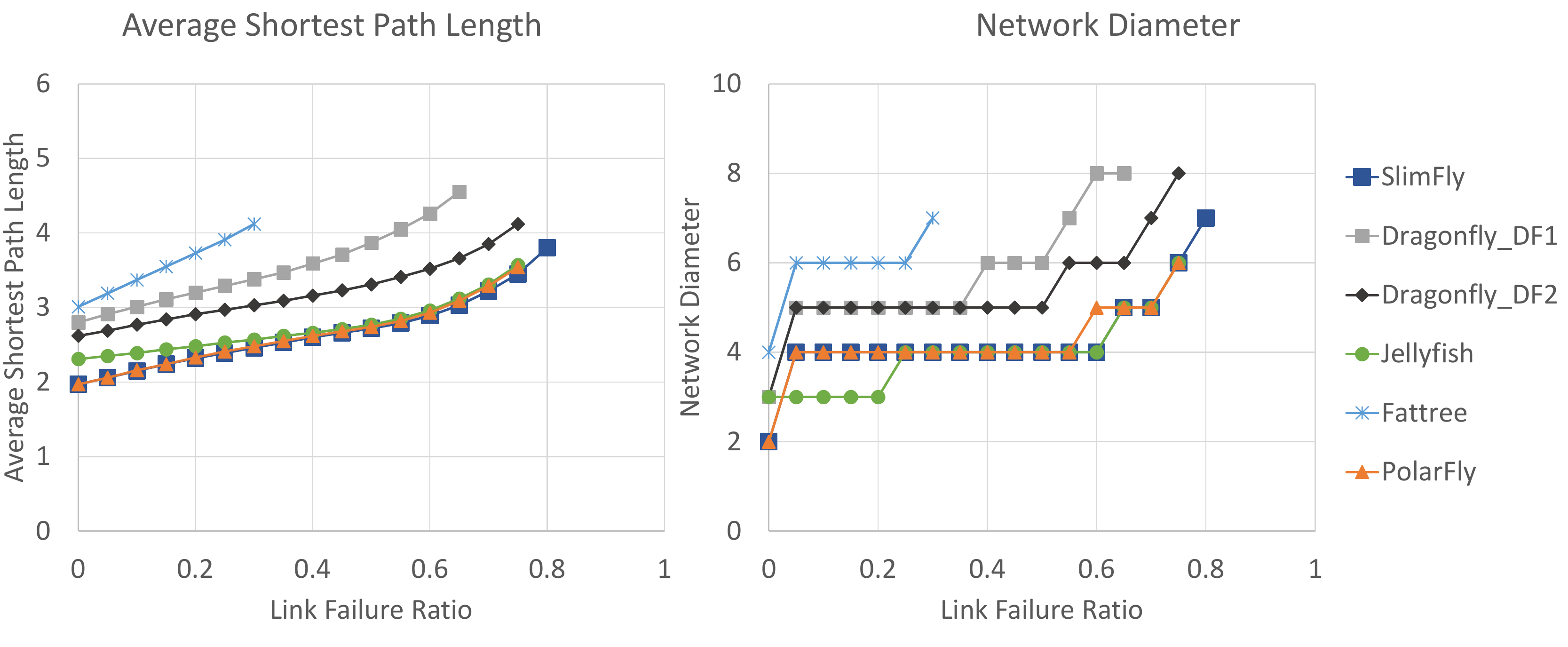}
\caption{Resilience properties of topologies as a function of the fraction of failed links.}
\label{fig:resilience}
\end{centering}
\end{figure}

	Jellyfish, being a random expander, is highly 
	resilient to link failures -- random failures in Jellyfish 
	result in just another random graph.
	PolarFly and SlimFly exhibit
	similar resilience, with higher
	disconnection ratio than 
	both Fat tree and balanced Dragonfly
	DF1.
	They are both expanders and have comparable
	resilience to Jellyfish. 
	However, being diameter-$2$ 
	networks with close to Moore bound scalability, the diameter of PolarFly 
	and SlimFly increases more rapidly 
	than Jellyfish. Compared to PolarFly, SlimFly has slightly more redundancy 
	in minimal paths, resulting in marginally
	higher disconnection ratio, even 
	though it reduces scalability.

	If a single link fails, the diameter of PolarFly increases to $3$, or $4$ if the link is from a quadric.
	Table~\ref{tab:pathdiv} shows that there are no $2$- or $3$-hop paths
	between quadrics and the adjacent vertices, which intuitively explains why PolarFly diameter increases to $4$ with only $5\%$ link
	failure, as in Figure~\ref{fig:resilience}. 
	However, PolarFly has a great deal of path
	diversity for path length 
	$4$, so its diameter
	stays at $4$ even when $55\%$ links fail.

	If a node $x$ fails, PolarFly diameter would increase from $2$ to $3$,
	as the $2$-hop minimal paths between
	neighbors of $x$ would be 
	lost. However, for any neighbor $v$ of $x$, the neighbors of $v$ 
	have $1$-hop or $2$-hop paths to other
	neighbors of $x$, that do not pass 
	through $x$. Hence, despite $x$ failing, $v$ can still reach other 
	nodes within $3$-hops.

\section{Cost Analysis}
We now analyze the cost of the network topologies under iso-injection bandwidth constraints. We fill focus on a specific case which is reflective of the latest technological developments: co-packaged Optical IO (OIO)~\cite{wade2020teraphy}. The primary cost indicator is the total number of optical IO ports: each port requires an OIO module, a laser, a connector and cables. Technological constraints limit to the number of OIO modules that can be co-packaged in a die due to shoreline limitations: the state of the art is 4 to 6 OIO modules per die, with 8 links per module. We consider configurations with approximately 1,024 nodes, with each topology having the same injection bandwidth. Given that not all the constructions have exactly that number of nodes, we normalize the number of links to a network configuration with 1,024 nodes. In addition, we also consider the achievable injection performance and we normalize the achievable performance, under two distinct scenarios: uniform and permutation traffic. While most networks reach comparable saturation points with uniform traffic, typically around 90$\%$, fat trees are almost insensitive to the type of permutation while direct topologies must resort to some type of misrouting, bringing their saturation points down to approximately 50$\%$. Both Polarfly and Slim Fly use 4 OIO modules with 32 links per node, while Dragonfly 6 OIO modules with 48 links. Fat trees use switches with 4 OIO modules and 32 links, and each of the 1,024 nodes has 2 OIOs with 16 injection links. Figure~\ref{fig:cost} shows the relative costs to PolarFly under the two traffic patterns. Slim Fly has a slight cost increase of about 20$\%$, reflective of the lower fraction of Moore's bound, while the Dragonfly is a diameter 3 network, so the ratio injection bandwidth to overall bandwidth is 1:3 vs 1:2 for PolarFly and Slim Fly. Due to packaging limitations, fat tree switches can only connect two input nodes with 16 links each, resulting in a rather deep 10-level construction of 512 switches per level, and 256 switches in the top level. PolarFly compares very favorably to fat trees with a 5.19X cost reduction under uniform traffic and 2.68X under permutation traffic.
\begin{figure}[htbp]
    \centering
    \includegraphics[width=\linewidth]{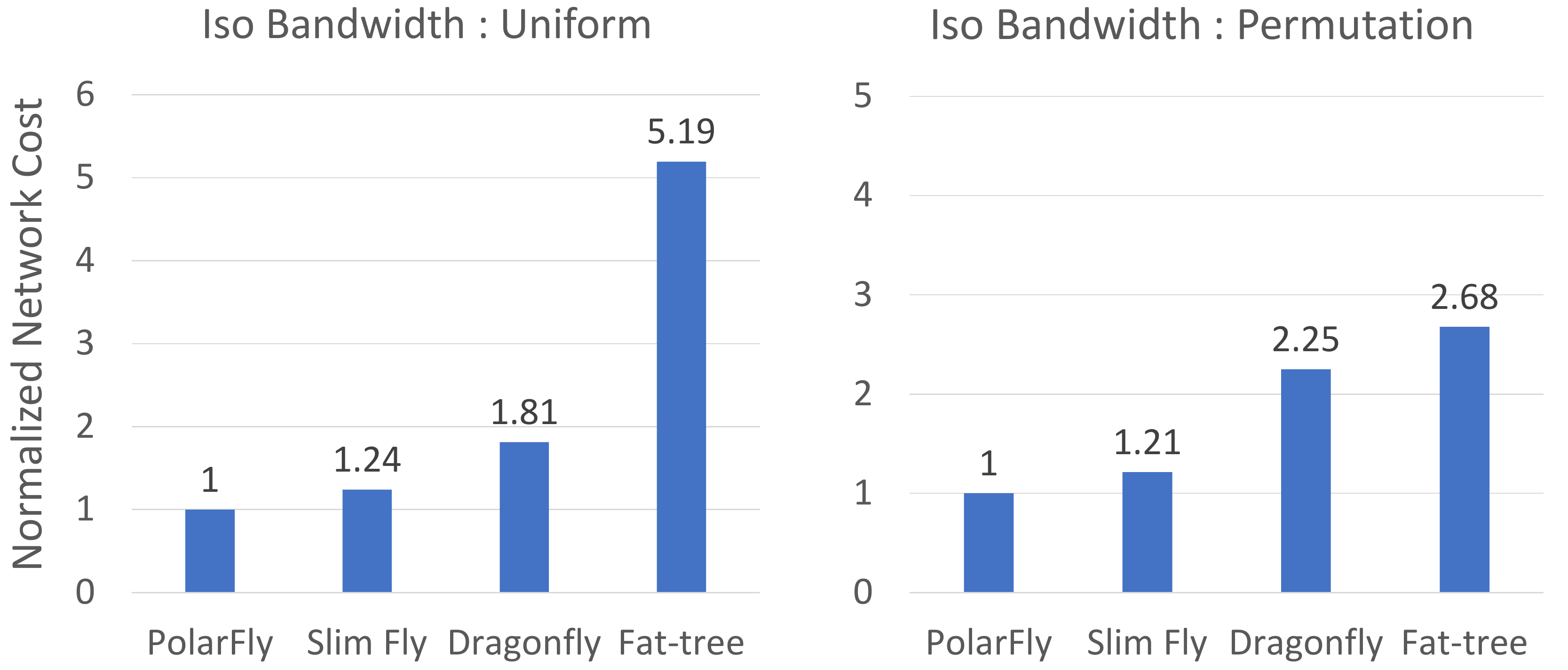}
    \caption{Cost per node under different topologies normalized to 1,024 nodes.}
    \label{fig:cost}
\end{figure}

\section{Related Work}
%\lmcomment{shouldn't this section be up in the Introduction?} \klcomment{I think this is fine here. In systems conferences, I haveseen related work at the end in many papers. For reader's context, table~\ref{tab:criteria} anyway summarizes previous topologies.}
Network topologies considered in this paper are described in detail in Section~\ref{sec:back}, more details are also provided in a recent survey~\cite{besta2020high}.
Early works into novel topologies with diameter lower than that of 3-stage Fat trees~\cite{Leiserson:1985:FUN:4492.4495} include Flattened Butterfly~\cite{dally07} and its generalization called HyperX~\cite{ahn2009hyperx}, and the Dragonfly topology~\cite{dally08, Arimilli:2010:PHI:1901617.1902282}. These designs mainly aimed at facilitating the physical layout of networks.
Lowering the diameter of a network in order to reduce cost and power consumption while maintaining high performance have been introduced in the Slim Fly class of interconnects~\cite{besta2014slim, besta2018slim}. 
Since then, several other designs followed, including Xpander~\cite{valadarsky2016xpander}, Megafly~\cite{flajslik2018megafly}, Bundlefly~\cite{bundlefly_2020} or Galaxyfly~\cite{lei2016galaxyfly}. However, they do not focus on diameter-2 and thus none of them improves upon key properties such as latency, cost, or power consumption.
PolarFly extends this line of work by exploiting a family of graphs that is asymptotically optimal with respect to the Moore Bound, allowing close to optimal scalability. It simultaneously offers superior cost, power consumption, and performance. Moreover, it specifically targets the recent developments into copackaged optics, something not addressed so far in the literature for scalable network design.

Routing in low-diameter networks has also been a subject of research, especially in recent years. For example, the FatPaths~\cite{besta2020fatpaths} routing architecture, enables adaptive multipathing in data-center and HPC clusters in
low-diameter networks, focusing on Slim Fly.
However, none of these works is particularly well suited for the unique structure of PolarFly in which some routers form intra-connected clusters while a single cluster of quadric routers forms an independent set. We address this with a novel adaptive UGAL routing protocol suited for PolarFly. 

The mathematical foundations of the Erd\H os-R\'enyi polarity graphs (ER) are embedded in projective geometry and 
were laid down in mid-$20^{\text{th}}$ century. Singer~\cite{singer1938theorem} first formulated  perfect difference sets -- 
a numerical structure that encodes
the incidence between lines and points of projective planes. Erd\H os and R\'enyi~\cite{erdosrenyi1962} discovered the polarity
quotient graph of this incidence structure, which forms
the basis of PolarFly. Indpendently, Brown~\cite{brown_1966} also constructed 
the same graph using orthogonality relationship of points in 
projective planes.

Building on these foundations, some prior works have proposed  an ER graph topology for high performance interconnection networks. Parhami et al.\cite{parhami2005perfect} use perfect difference sets to 
construct the bipartite network of same degree, diameter and order as the incidence graph described in section~\ref{sec:bipartite}. Brahme et al.\cite{brahme2013symsig} rediscover the ER graphs by defining
a symmetric adjacency equation on the perfect difference sets. They also compare the performance of certain 
communication primitives on this topology and the Clos network. Camarero et al.\cite{camarero} use the polarity-map based construction of ER graphs~(section~\ref{sec:polarity}) and compare the cost of conventional networks with various topologies.

% The SymSig topology~\cite{brahme2013symsig} constructed ER graphs using a using Singer's Perfect Difference sets~\cite{singer1938theorem}.
% The PolarFly topology presented in this paper, was also independently developed by~\cite{camarero} and~\cite{brahme2013symsig}, using ER polarity graphs.
%\kledit{Camarero et al.\cite{camarero} also proposed the use of Erd\H os-R\'enyi (ER) polarity graphs for HPC networks
%and describes the polarity based construction given in section~\ref{sec:formal}, with emphasis on 
%the low-diameter and scalability of ER graphs along with a cost effectiveness for a conventional network design.

% \lmdelete{Note that t}The ER graphs, \lmedit{with }their polarity based construction and diameter 
% properties have been known since mid $20^{\text{th}}$ 
% century~\cite{erdosrenyi1962, brown_1966}, much before~\cite{camarero},~\cite{brahme2013symsig} and this work.

To the best of our knowledge, PolarFly is the first work to comprehensively analyze
networking properties of ER graphs, covering several aspects beyond the prior 
attempts~\cite{camarero, brahme2013symsig}, including a comparison of feasible radixes, performance for various traffic
patterns, bisection width, resilience, and network cost under a co-packaged and iso-bandwidth setting. 
We also develop a novel modular layout,
incremental 
expansion strategies, and routing schemes to exploit 
non-minimal path diversity, all of which utilize new mathematical properties of the ER graphs that are presented for
the first time in this paper.
In this way, our work extends 
the feasibility of ER graphs as network topologies, well beyond
the existing literature. 
% Moreover, for several networking properties, including the layout, path diversity, expandability and routing, we derive and exploit previously unknown mathematical properties of the ER graphs.

\section{Conclusion}

% \kledit{Recent developments in co-packaged photonics and the ever-growing demand for more scalable interconnects pose significant challenges for network design.}
%
% To alleviate this, 
In this paper, we propose PolarFly, a diameter-$2$ network that asymptotically reaches the Moore upper bound on the number of nodes for a given degree and diameter.
\flyN{} improves upon Slim Fly, being more performant, scalable and cost-effective by up to 10\%.
Importantly, \flyN{} is flexible (it offers a wide range of feasible designs using manufacturable routers), modular (its structure can be decomposed into groups), and expandable (one can incrementally increase its size without much performance loss).
%
% We develop adaptive routing protocols for \fly{} that ensure very high bandwidth and low latency, outperforming Slim Fly and Dragonfly.
%
We expect that PolarFly will become the enabler for more energy-efficient interconnects in the next-generation era of co-packaged devices.

{
\section*{Acknowledgment}
The authors would like to thank Guillermo Pineda-Villavicencio for many insightful discussions on the nature of graphs approaching the Moore bound.
}
\bibliographystyle{IEEEtran}
%\bibliographystyle{abbrv}
%\balance
% \clearpage
\bibliography{references}

%\appendix
%\input{appendix.tex}

%\maciej{\ul{IDEAS for the network name:} PolarFly, LegoFly ClusterFly, ModuleFly, UnitFly, RackFly, GroupFly}

\end{document}